\documentclass[11pt]{article}
\usepackage{fullpage}
\usepackage{graphicx}
\usepackage{amssymb}
\usepackage{amsmath}

\usepackage{stmaryrd}
\usepackage{mathrsfs}

\usepackage{latexsym}
\usepackage{rotating}

\newenvironment{proof}{\noindent {\bf Proof:}}{\hfill$\Box$}
\newtheorem{theorem}{Theorem}[section]
\newtheorem{lemma}[theorem]{Lemma}
\newtheorem{recurrence}[theorem]{Recurrence}

\newtheorem{definition}[theorem]{Definition}
\newtheorem{remark}[theorem]{Remark}

\newcommand{\ignore}[1]{}
             
\newcommand{\hcm}[1][1]{\hspace*{#1 cm}}
\newcommand{\rb}[2]{\raisebox{#1 mm}[0mm][0mm]{#2}}
\newcommand{\istrut}[2][0]{\rule[- #1 mm]{0mm}{#1 mm}\rule{0mm}{#2 mm}}
\newcommand{\zero}[1]{\makebox[0mm][l]{$#1$}}

\newcommand{\paren}[1]{{\left( #1 \right)}}

\newcommand{\angbrack}[1]{\left< #1 \right>}
\newcommand{\Angbrack}[1]{\Big< #1 \Big>}

\newcommand{\Bigcurly}[1]{\Big\{ #1 \Big\}}

\newcommand{\sqbrack}[1]{\left[ #1 \right]}
\newcommand{\SqBrack}[1]{\Big[ #1 \Big]}

\newcommand{\Paren}[1]{\Big( #1 \Big)}

\newcommand{\ceil}[1]{\lceil #1 \rceil}
\newcommand{\floor}[1]{\lfloor #1 \rfloor}
\newcommand{\f}[2]{\frac{#1}{#2}}
\newcommand{\fr}[2]{\mbox{$\frac{#1}{#2}$}}

\newcommand{\poly}{\operatorname{poly}}

\newcommand{\Furedi}{F\"{u}redi}
\newcommand{\Szemeredi}{Szemer\'{e}di}

\newcommand{\Ex}{\operatorname{Ex}}

\newcommand{\DS}[1]{\lambda_{#1}}
\newcommand{\dblDS}[1]{\DS{{#1}}^{\scriptscriptstyle\inflate}}
\newcommand{\Feather}[1]{\Phi_{#1}}
\newcommand{\FFeather}{\Phi'}
\newcommand{\DblFeather}{\Phi''}

\newcommand*\xbar[1]{%
  \hbox{%
    \vbox{%
      \hrule height 0.5pt 
      \kern0.3ex
      \hbox{%
        \kern-0.1em
        \ensuremath{#1}%
        \kern-0.1em
      }%
    }%
  }%
}

\newcommand{\compose}{\operatorname{\circ}}

\newcommand{\shuffle}{\operatorname{\diamond}}

\newcommand{\subseq}{\prec}
\newcommand{\nsubseq}{\nprec}

\newcommand{\subseqe}{\;\xbar{\subseq}\;}

\newcommand{\inflate}{\operatorname{dbl}}

\newcommand{\GSigma}{\hat{\Sigma}}
\newcommand{\LSigma}{\check{\Sigma}}
\newcommand{\Gm}{\hat{m}}
\newcommand{\Lm}{m}
\newcommand{\Gn}{\hat{n}}
\newcommand{\Ln}{\check{n}}
\newcommand{\GfSigma}{\acute{\Sigma}}
\newcommand{\GlSigma}{\grave{\Sigma}}
\newcommand{\GmSigma}{\bar{\Sigma}}
\newcommand{\Gfn}{\acute{n}}
\newcommand{\Gln}{\grave{n}}
\newcommand{\Gmn}{\bar{n}}
\newcommand{\GS}{\hat{S}}
\newcommand{\GfS}{\acute{S}}
\newcommand{\GlS}{\grave{S}}
\newcommand{\GmS}{\bar{S}}
\newcommand{\GfeatherS}{\tilde{S}}
\newcommand{\GlnonfeatherS}{\dot{S}}
\newcommand{\GrnonfeatherS}{\ddot{S}}
\newcommand{\GnonfeatherS}{\breve{S}}
\newcommand{\LS}{\check{S}}

\newcommand{\Tree}{\mathcal{T}}
\newcommand{\GTree}{\hat{\Tree}}
\newcommand{\LTree}{\check{\Tree}}
\newcommand{\FirstTree}{\acute{\mathcal{T}}}
\newcommand{\LastTree}{\grave{\mathcal{T}}}
\newcommand{\Ensemble}{\mathcal{E}}
\newcommand{\block}{\mathcal{B}}
\newcommand{\head}{\operatorname{he}}
\newcommand{\otherhead}{\overline{\head}}
\newcommand{\lefthead}{\operatorname{lhe}}
\newcommand{\righthead}{\operatorname{rhe}}
\newcommand{\crown}{\operatorname{cr}}
\newcommand{\quill}{\operatorname{qu}}
\newcommand{\wing}{\operatorname{wi}}
\newcommand{\tip}{\operatorname{wt}}
\newcommand{\othertip}{\overline{\tip}}
\newcommand{\lefttip}{\operatorname{lwt}}
\newcommand{\righttip}{\operatorname{rwt}}
\newcommand{\feather}{\operatorname{fe}}

\newcommand{\M}[2]{\mu_{#1,#2}}

\newcommand{\N}[2]{\nu_{#1,#2}}

\newcommand{\bl}[1]{\llbracket #1 \rrbracket}

\newcommand{\Sbot}{S_{\operatorname{bot}}}
\newcommand{\Stop}{S_{\operatorname{top}}}
\newcommand{\Smid}{S_{\operatorname{mid}}}
\newcommand{\Ssub}{S_{\operatorname{sub}}}
\newcommand{\Ssh}{S_{\operatorname{sh}}}

\newcommand{\Perm}[1]{\operatorname{Perm}(#1)}

\newcommand{\PERM}[1]{\Lambda_{#1}}

\title{Sharp Bounds on Davenport-Schinzel Sequences of Every Order\thanks{This work is supported by NSF CAREER grant no. CCF-0746673,
NSF grant no. CCF-1217338,
and a grant from the US-Israel Binational Science Foundation.}}

\author{Seth Pettie\\ University of Michigan}

\begin{document}

\maketitle

\begin{abstract}
One of the longest-standing open problems in computational geometry 
is to bound the lower envelope of $n$ univariate functions, each pair of which
crosses at most $s$ times, for some fixed $s$.  This problem is known to be equivalent
to bounding the length of an order-$s$ Davenport-Schinzel sequence, namely
a sequence over an $n$-letter alphabet that avoids alternating subsequences of the form 
$a \cdots b \cdots a \cdots b \cdots$ with length $s+2$.
These sequences were introduced by Davenport and Schinzel in 1965
to model a certain problem in differential equations and have since 
been applied to bounding the running times of geometric algorithms, data structures,
and the combinatorial complexity of geometric arrangements.

Let $\DS{s}(n)$ be the maximum length of an order-$s$ DS sequence over $n$ letters.
What is $\DS{s}$ asymptotically?
This question has been answered satisfactorily (by Hart and Sharir, Agarwal, Sharir, and Shor, Klazar, and Nivasch)
when $s$ is even or $s\le 3$.
However, since the work of Agarwal, Sharir, and Shor in the mid-1980s there
has been a persistent gap in our understanding of the odd orders.

In this work we effectively close the problem by establishing 
sharp bounds on Davenport-Schinzel sequences of {\em every} order $s$.
Our results reveal that, contrary to one's intuition, $\DS{s}(n)$ behaves essentially like $\DS{s-1}(n)$ when $s$ is odd.
This refutes conjectures due to Alon et al. (JACM, 2008) and Nivasch (JACM, 2010).

\ignore{
Together with previous results of Hart--Sharir, Agarwal--Sharir--Shor, Klazar, and Nivasch,
we prove the extremal function of order-$s$ sequences is:
\[
\DS{s}(n) = 
\left\{
\begin{array}{l@{\hcm}l@{\istrut[3]{0}}}
n								& s=1\\
2n-1								& s=2\\
(2\pm o(1))n\alpha(n)					& s=3\\
\Theta(n2^{\alpha(n)})				& s=4\\
\Omega(n\alpha(n)2^{\alpha(n)})	 \mbox{ and } O(n\alpha^2(n)2^{\alpha(n)})	& s=5\\
n\cdot 2^{(1\pm o(1))\alpha^t(n)/t!}			& s\ge 6, \; t = \floor{\frac{s-2}{2}}
\end{array}
\right.
\]
}
\end{abstract}

\section{Introduction}

Consider the  problem of bounding the complexity of the lower envelope of $n$ continuous univariate functions $f_1,\ldots,f_n$,
each pair of which cross at most $s$ times.
In other words, how many maximal connected intervals of the $\{f_i\}$ are contained in the graph of the function $f_{\min}(x) = \min\{f_1(x),\ldots,f_n(x)\}$?
In the absence of any constraints on $\{f_i\}$ this problem can be completely stripped of its geometry by {\em transcribing} the lower envelope 
$f_{\min}$ as a {\em Davenport-Schinzel (DS) sequence} of order $s$, namely, a repetition-free sequence over the alphabet $\{1,\ldots,n\}$
that does not contain any alternating subsequences of the form $\cdots a\cdots b \cdots a \cdots b \cdots$ with length 
$s+2$, for any $a,b\in\{1,\ldots,n\}$.\footnote{If the sequence corresponding to the lower envelope contained
an alternating subsequence $abab\cdots$ with length $s+2$ 
then the functions $f_a$ and $f_b$ must have crossed at least $s+1$ times, a contradiction.}
Although Davenport and Schinzel~\cite{DS65} introduced this problem nearly 50 years ago,
it was only in the early 1980s that DS sequences became well known in the computational geometry community~\cite{Atallah85,SharirCKLPS86}.
Since then DS sequences have found a startling
number of geometric applications, with a growing
number~\cite{SalvoP07,LopezPR12,AlstrupLST97,BurkardD01,K92,Pettie-Deque-08} 
that are not overtly geometric.\footnote{To cite a
fraction of the literature, DS sequences/lower envelopes are routinely applied to 
problems related to geometric arrangements
\cite{MatousekPSSW94,AronovD11,AronovB12,EdelsbrunnerPSS87,Pettie-SoCG11,deBerg10,DumitrescuST09,BernEPY91,HuttenlocherK90,HuttenlocherKS93,Koltun04,Har-Peled00,Efrat05,KoltunS03, Tagansky96,AgarwalS02,JaromczykK89},
in kinetic data structures and dynamic geometric algorithms
\cite{AlbersGMR98,GuibasMR92,KaplanRS11,AbamBG10,AgarwalKS08,Atallah85,HuttenlocherKK92,WahidKH10},
in visibility 
\cite{ColeS1989,SharirCKLPS86,MoetKvK08}, 
motion planning \cite{SharirCKLPS86,LevenS87},
and geometric containment problems~\cite{AonumaIIT90,SharirCKLPS86,SharirT94,AugustinePR10},
as well as variations on classical problems such as 
computing shortest paths~\cite{BaeO12,AronovBBJJKLLSS11,BaltsanS88} 
and convex hulls~\cite{EzraM12,AurenhammerJ12}.
They have also been used in some industrial applications~\cite{IlushinEHWK05, BerrettyGOvdS01}.
Refer to Sharir and Agarwal~\cite{AgarwalSharir95} for a 
survey of DS sequences and their early applications in computational geometry
and to Klazar~\cite{Klazar02} for a survey of DS sequences and related problems in extremal combinatorics.
}
In each of these applications some quantity (e.g., running time, combinatorial complexity)
is expressed in terms of $\DS{s}(n)$, the maximum length of an order-$s$ DS sequence over an $n$-letter alphabet.
To improve bounds on $\DS{s}$ is, therefore, to improve
our understanding of numerous problems in algorithms, data structures, and discrete geometry.\\

Davenport and Schinzel~\cite{DS65} established $n^{1+o(1)}$ upper bounds on $\DS{s}(n)$ for 
every order $s$.  In order to properly survey the improvements
that followed~\cite{Davenport70,Szemeredi73,HS86,Sharir87,Sharir88,Komjath88,ASS89,Klazar99,Nivasch10}
we must define some notation for forbidden sequences and their extremal functions.

\subsection{Sequence Notation and Terminology}\label{sect:notation-and-terminology}

Let $|\sigma|$ be the length of a sequence $\sigma = (\sigma(i))_{1\le i\le |\sigma|}$ and let $\|\sigma\|$ be the size
of its alphabet $\Sigma(\sigma) = \{\sigma(i)\}$.
Two equal length sequences are {\em isomorphic} if they are the same up to a renaming of their alphabets.
We say $\sigma$ is a {\em subsequence} of $\sigma'$, written $\sigma\subseqe \sigma'$, 
if $\sigma$ can be obtained by deleting symbols from $\sigma'$.
The predicate $\sigma \subseq \sigma'$ asserts that $\sigma$ is isomorphic to a subsequence of $\sigma'$.
If $\sigma\nsubseq\sigma'$ we say $\sigma'$ is {\em $\sigma$-free}.
If $P$ is a {\em set} of sequences, $P\nsubseq \sigma'$ holds if $\sigma\nsubseq \sigma'$ for {\em every} $\sigma\in P$.
The assertion that $\sigma$ {\em appears in} or {\em occurs in} or {\em is contained in} $\sigma'$ 
means either $\sigma\subseq \sigma'$ or $\sigma\subseqe \sigma'$, which one being clear from context.
The {\em projection} of a sequence $\sigma$ onto $G\subseteq \Sigma(\sigma)$ is obtained by deleting
all non-$G$ symbols from $\sigma$.
A sequence $\sigma$ is {\em $k$-sparse} if whenever $\sigma(i) = \sigma(j)$ and $i\neq j$, then $|i-j| \ge k$.
A {\em block} is a sequence of distinct symbols.  If $\sigma$
is understood to be partitioned into a sequence of blocks, $\bl{\sigma}$ is the number of blocks.
The predicate 
$\bl{\sigma}=m$ asserts that $\sigma$ can be partitioned into at most $m$ blocks.
The extremal functions for {\em generalized} Davenport-Schinzel sequences are defined to be
\begin{align*}
\Ex(\sigma,n,m) &= \max\{|S| \;:\; \sigma\nsubseq S, \; \|S\|=n, \mbox{ and } \bl{S}\le m\}\\
\Ex(\sigma,n)	&= \max\{|S| \;:\;  \sigma\nsubseq S, \; \|S\|=n, \mbox{ and } S \mbox{ is } \|\sigma\|\mbox{-sparse}\}\\
\intertext{where $\sigma$ may be a single sequence or a set of sequences.  
The conditions ``$\bl{S}\le m$'' and ``$S$ is $\|\sigma\|$-sparse'' guarantee that the extremal functions are finite.
For example, if $\|\sigma\|=2$ the sparsity criterion forbids immediate repetitions and 
such infinite degenerate sequences as $aaaaa\cdots$.
Blocked sequences, on the other hand, have no sparsity criterion.
The extremal functions for (standard) Davenport-Schinzel sequences are defined to be}
\DS{s}(n,m) &= \Ex(\overbrace{ababa\cdots}^{\mbox{\scriptsize length $s+2$}},n,m)
\hcm[.4]
\mbox{ and }
\hcm[.4]
\DS{s}(n) = \Ex(\overbrace{ababa\cdots}^{\mbox{\scriptsize length $s+2$}},n)
\end{align*}

Bounds on generalized Davenport-Schinzel sequences are expressed as a function of the
inverse-Ackermann function $\alpha$, yet there is no universally
agreed-upon definition of Ackermann's function or its inverse.  All definitions in the literature
differ by at most a constant, which usually obviates the need for more specificity.
In this article we use the following definition of Ackermann's function.
\begin{align*}
a_{1,j} &\;\zero{= 2^j}		& j\ge 1\\
a_{i,1} &\;\zero{= 2}		& i\ge 2\\	
a_{i,j}  &\;\zero{= w\cdot a_{i-1,w}}		& i,j\ge 2\\
& \hcm[.5]\mbox{where $w=a_{i,j-1}$}
\intertext{Note that in the table of $\{a_{i,j}\}$ values, the first column is constant
($a_{i,1} = 2$) and the second merely exponential ($a_{i,2} = 2^i$), so we have 
to look to the third column to find Ackermann-type growth.  We define the double- and single-argument versions of the inverse-Ackermann function to be}
\alpha(n,m) &= \zero{\min\{i \;|\; a_{i,j} \ge m,\, \mbox{ where } j = \max\{\ceil{n/m},3\}\}}\\
\alpha(n) &= \alpha(n,n)
\end{align*}
We could have defined $\alpha(n,m)$ without direct reference to Ackermann's function.
Note that $j = \log(a_{1,j})$.  One may convince oneself that 
$j = \log^\star(a_{2,j}) - O(1)$, $j = \log^{\star\star}(a_{3,j})-O(1)$, and in general,
that $j = \log^{[i-1]}(a_{i,j}) - O(1)$, 
where $[i-1]$ is short for $i-1$ $\star$s.\footnote{If $f : \mathbb{N}\backslash\{0\}\rightarrow \mathbb{N}$
is a decreasing function, $f^\star(m)$ is, by definition, $\min\{\ell \;|\; f^{(\ell)}(m) \le 1\}$, where
$f^{(0)}(m)=m$ and $f^{(\ell)}(m) = f(f^{(\ell-1)}(m))$.}
Thus, up to $O(1)$ differences $\alpha(n,m)$ could be defined as 
$\min\left\{i \;\left|\; \log^{[i-1]}(m) \leq \max\{\ceil{n/m},3\}\right.\right\}$.
We state previous results in terms of the single argument version of $\alpha$.  However, they all generalize
to the two-argument version by replacing $\DS{s}(n)$ with $\DS{s}(n,m)$ and $\alpha(n)$ with $\alpha(n,m)$.

\subsection{A Brief History of $\DS{s}$}

After introducing the problem in 1965, Davenport and Schinzel~\cite{DS65} proved that $\DS{1}(n) = n, \DS{2}(n)=2n-1,
\DS{3}(n) = O(n\log n)$, and for all $s\ge 4$, that $\DS{s}(n) = n\cdot 2^{O(\sqrt{\log n})}$, where the leading constant in the exponent depends on $s$.
Shortly thereafter Davenport~\cite{Davenport70} improved the bound on $\DS{3}(n)$ to $O(n\log n/\log\log n)$.
In 1973 \Szemeredi{}~\cite{Szemeredi73} dramatically improved the upper bounds for all $s\ge 3$, showing that
$\DS{s}(n) = O(n\log^\star n)$, where the leading constant depends on $s$.

From a purely numerical perspective \Szemeredi's bound settled the problem for all values of $n$ one might encounter in nature
(the log-star function being at most 5 for $n$ less than $10^{19,000}$),
so why should any thoughtful mathematician continue to work on the problem? 
In our view, the problem of quantitatively estimating  $\DS{s}(n)$ 
has always been a proxy for several qualitative questions: is $\DS{s}(n)$ linear or nonlinear?
what is the structure of extremal sequences realizing $\DS{s}(n)$?
and does it even matter what $s$ is?
In 1984 Hart and Sharir~\cite{HS86} answered the first two questions for order-3 DS sequences.  They gave a bijection
between order-$3$ ($ababa$-free) DS sequences and so-called generalized 
postorder path compression schemes.  Although these schemes resembled the path compressions
found in set-union data structures, Tarjan's analysis~\cite{Tar75} did not imply any non-trivial upper or lower bounds on their length.
Hart and Sharir proved that such path compression schemes have length $\Theta(n\alpha(n))$,
thereby settling the asymptotics of $\DS{3}(n)$.  

In 1989 Agarwal, Sharir, and Shor~\cite{ASS89} (improving on~\cite{Sharir87,Sharir88})
gave asymptotically tight bounds on order-$4$ DS sequences
and reasonably tight bounds on  higher order sequences.

\begin{align*}
\DS{4}(n) &= \Theta(n\cdot 2^{\alpha(n)})\\
\istrut{7}
\DS{s}(n) & \; \left\{
\begin{array}{l}
> n\cdot \zero{2^{(1+o(1))\alpha^t(n)\,/\, t!}}		\\
< n\cdot \zero{2^{(1+o(1))\alpha^t(n)}	}		\istrut{6}\\
\end{array}
\right. 
& \rb{0}{\mbox{for even $s \ge 6, t=\floor{\f{s-2}{2}}$.}}\\
\istrut{7.5}\DS{s}(n) & \; \left\{
\begin{array}{l}
> \DS{s-1}(n)\\
< n\cdot \zero{(\alpha(n))^{(1+o(1))\alpha^t(n)}}	\istrut{6}\\
\end{array}
\right. 
& \rb{0}{\mbox{for odd $s \ge 5$, $t=\floor{\f{s-2}{2}}$.}}
\end{align*}

For even $s$ their bounds were tight up to the constant in the exponent: $1$ for the upper bound
and $1/t!$ for the lower bound.  Moreover, their lower bound construction gave
a qualitatively satisfying answer to the question of how extremal sequences are structured when $s$ is even.
For odd $s$ the gap between upper and lower bounds was wider, the base of the exponent being 2
at the lower bound and $\alpha(n)$ at the upper bound.

\begin{remark}\label{remark:Ackermann-invariance}
The results of Agarwal, Sharir, and Shor~\cite{ASS89} force us to confront another question, namely,
when is it safe to declare victory and call the problem closed?  As Nivasch~\cite[\S8]{Nivasch10}
observed, the ``$+o(1)$'' in the exponent {\em necessarily} hides a $\pm\, \Omega(\alpha^{t-1}(n))$ term
if we express the bound in an ``{\em Ackermann-invariant}'' fashion, that is,
in terms of the generic $\alpha(n)$, without specifying the precise
variant of Ackermann's function for which it is the inverse.
Furthermore, under any of the definitions in the literature
$\alpha(n)$ is an integer-valued function whereas $\DS{s}(n)/n$ must increase fairly smoothly with $n$,
that is, an estimate of $\DS{s}(n)$ that is expressed as a function of 
\underline{any} integer-valued $\alpha(n)$
must be off by at least a $2^{\Omega(\alpha^{t-1}(n))}$ factor.
A reasonable definition of {\em sharp bound} (when dealing with generalized Davenport-Schinzel sequences)
is an expression that cannot be improved, given $\pm\, \Theta(1)$ uncertainty in the definition of 
$\alpha(n)$. For example, $\DS{4}(n) = \Theta(n2^{\alpha(n)})$ is sharp in this sense
since the constant hidden by $\Theta$ reflects this uncertainty.
In contrast, $\DS{3}(n) = \Theta(n\alpha(n))$ is not sharp in an Ackermann-invariant sense.
See the tighter bounds on $\DS{3}(n)$ cited below and in Theorem~\ref{thm:main}.
\end{remark}

In 2009
Nivasch~\cite{Nivasch10} 
presented a superior method for upper bounding $\DS{s}(n)$.
In addition, he provided a new construction of order-3 DS sequences that
matched an earlier upper bound of Klazar~\cite{Klazar99} up to the leading constant.

\begin{align}
\DS{s}(n) =& \,\left\{
\begin{array}{l@{\hcm[1.65]}l}
2n\alpha(n)+O(n\sqrt{\alpha(n)})		& \mbox{for $s=3$; upper bound is due to~\cite{Klazar99}.}\\
\Theta(n\cdot 2^{\alpha(n)})			& \istrut[3]{5}\mbox{for $s=4$.}\\
n\cdot 2^{(1 + o(1))\alpha^t(n)\,/\, t!}	& \mbox{for even $s \ge 6,$ $t=\floor{\f{s-2}{2}}$.}
\end{array}
\right.\nonumber\\
\istrut{9}
\DS{s}(n) & \; \left\{
\begin{array}{l@{\hcm[1]}l}
> \DS{s-1}(n)\\
< n\cdot (\alpha(n))^{(1+o(1))\alpha^t(n)\,/\, t!}	\istrut{6} & \rb{2}{\mbox{for odd $s \ge 5$, $t=\floor{\f{s-2}{2}}$.}}\\
\end{array}
\right. 
\tag{Niv}
\label{eqn:Nivasch}
\end{align}

This closed the problem for even $s\ge 6$ (the leading constant in the exponent being precisely $1/t!$)
but left the odd case open.
Alon, Kaplan, Nivasch, Sharir, and Smorodinsky~\cite{AlonKNSS08-J,Nivasch10} conjectured
that the upper bounds (\ref{eqn:Nivasch}) for odd orders are tight, that is, the base of the exponent is, in fact, $\alpha(n)$.
This conjecture was spurred by their discovery of similar functions that arose in an apparently unrelated combinatorial 
problem, {\em stabbing interval chains with $j$-tuples}~\cite{AlonKNSS08-J}.  

\subsection{New Results}

We provide new upper and lower bounds on the length of Davenport-Schinzel sequences
and in the process refute conjectures due to
Alon et al.~\cite[\S 5]{AlonKNSS08-J}, Nivasch~\cite[\S 8]{Nivasch10}, and Pettie~\cite[\S 7]{Pettie-GenDS11}.

\begin{theorem}\label{thm:main}
Let $\DS{s}(n)$ be the maximum length of a repetition-free sequence over an $n$-letter alphabet
avoiding subsequences isomorphic to $abab\cdots$ (length $s+2$), or, equivalently, the maximum
complexity of the lower envelope of $n$ continuous univariate functions, each pair of which coincide at most $s$ times.
For any $s\ge 1$, $\DS{s}$ satisfies:
\[
\DS{s}(n) = \left\{
\begin{array}{l@{\hcm}l@{\istrut[3]{0}}}
n								& s=1\\
2n-1								& s=2\\
2n\alpha(n) + O(n)					& s=3\\
\Theta(n2^{\alpha(n)})				& s=4\\
\Theta(n\alpha(n)2^{\alpha(n)})			& s=5\\
n\cdot 2^{(1 + o(1))\alpha^t(n)/t!}			& \mbox{both even and odd } s\ge 6, \; t = \floor{\frac{s-2}{2}}.
\end{array}
\right.
\]
\end{theorem}

Theorem~\ref{thm:main} is optimal in that it provides
the tightest bounds
that can be expressed in an Ackermann-invaraint fashion
(see Remark~\ref{remark:Ackermann-invariance}), 
and in this sense closes the Davenport-Schinzel problem.\footnote{The exponent $(1+o(1))\alpha^t(n)/t!$ is
the Ackermann-invariant expression $\alpha^t(n)/t! + O(\alpha^{t-1}(n))$.}
However, we believe our primary contributions are not the tight asymptotic bounds {\em per se}
but the structural differences they reveal between even and odd $s$.
We can now give a cogent explanation for {\em why} odd orders $s\ge 5$ behave essentially like
the preceding even orders and yet why they are intrinsically more difficult to understand.

\subsection{Generalizations of Davenport-Schinzel Sequences}

The (\ref{eqn:Nivasch}) bounds are actually corollaries of a more general theorem
in~\cite{Nivasch10} 
concerning the length of sequences avoiding {\em catenated permutations},\footnote{Nivasch called
these formation-free sequences.}
which were introduced by Klazar~\cite{Klazar92}.
Define $\Perm{r,s+1}$ to be the set of all sequences obtained by concatenating $s+1$ permutations
over an $r$-letter alphabet.  For example, $abcd\; cbad\; badc\; abcd\; dcba\in \Perm{4,5}$.
Define the extremal function of $\Perm{r,s+1}$-free sequences to be 
\begin{align*}
\PERM{r,s}(n) &= \Ex(\Perm{r,s+1},n)
\end{align*}
The ``$s+1$'' here is chosen to highlight the 
parallels with order-$s$ DS sequences.  Every $\sigma\in\Perm{2,s+1}$ contains an alternating sequence 
$abab\cdots$ with length $s+2$,\footnote{The first permutation contributes two symbols and every subsequent permutation contributes
at least one.}
so order-$s$ DS sequences are also $\Perm{2,s+1}$-free, implying that $\DS{s}(n) \le \PERM{2,s}(n)$.
Nivasch~\cite{Nivasch10} proved that $\PERM{r,s}(n)$ obeys all the {\em upper bounds}
of (\ref{eqn:Nivasch}), as well as its  lower bounds when $s\ge 4$ is even or $s\le 3$. 

There are other natural ways to generalize standard Davenport-Schinzel sequences.
{\em Doubled} Davenport-Schinzel sequences were studied in~\cite{DS65b,AKV92,KV94,Pettie-GenDS11}.
Define $\dblDS{s}(n)$ to be the extremal function of $\inflate(abab\cdots)$-free sequences,
where the alternating sequence has length $s+2$ and
$\inflate(\sigma)$ is obtained by doubling every symbol in $\sigma$ save the first and last. For example, 
$\inflate(abab) = abbaab$.\footnote{Why not consider higher multiplicities?  It is fairly easy to show
that repeating symbols more than twice, or repeating the first and last at all, affects the extremal function by 
at most a constant factor.  See Adamec, Klazar, and Valtr~\cite{AKV92}.}
Davenport and Schinzel~\cite{DS65b} noted that $\dblDS{1}(n)=O(\DS{1}(n))=O(n)$ (see~\cite{Klazar02})
and Adamec, Klazar, and Valtr~\cite{AKV92} proved that $\dblDS{2}(n)=O(\DS{2}(n)) = O(n)$.
Pettie proved that $\dblDS{3}(n) = O(n\alpha^2(n))$ and that
$\dblDS{s}(n)$ obeys all the upper bounds of (\ref{eqn:Nivasch}) for $s\ge 4$.

If one views alternating sequences as forming a zigzagging pattern,
an obvious generalization is to extend the length of each zig and zag to include a larger alphabet.  
For example, the $N$-shaped sequences
$N_{k} = 12\cdots k(k+1)\cdots 212 \cdots k(k+1)$ generalize $abab=N_1$
and the $M$-shaped sequences $M_{k} = 12\cdots k(k+1)k\cdots 212\cdots k(k+1)k\cdots 21$ generalize
$ababa=M_1$.
Klazar and Valtr~\cite{KV94,Pettie-SoCG11} proved that 
$\Ex(\inflate(N_k),n) = O(\DS{2}(n)) = O(n)$
and Pettie~\cite{Pettie-SoCG11} proved that $\Ex(\{M_k,ababab\},n) = O(\DS{3}(n))$.
Sequences avoiding $N$- and $M$-shaped sequences have proved very useful
in bounding the complexity of geometric graphs~\cite{Valtr97,Suk11,FoxPS11,Pettie-SoCG11}. 

In a companion paper to be published separately we provide new upper and lower bounds on doubled DS sequences,
$M_k$-free sequences, and $\Perm{r,s+1}$-free sequences.
The strangest of these results is that $\PERM{r,s}$ is very sensitive to the alphabet size $r$, but only when $s$ is odd and at least 5.  In particular $\PERM{2,s}(n)=\Theta(\dblDS{s}(n)) = \Theta(\DS{s}(n))$
but this is not true for general $r\neq 2$.

\ignore{
We prove that $\dblDS{s}(n)$ obeys all the upper and lower
bound of Theorem~\ref{thm:main}, except at $s=5$,
where the upper bound is weaker by an $\alpha^{1+o(1)}(n)$ factor.
The most technically challenging part of this proof is establishing an $O(n\alpha(n))$ bound on $\dblDS{3}(n)$.\footnote{Klazar  and Valtr~\cite{KV94} claimed that this bound follows easily from Hart and Sharir's~\cite{HS86} method.
This claim was later retracted by Klazar~\cite{Klazar02} and highlighted as an open question.}

The situation becomes stranger when we attempt to extend the bounds of Theorem~\ref{thm:main} to 
$\Perm{r,s+1}$-free sequences.  When $s\le 3$ or $s$ is even it was known that $\PERM{r,s}(n)$ behaves like $\DS{s}(n)$ for any $r\ge 2$.
(Nivasch's upper bounds~\cite{Nivasch10} match the lower bounds of~\cite{HS86} and \cite{ASS89} for these parameters.)
We prove that for every {\em odd} $s\ge 5$, $\PERM{r,s}(n)$ behaves essentially like $\DS{s}(n)$ {\em but only when $r=2$}.  For $r\ge 3$ we give
a new lower bound construction showing that Nivasch's upper bound is essentially tight.
}

\begin{theorem}\label{thm:PERM}
The following bounds hold for all $r\ge 2, s\ge 1$, where $t = \floor{\frac{s-2}{2}}$.
\begin{align*}
\PERM{r,s}(n) &= \left\{
\begin{array}{l@{\hcm[.4]\istrut[2.5]{0}}l}
\Theta(n)		& \mbox{for $s\in\{1,2\}$ and all $r\ge 2$}\\
\Theta(n\alpha(n)) & \mbox{for $s=3$ and all $r\ge 2$}\\
\Theta(n2^{\alpha(n)}) & \mbox{for $s=4$ and all $r\ge 2$}\\
\Theta(n\alpha(n)2^{\alpha(n)})    & \istrut[3]{0}\mbox{for $s=5$ and  $r=2$}\\
n\cdot (\alpha(n))^{(1+o(1))\,\alpha(n)}						& \mbox{for $s=5$ and all $r\ge 3$}\\
n\cdot 2^{(1+o(1))\,\alpha^t(n)/t!}						& \mbox{for even $s\ge 6$ and all $r\ge 2$}\\
n\cdot 2^{(1+o(1))\,\alpha^t(n)/t!}						& \mbox{for odd $s\ge 7$ and $r=2$}\\
n\cdot (\alpha(n))^{(1+o(1))\,\alpha^t(n)/t!} \;\;\;					& \mbox{for odd $s\ge 7$ and all $r\ge 3$}
\end{array}\right.
\end{align*}
\end{theorem}

\ignore{
One consequence of 
Theorem~\ref{thm:PERM} is that Cibulka and Kyn\v{c}l's~\cite{CibulkaK12,Raz00}
recent upper bounds on the size of sets of permutations with fixed VC dimension is tight.
}

Theorem~\ref{thm:PERM} is rather surprising, even given Theorem~\ref{thm:main} and even in retrospect.
One consequence of Theorem~\ref{thm:PERM} is that Cibulka and Kyn\v{c}l's~\cite{CibulkaK12,Raz00}
upper bounds on the size of sets of permutations with fixed VC-dimension are tight.

\subsection{Organization.}

In Section~\ref{sect:tour} we present an informal discussion of the method
of Agarwal, Sharir, and Shor~\cite{ASS89} and Nivasch~\cite{Nivasch10},
its limitations for dealing with odd-order DS sequences, and 
the key ideas behind the proof of Theorem~\ref{thm:main}.
Section~\ref{sect:basic} reviews Nivasch's recurrence for $\DS{s}$
as well as some basic upper bounds on $\DS{s}$.
The critical structure in our analysis is the {\em derivation tree}
of a DS sequence.
Its properties are analyzed in Section~\ref{sect:derivation-tree}.
In Section~\ref{sect:odd} we use the derivation tree to 
obtain a new recurrence for odd-order DS sequences.
The recurrences for even- and odd-order DS sequences are
solved in Section~\ref{sect:analysis}.
In Section~\ref{sect:proofthmmain} we complete the proof of the upper bounds of 
Theorem~\ref{thm:main} for all orders except $s=5$.  
In Sections~\ref{sect:lb} and~\ref{sect:order5-ub} we establish 
Theorem~\ref{thm:main}'s lower bounds and upper bounds 
on order-$5$ DS sequences.  
We discuss several open problems in Section~\ref{sect:discussion}.
Some proofs appear in Appendices~\ref{appendix:Sharir}--\ref{sect:appendix:order5}.

\ignore{
Section~\ref{sect:abbaabba} extends Theorem~\ref{thm:main} and the derivation tree method
to double Davenport-Schinzel sequences. 
It is proved that $\dblDS{3}(n) = \Theta(\DS{3}(n)) = \Theta(n\alpha(n))$
and, in general, that $\dblDS{s}(n)$ behaves essentially like $\DS{s}(n)$.
Section~\ref{sect:abcbabcba} presents a new class of nonlinear sequences
with lengths of the form $\Theta(n\alpha^k(n))$, where $k\ge 1$,
and provides new lower bounds on sequences avoiding $M$-shaped and comb-shaped subsequences.
Theorem~\ref{thm:PERM} is proved in Section~\ref{sect:perm}.
A new construction of $\Perm{r,s+1}$-free sequences (using those of Section~\ref{sect:abcbabcba} as a basis) 
is given in Section~\ref{sect:PERMlb}.  In the interest of completeness, 
an alternative proof of Nivasch's upper bound~\cite{Nivasch10} is presented
in Section~\ref{sect:PERMub}.
}

\section{A Tour of the Proof}\label{sect:tour}

The proof of Theorem~\ref{thm:main} diverges sharply from previous 
analyses~\cite{Sharir87,ASS89,Nivasch10}
in that it treats even and odd orders as fundamentally different beasts.
To understand why all orders cannot be analyzed in a uniform fashion we must
review the method of Agarwal, Sharir, and Shor~\cite{ASS89} and Nivasch~\cite{Nivasch10}.

The basic inductive hypothesis of~\cite{Nivasch10} is that there are values $\{\mu_{s,i}\}$, 
increasing in both $s$ and $i$,
for which
$\DS{s}(n,m) < \mu_{s,i}\paren{n + m\poly(\log^{[i-1]}(m))}$,
for any choice of $i$.\footnote{Recall that $\log^{[i-1]}(m)$ is the $\log^{\star\cdots\star}(m)$ function, with $i-1$ $\star$s.}
In other words, the multiplicity of symbols is at most $\mu_{s,i}$, up to an additive term that depends on the block count $m$,
which may be the dominant term if $i$ is too small.
One would, ultimately, choose $i$ to make $m\poly(\log^{[i-1]}(m)) = O(n)$ (and $\alpha(n,m)$ is a good choice)
but for the sake of simplifying the discussion we ignore the dependence on $m$.  Given a sequence $S$ with parameters $s,n,m$, 
to {\em invoke the inductive hypothesis with parameter $i$} means to upper bound $|S|$ by $\mu_{s,i}n$, with the understanding
that $i$ is not chosen to be too small.

Suppose $S$ is an order-$s$, $m$-block DS sequence over an $n$-letter alphabet.  The analysis of~\cite{ASS89,Nivasch10}
begins by partitioning $S$ into $\Gm$ intervals of consecutive blocks, $\Gm$ typically being much smaller than $m$.
Write $S$ as $S_1\cdots S_{\Gm}$.  We can put symbols into two categories: {\em local} symbols are those that appear exclusively
in one interval $S_q$ and {\em global} symbols are those that appear in multiple intervals.  
Let $\LS=\LS_1\cdots \LS_{\Gm}$ be the subsequence of $S$ consisting of local symbols
and $\GS = \GS_1\cdots \GS_{\Gm}$ the subsequence of $S$ consisting of global symbols, so $|S| = |\LS| + |\GS|$.
On each $\LS_q$ (for $1\le q\le \Gm$) we invoke the inductive hypothesis with parameter $i$ and deduce that $|\LS| \le \mu_{s,i}\|\LS\|$.
What remains is to bound the length of $\GS$.

The next step is to form a {\em contracted} sequence from $\GS$ that has a much higher alphabet-to-block count
ratio, thereby allowing us to invoke the inductive hypothesis with a smaller `$i$' parameter.  Let $\GS'$ be obtained from $\GS$ by replacing each 
interval $\GS_q$ with a {\em single} block $\beta_q$ containing the first occurrence of each distinct symbol in $\GS_q$.  Thus, the alphabet of 
$\GS'$ is the same as $\GS$ but it consists of just $\Gm$ blocks $\beta_1\beta_2\cdots\beta_{\Gm}$.  
On $\GS'$ we invoke the inductive hypothesis with parameter $i-1$
and conclude that $|\GS'| \le \mu_{s,i-1}\|\GS\|$.  One cannot immediately deduce any bound on $\GS$ from a bound on $\GS'$ since
each interval $\GS_q$ could contain {\em numerous} copies of a symbol, only one of which is retained in $\beta_q$.

Imagine reversing the contraction operation.  We replace each block $\beta_q$ with a sequence $\GS_q$, thereby reconstructing $\GS$.
To bound the length of $\GS_q$ in terms of $|\beta_q| = \|\GS_q\|$ we will invoke the inductive hypothesis three more times.
Put the symbols of $\beta_q$ in three categories: those that make their first appearance (in $\GS'$) in $\beta_q$, those that make their last appearance in $\beta_q$,
and those that make a middle (non-first, non-last) appearance in $\beta_q$. 
Discard from $\GS_q$ all symbols not classified as first in $\beta_q$ and call the resulting sequence $\GfS_q$.  
Every symbol in $\GfS_q$ appears at least once after $\GS_q$ (by virtue of being categorized as first in $\beta_q$), 
which implies that $\GfS=\GfS_1\cdots\GfS_{\Gm}$ is an order-$(s-1)$ DS sequence.  See the diagram below.\\

\centerline{\scalebox{.3}{\includegraphics{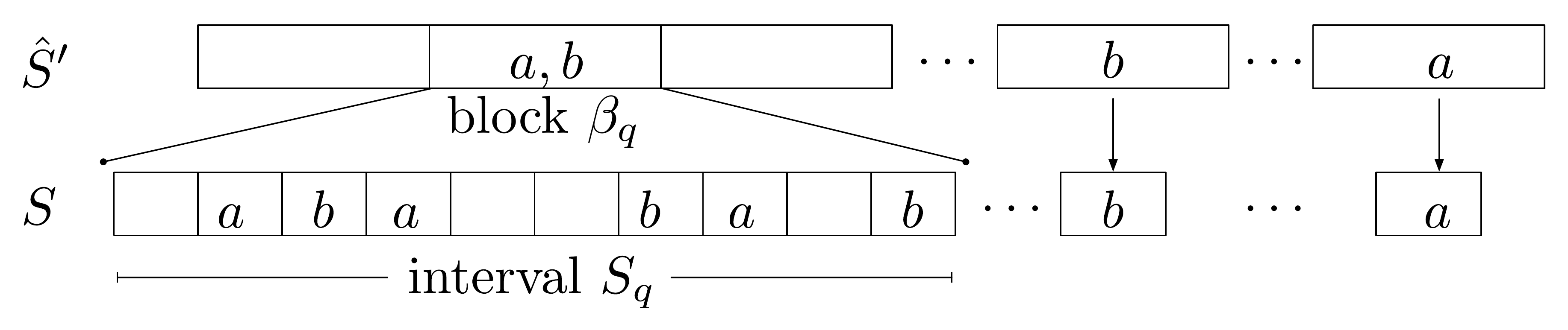}}}

\noindent An occurrence of $\sigma_{s+1} = abab\cdots$ (length $s+1$) in $\GfS_q$,  together with an $a$ or $b$ following $\GS_q$ (depending on whether $\sigma_{s+1}$ ends in $b$ or $a$) gives an occurrence of $\sigma_{s+2}$ in $\GS$, contradicting the fact that it has order $s$.  The same argument applies
in a symmetric fashion to the subsequence of $\GS_q$ formed by symbols making their last appearance in $\beta_q$, call it $\GlS_q$.
By invoking the inductive hypothesis with parameter $i$ on $\GfS_1\cdots \GfS_{\Gm}$ and $\GlS_1\cdots \GlS_{\Gm}$ we can conclude the contribution
of first and last symbols to $S$ is $2\mu_{s-1,i}\|\GS\|$.

The length of the subsequence of middle symbols in $\GS_q$, call it $\GmS_q$, 
is bounded with the same argument, except now there are, by definition of {\em middle}, occurrences of both $a,b\in\Sigma(\GmS_q)$ both before {\em and} after $\GS_q$.
That is, if $\sigma_{s} = baba\cdots$ (length $s$) appeared in $\GmS_q$ then, together with an $a$ preceding $\GS_q$ and either an $a$ or $b$ following $\GS_q$ (depending on whether $\sigma_{s}$ ends in $b$ or $a$) there would be an occurrence of $\sigma_{s+2}$ in $\GS$, contradicting the fact that it has order $s$.
We invoke the inductive hypothesis one last time, with parameter $i$, on {\em each} $\GmS_1,\ldots,\GmS_{\Gm}$, which implies that 
$|\GmS_q| \le \mu_{s-2,i}\|\GmS_q\|$.  

Recall that each symbol in $\GS'$ appeared $\mu_{s,i-1}$ times, 
$\mu_{s,i-1}-2$ times 
in blocks where it was categorized as middle.  Thus, the contribution of middle symbols to $|\GS|$ is $\mu_{s-2,i}(\mu_{s,i-1}-2)$.  In order for every
symbol, local and global alike, to appear in $S$ with multiplicity at most $\mu_{s,i}$, we must have
\begin{align}
\mu_{s,i} &\ge 2\mu_{s-1,i} + \mu_{s-2,i}(\mu_{s,i-1}-2)\label{eqn:Tour-basic}
\end{align}
When $s=3$ we do not need to use an inductive hypothesis to determine $\mu_{1,i}$ and $\mu_{2,i}$.  They are just 1 and 2; the $i$ parameter does not come into play.\footnote{It is easy to show that 
$\DS{1}(n,m) < \underline{1}\cdot n + m$ and $\DS{2}(n,m) < \underline{2}\cdot n + m$.  See Lemma~\ref{lem:DS12}.}  This leads to a bound of $\mu_{3,i} = 2i+O(1)$.\footnote{We have not said what to do in the base case when $i=1$, which determines the $O(1)$ term.}
Although the contribution of first and last symbols is significant at $s=3$,
entertain the idea that their contribution becomes negligible at higher orders, 
so we can further simplify (\ref{eqn:Tour-basic}) as follows
\begin{align}
\mu_{s,i} &\ge \mu_{s-2,i}\mu_{s,i-1}  \label{eqn:Tour-simplified}
\end{align}
Inequality (\ref{eqn:Tour-simplified}) is satisfied when $\mu_{s,i} = g^{i+t\choose t}$ for {\em any} base $g$; recall that 
$t=\floor{\frac{s-2}{2}}$ by definition.
By Pascal's identity $g^{i+t\choose t} = g^{i+(t-1)\choose (t-1)}\cdot g^{(i-1) + t\choose t}$.  The correct base depends on where
the inductively defined inequality (\ref{eqn:Tour-simplified}) bottoms out: at order $2$ when $s\ge 4$ is even and at order $3$ when $s\ge 5$ is odd.
When $s$ is even the correct base is $2 = \mu_{2,i}$.  When $s$ is odd the calculations are less clean since $\mu_{3,i} = 2i+ O(1)$ is not constant 
but depends on $i$.  Nonetheless, the correct base is on the order of $i$, that is, $\mu_{s,i} = \Theta(i)^{i+t\choose t}$ satisfies (\ref{eqn:Tour-simplified})
at the odd orders.  Plugging in $\alpha(n,m)$ for $i$ ultimately leads to Nivasch's bounds~(\ref{eqn:Nivasch}), since 
${i+t\choose t} = i^t/t! + O(i^{t-1}) = (1+o(1))i^t/t!$.\\

To obtain a construction of order-$s$ sequences realizing the (\ref{eqn:Nivasch}) bounds one should start
by attempting to {\em reverse-engineer} the argument above.  To form an order-$s$ sequence $S$ with certain alphabet and block parameters, 
start by generating (inductively) local order-$s$ sequences $\LS_1\cdots\LS_{\Gm}$ over disjoint alphabets, and a single 
global order-$s$ sequence $\GS'$ having $\Gm$ blocks.  Take some block $\beta_q$ in $\GS'$ and suppose for the sake of simplicity that $\beta_q$ consists
solely of middle symbols.  We need to substitute for $\beta_q$ an order-$(s-2)$ DS sequence $\GmS_q$ and then somehow merge it with $\LS_q$,
in a way that does not introduce into 
$S$ an alternating sequence with length $s+2$.  This is the point at which the even and odd orders diverge. 

If $s$ is even the longest alternating sequence $baba\cdots b$ in $\GmS_q$ has length $s-1$ and therefore begins and ends with $b$.  
We can only afford to introduce one alternation at each boundary of $\GmS_q$, so the pattern of $a$s and $b$s on either side of $\beta_q$ must
look like $a^* \; b^* \;\, \beta_q \;\, b^* \; a^*$, as in the diagram below.
We will call $a$ and $b$ {\em nested in $\beta_q$} if the sequence contains
$a\, b\, \beta_q\, b\, a$ or the equivalent $b\, a\, \beta_q\, a\, b$.\\
See the diagram below.

\centerline{\scalebox{.3}{\includegraphics{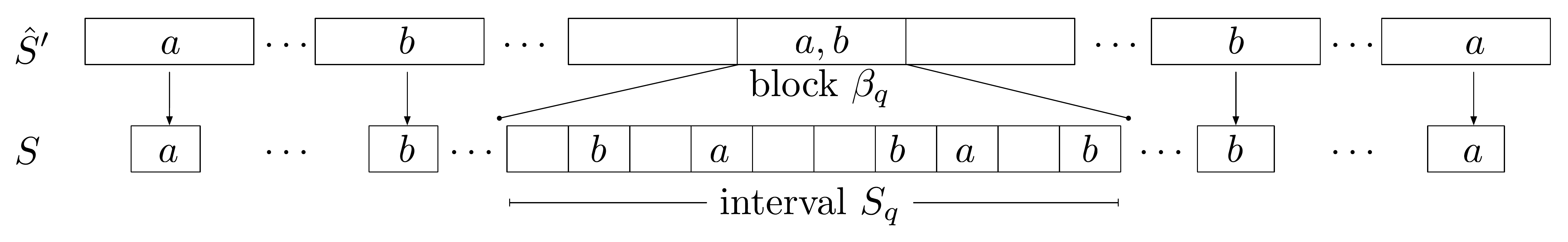}}}

On the other hand, if $s$ is odd then 
the longest alternating sequence $baba\cdots a$ in $\GmS_q$ has length $s-1$, begins with $b$ and ends with $a$, so the pattern
of $a$s and $b$s in $\GS'$ looks like $a^* \; b^* \; \beta_q \; a^* \; b^*$.
A pair of middle symbols that are not nested in $\beta_q$ are called {\em interleaved in $\beta_q$}.\\

\centerline{\scalebox{.3}{\includegraphics{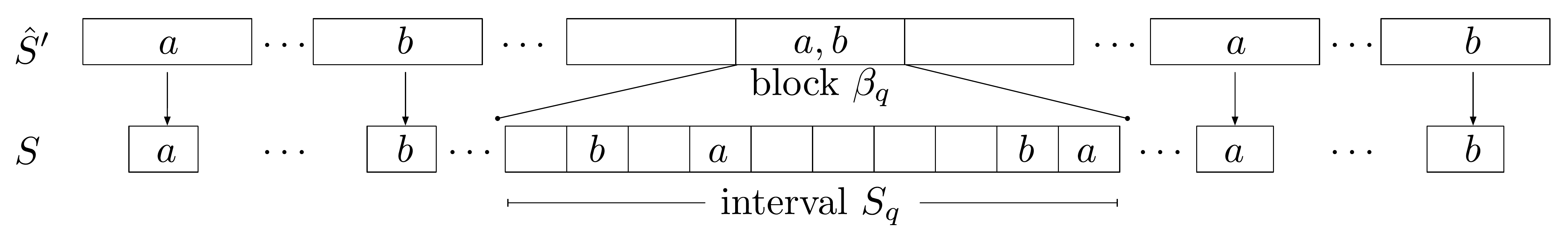}}}

If the (\ref{eqn:Nivasch}) bounds prove to be tight, there must be two systems for generating sequences:
one where nesting is the norm, when $s$ is even, and one where interleaving is the norm, when $s$ is odd.  
If interleaving were somehow outlawed then
to avoid creating an alternating sequence with length $s+2$, the sequence $\GmS_q$ substituted for $\beta_q$ would have to be an 
order-$(s-3)$ DS sequence rather than an order-$(s-2)$ one.
However,
it is clearly impossible to claim that interleaving simply cannot exist.

\ignore{
Even if one could argue that nesting is the norm
and that there are very few pairs\footnote{As a rule, very few of the $n\choose 2$ pairs of symbols in a Davenport-Schinzel sequence form any non-trivial alternating pattern.} of interleaving symbols in each block, that would not be sufficient to prove that $\GS_q$ has order $s-3$.
}

What makes the argument of~\cite{ASS89,Nivasch10} brilliantly simple
is how little it leaves to direct calculation.  The length of every sequence ($\LS_q,\GS',\GfS_q,\GmS_q,$ etc.) 
is bounded by delegation to an inductive hypothesis.  However, such useful notions as {\em nearly all middle symbols in a block
are mutually nested} are difficult to capture in a strengthened inductive hypothesis.
We need to understand and characterize 
the phenomenon of nestedness to improve on~\cite{ASS89,Nivasch10}.
This requires a deeper understanding of the structure of Davenport-Schinzel sequences.

\paragraph{The Derivation Tree.}
Inductively defined objects can be apprehended inductively or, alternatively,
apprehended holistically 
by completely ``unrolling'' the induction. 
From the first perspective $S$ is the merger of $\LS$ and $\GS$, which is derived from $\GS'$, all of which are
analyzed inductively.
By iteratively unrolling the decomposition of $\LS$ and $\GS'$ we obtain a {\em derivation tree} $\Tree$ whose nodes 
represent {\em every} block in {\em every} sequence encountered in the recursive decomposition of $S$.
Whereas $S$ occupies the leaves of $\Tree$, derived sequences such as $\GS'$ occupy levels higher in $\Tree$.
Whereas $S$ (and every sequence) is a static object, $\Tree$ can be thought of a {\em process} for generating $S$ whose history 
can be reasoned about explicitly.
But how does $\Tree$ let us deduce something about the nestedness and non-nestedness
of symbols in a common block?

Suppose we are interested in the 
nestedness of middle symbols $a,b$ in block $\beta$, which corresponds
to a leaf-node in $\Tree$.  Imagine taking $\Tree$ and deleting every node whose block does not contain $b$,
that is, {\em projecting} $\Tree$ onto the symbol $b$.
What remains, $\Tree_{|b}$, is a tree rooted at the location in $\Tree$ where $b$ is ``born'' and represents how occurrences
of $b$ have proliferated during the process that culminates in the construction of $S$.
The block/node $\beta$ occupies a location in $\Tree_{|b}$ and a location in $\Tree_{|a}$,
whose node sets are only guaranteed to intersect at $\beta$.
Some locations in $\Tree_{|a}$ and $\Tree_{|b}$ are intrinsically bad---these are called {\em feathers} in Section~\ref{sect:derivation-tree}.
(Whether a node is a feather in $\Tree_{|a}$ depends solely on the structure of $\Tree_{|a}$, not how it is embedded
in $\Tree$ nor its relationship to a different $\Tree_{|b}$.)
We show that if $\beta$ is not a feather in $\Tree_{|a}$ and not a feather in $\Tree_{|b}$, then 
$a$ and $b$ are nested in $\beta$.  
In other words, the middle symbols in $\beta$ are partitioned into two equivalence classes, depending
on whether or not they appear at feathers in their respective derivation trees.  We could not outlaw
interleavedness in general, yet we managed to outlaw it {\em within} one equivalence class!
The question is, what are the relative sizes of these two equivalence classes, and in particular,
how many feathers can a $\Tree_{|b}$ have?\\

Our aim is to get stronger asymptotic bounds on $\DS{s}$ for odd $s$, which means
the number of feathers should be a negligible ($o(1)$ fraction) of the size of $\Tree_{|b}$.
In the same way that the multiplicities $\mu_{s,i}$ are bounded inductively, as in (\ref{eqn:Tour-basic}) for example, 
we are able to bound the number of feathers in $\Tree_{|b}$ inductively, call it $\nu_{s,i}$, in terms of $\nu_{s,i-1}$ and
$\mu_{s-1,i}$.  However, now $\mu_{s,i}$ is bounded in terms of $\mu_{s,i-1}$ (the multiplicity of symbols in the contracted sequence $\GS'$),
$\mu_{s-1,i}$ (the multiplicity of symbols in $\GS$ begat by first and last occurrences in $\GS'$),
$\mu_{s-3,i}$ (the multiplicity of middle occurrences in $\GS$ begat by non-feathers in $\GS'$),
and both
$\nu_{s,i-1}$ and $\mu_{s-2,i}$, which count the number of feathers in 
$\GS'$ and the multiplicity of middle occurrences in $\GS$ begat by feathers in $\GS'$.
This leads to a system of three interconnected recurrences: one for $\mu_{s,i}$ at odd $s$, one for $\mu_{s,i}$ at 
even $s$, and one for the feather count $\nu_{s,i}$.
An elementary (though necessarily detailed)
proof by induction gives solutions for $\mu_{s,i}$ and $\nu_{s,i}$ 
that ultimately lead to the upper bounds of Theorem~\ref{thm:main}, 
with one exception.  At order $s=5$ this method only gives us 
an $O(n\alpha^2(n)2^{\alpha(n)})$ upper bound on $\DS{5}(n)$.
To obtain a sharp $O(n\alpha(n)2^{\alpha(n)})$ bound we
are forced to analyze not just one derivation tree $\Tree$ (of an order-5 DS sequence) but
a system of derivation trees of order-4 DS sequences associated with all the 
sequences $\GfS$ and $\GlS$ (of global first and last occurrences) encountered in the construction of $\Tree$.

\ignore{
Although the derivation tree $\Tree$ was defined with standard Davenport-Schinzel sequences in mind,
it has further uses.  Our proof that $\dblDS{3}(n) = \Theta(n\alpha(n))$
is based on different properties of $\Tree$.

\begin{remark}\label{remark:r=3}
Why, one might ask, does nestedness give us an edge in bounding the length
of Davenport-Schinzel sequences and $\Perm{2,s}$-free sequences but
not $\Perm{r,s}$-free sequences for $r\ge 3$?  (See Theorem~\ref{thm:PERM}.)
Roughly speaking, the twoness of sequences (each interval in a sequence has two boundaries with the rest of the sequence)
matches the twoness of forbidden subsequences over a 2-letter alphabet.
When the forbidden subsequence has $r\ge 3$ letters, the third letter acts
as a kind of ``sentinel'' at one boundary of an interval, which makes the distinction between 
nested and interleaved symbols irrelevant.  
(This explanation is merely intended to foreshadow the proof of Lemma~\ref{lem:perm-U} in Section~\ref{sect:PERMlb},
not to be understood in isolation.)
Why, then, do $\PERM{r,s}(n)$ and $\DS{s}(n)$ behave differently at the odd orders $s\ge 5$
but the same at order $s=3$?  
When $s=3$ the global middle symbols in $\GS'$ beget $\mu_{s-2,\cdot} = \mu_{1,\cdot}$ symbols in $S$,
where $\mu_{1,\cdot}$ represents the multiplicity of symbols in $\DS{1}$ or $\PERM{r,1}$.
Since $\mu_{1,\cdot} = 1$, whether middle symbols sharing a block in $\GS'$ are nested or not
makes no difference.
\end{remark}
}

\section{Basic Upper Bounds}\label{sect:basic}

In Section~\ref{sect:notation} we review and expand on the notation
introduced informally in Section~\ref{sect:tour}.
It will be used repeatedly throughout 
Sections \ref{sect:derivation-tree}--\ref{sect:order5-ub}.

\subsection{Sequence Decomposition}\label{sect:notation}

Let $S$ be a sequence over an $n=\|S\|$ letter alphabet consisting of $m=\bl{S}$ blocks.
Suppose we partition $S$ into $\Gm$ intervals of consecutive blocks 
$S_1S_2\cdots S_{\Gm}$, where $\Lm_q = \bl{S_q}$ is the number of blocks in interval $q$.
Let $\LSigma_q$ be the alphabet of symbols
{\em local} to $S_q$ (that do not appear in any $S_p$, $p\neq q$)
and let $\GSigma = \Sigma(S) \backslash \bigcup_{q}\LSigma_q$ be the alphabet of all other {\em global} symbols.
The cardinalities of $\LSigma_q$ and $\GSigma$ are $\Ln_q$ and $\Gn$, thus $n = \Gn + \sum_{q=1}^{\Gm} \Ln_q$.
A global symbol in $S_q$ is called {\em first}, {\em last}, or {\em middle} if it appears in no earlier interval, no later interval, or appears in both earlier and later intervals, respectively.
Let $\GfSigma_q,\GlSigma_q,\GmSigma_q,\GSigma_q$ be the subset of $\Sigma(S_q)$ consisting of, respectively, first, last, middle, and all global symbols,
and let $\Gfn_q,\Gln_q,\Gmn_q,$ and $\Gn_q$ be their cardinalities.
Let $\LS_q,\GS_q,\GfS_q,\GlS_q,\GmS_q$ be the projection of 
$S_q$ onto $\LSigma_q,\GSigma_q,\GfSigma_q,\GlSigma_q,$ and $\GmSigma_q$.
Note that $\GS_1$ consists solely of first occurrences; if the last occurrence of a symbol appeared in $\GS_1$ the symbol would be
classified as local to $S_1$, not global.  The same argument shows that $\GS_{\Gm}$ consists solely of last occurrences.
Let $\LS,\GS,\GfS,\GlS,$ and $\GmS$ be the subsequences of local, global, first, last, and middle occurrences, respectively,
that is, 
$\LS = \LS_1\cdots\LS_{\Gm}$, 
$\GS = \GS_1\cdots\GS_{\Gm}$,
$\GfS = \GfS_1\cdots \GfS_{\Gm-1}$,
$\GlS = \GlS_2\cdots \GlS_{\Gm}$,
and
$\GmS = \GmS_2\cdots \GmS_{\Gm-1}$, the last of which would be empty if $\Gm=2$.
Let $\GS' = \beta_1\cdots\beta_{\Gm}$ be an $\Gm$-block sequence
obtained from $\GS$ by replacing each $\GS_q$ with a single block $\beta_q$ 
containing its alphabet $\GSigma_q$, listed in order of {\em first} appearance in $\GS_q$.

\subsection{2-Sparse vs. Blocked Sequences}\label{sect:sparse-vs-blocked}

Every analysis of Davenport-Schinzel sequences since~\cite{HS86} 
uses Lemma~\ref{lem:Sharir}(\ref{item:Sharir2}) to reduce the problem
of bounding 2-sparse DS sequences to bounding $m$-block DS sequences, 
that is, expressing $\DS{s}(n)$ in terms of $\DS{s}(n,m)$, where $m=O(n)$.

\begin{lemma}\label{lem:Sharir}
Let $\gamma_s(n) : \mathbb{N}\rightarrow\mathbb{N}$ be a non-decreasing function such that 
$\DS{s}(n) \le \gamma_s(n)\cdot n$.
\begin{enumerate}
\item (Trivial) For $s\ge 1$,  $\DS{s}(n,m) \le m-1 + \DS{s}(n)$.\label{item:Sharir1}
\item (Sharir~\cite{Sharir87}) For $s\ge 3$, $\DS{s}(n) \le \gamma_{s-2}(n)\cdot \DS{s}(n,2n-1)$.
(This generalizes Hart and Sharir's proof~\cite{HS86} for $s=3$.)\label{item:Sharir2}
\item (Sharir~\cite{Sharir87}) For $s\ge 2$, $\DS{s}(n) \le \gamma_{s-1}(n)\cdot \DS{s}(n,n)$.\label{item:Sharir3}
\item (New) For $s\ge 3$, $\DS{s}(n) = \gamma_{s-2}(\gamma_s(n))\cdot \DS{s}(n,3n-1)$.
\label{item:Sharir4}
\end{enumerate}
\end{lemma}

Lemma~\ref{lem:Sharir}(\ref{item:Sharir4}) is obtained by synthesizing ideas from Sharir~\cite{Sharir87}
and \Furedi{} and Hajnal~\cite{FurediH92}.  
Refer to Appendix~\ref{appendix:Sharir} for the proof of Lemma~\ref{lem:Sharir}.

\subsection{Orders 1 and 2}\label{sect:order12}

In the interest of completeness we shall reestablish the known bounds on order-1 and order-2 DS sequences,
in both their 2-sparse and blocked forms.

\begin{lemma}\label{lem:DS12}
(Davenport and Schinzel~\cite{DS65})
The extremal functions for order-1 and order-2 DS sequences are
\begin{align*}
\DS{1}(n) &= n\\
\DS{2}(n) &= 2n-1\\
\DS{1}(n,m) &= n+ m-1\\
\DS{2}(n,m) &= 2n+m-2		& \mbox{for $m\ge 2$}
\end{align*}
\end{lemma}

\begin{proof}
Let $S$ be a 2-sparse sequence with $n=\|S\|$.
If $|S| > n$ then there are two copies of some symbol, say $a$.
The $a$s cannot be adjacent, due to 2-sparseness, so $S$
must contain a subsequence $aba$, for some $b\neq a$.
Such an $S$ is not an order-1 DS sequence, hence $\DS{1}(n) \le n$.

If $S$ has order 2 then some symbol must appear exactly once.
To see this, consider the closest pair of occurrences of some symbol, say $a$.
If every symbol $b$ appearing between this pair of $a$s occurred twice in $S$
then $S$ would contain $baba, abab,$ or $abba$.  The first two are precluded since
$S$ has order 2 and the third violates the fact that the two $a$s are the {\em closest} such pair.
Thus, every symbol $b$ between the two $a$s occurs once.  Remove one such $b$;
if this causes the two $a$s to become adjacent, remove one of the $a$s.  What remains
is a 2-sparse sequence over an $(n-1)$-letter alphabet, so $\DS{2}(n) \le \DS{2}(n-1) + 2$. 
Since $\DS{2}(1)=1$ we have $\DS{2}(n) \le 2n-1$.

Lemma~\ref{lem:Sharir}(\ref{item:Sharir1}) and the bounds established above
imply $\DS{1}(n,m) \le n + m-1$ and $\DS{2}(n,m) \le 2n + m - 2$.
All these upper bounds are tight.  The unique extremal order-1, 2-sparse DS sequence 
is $123\cdots n$, which can be converted into an extremal $m$-block sequence
$[123\cdots n][n]^{m-1}$.  Brackets mark block boundaries.
There are exponentially many extremal DS sequences of order $2$,
each corresponding to an Euler tour around a rooted tree with vertex labels from $\{1,\ldots,n\}$.
For example, $123\cdots (n-1)n(n-1)\cdots 321$ and $1213141\cdots 1(n-1)1n1$ are extremal 2-sparse, 
order-2 DS sequences.  
The first corresponds to an Euler tour around a path, the second an Euler tour around a star.
The first sequence can be converted into an extremal $m$-block, order-2 DS sequence
$[12\cdots (n-1)n][n(n-1)\cdots 21][1]^{m-2}$, assuming that $m\ge 2$.  When there is only 1 block
we have $\DS{s}(n,1) = n$, regardless of the order $s$.
\end{proof}

\subsection{Nivasch's Recurrence}\label{sect:basic-ub}

Nivasch's~\cite{Nivasch10} upper bounds (\ref{eqn:Nivasch}) are a consequence of a recurrence for $\DS{s}$ 
that is stronger than one of Agarwal, Sharir, and Shor~\cite{ASS89}.  Here we 
present a streamlined version of Nivasch's recurrence.

\begin{recurrence}\label{thm:recurrence-even}
Let $m,n,$ and $s\ge 3$ be the block count, alphabet size, and order parameters.
For any $\Gm<m$, any block partition $\{\Lm_q\}_{1\le q\le \Gm}$,
and any alphabet partition $\{\Gn\}\cup \{\Ln_q\}_{1\le q\le \Gm}$,
where $m = \sum_q \Lm_q$ and $n=\Gn + \sum_q \Ln_q$, we have
\begin{equation}\label{eqn:even}
\DS{s}(n,m) \;\leq\; \sum_{q=1}^{\Gm} \DS{s}(\Ln_q,\Lm_q) \;+\; 2\cdot \DS{s-1}(\Gn,m) \;+\; \DS{s-2}(\DS{s}(\Gn,\Gm) - 2\Gn, m)\nonumber
\end{equation}
\end{recurrence}

\begin{proof}
We adopt the notation and definitions from Section~\ref{sect:notation}, where
$S$ is an extremal order-$s$ DS sequence with $\|S\|=n$ and $\bl{S}=m$.
We shall bound $|S|$ by considering its four constituent subsequences
$\LS,\GfS,\GlS,$ and $\GmS$.
%

Each $\LS_q$ is an order-$s$ DS sequence, therefore the contribution of local symbols is 
$|\LS|\le \sum_{q=1}^{\Gm} \DS{s}(\Ln_q,\Lm_q)$.
We claim each $\GfS_q$ is an order-$(s-1)$ DS sequence.  By virtue of being categorized as {\em first} in $\GS_q$,
every symbol in $\GfS_q$ appears at least once after $\GfS_q$.
Therefore an occurrence of an alternating sequence $\sigma_{s+1} = abab\cdots$ (length $s+1$),
in $\GfS_q$ would imply an occurrence of $\sigma_{s+2}$ in $S$, a contradiction.
By symmetry it also follows that $\GlS_q$ is an order-$(s-1)$ DS sequence, hence
$|\GfS|  = \sum_{q=1}^{\Gm-1} \DS{s-1}(\Gfn_q,\Lm_q)$ and
$|\GlS| = \sum_{q=2}^{\Gm} \DS{s-1}(\Gln_q,\Lm_q)$.  Since $\DS{s}$
is clearly superadditive\footnote{It is straightforward to show that
$\DS{s}(n',m') + \DS{s}(n'',m'') \le \DS{s}(n'+n'',m'+m''-1)$, for all $n',n'',m',m''$.}
we can bound these sums 
by $\DS{s-1}(\Gn,m-m_{\Gm})$ and $\DS{s-1}(\Gn,m-m_1)$.
(Note that $\sum_q \Gfn_q = \Gn$ and $\sum_q \Gln_q = \Gn$ as each sum counts 
each global symbol exactly once.)
The contribution of first and last symbols is therefore upper bounded by $2\cdot \DS{s-1}(\Gn,m)$.

The same argument shows that $\GmS_q$ is an order-$(s-2)$ DS sequence.
Symbols in $\GmS_q$ were categorized as {\em middle}, so an alternating subsequence
$\sigma_s = baba\cdots$ (length $s$) in $\GmS_q$, together with an $a$ preceding $\GmS_q$
and either an $a$ or $b$ following $\GmS_q$ (depending on whether $s$ is even or odd), 
yields an instance of $\sigma_{s+2}$ in $S$, a contradiction.
Thus the contribution of middle symbols
is
\begin{align}
|\GmS| &\le \sum_{q=2}^{\Gm-1} \DS{s-2}(\Gmn_q,m_q)\nonumber\\
&\le \DS{s-2}\paren{\sum_{q=2}^{\Gm-1} \Gmn_q, m-m_1-m_{\Gm}} & \mbox{\{superadditivity of $\DS{s-2}$\}}\nonumber\\
&\le \DS{s-2}(|\GS'|-2\Gn,m-m_1-m_{\Gm})\label{eqn:middle1}\\
&\le \DS{s-2}(\DS{s}(\Gn,\Gm)-2\Gn,m)\label{eqn:middle2}
\end{align}
Inequality (\ref{eqn:middle1}) follows from the fact that
$\sum_q \Gmn_q$ counts the length of $\GS'$, save the first
and last occurrence of each global symbol, that is, $2\Gn$ occurrences in total.
Since $\GS'$ is a subsequence of $S$, it too is an order-$s$ DS sequence,
so $|\GS'| \le \DS{s}(\Gn,\Gm)$.  Inequality (\ref{eqn:middle2}) follows.\\
\end{proof}

Recurrence~\ref{thm:recurrence-even} offers us the freedom to choose the block
partition $\{m_q\}_{1\le q\le \Gm}$ but it does not suggest what the optimal partition
might look like.  One natural starting place~\cite{HS86,ASS89,Nivasch10} is
to always choose $\Gm=2$, partitioning the sequence into 2 intervals each 
containing $m/2$ blocks.  This choice leads to $O(n + m\log^{s-2}m)$ upper bounds on $\DS{s}(n,m)$,
which is $O(n+m)$ if the alphabet/block density $n/m=\Omega(\log^{s-2}m)$.
Call this Analysis (1).
Given Analysis (1) we can conduct a stronger Analysis (2) by selecting $\Gm = m/\log^{s-2} m$,
so each interval contains $\log^{s-2} m$ blocks.  The $\DS{s}(\Gn,\Gm)$ term is bounded
via Analysis (1) (that is, $\DS{s}(\Gn,\Gm) = O(\Gn + \Gm\log^{s-2}\Gm) = O(\Gn + m)$) 
and the remaining terms bounded inductively via analysis Analysis (2).
This leads to bounds of the form $\DS{s}(n,m) = O(n + m\poly(\log^* m))$.
By iterating this process, Analysis ($i$) gives bounds of the form 
$O(n + m\poly(\log^{[i-1]} m))$.\footnote{(the $[i-1]$ here being short for $i-1$ $\star$s)}
We cannot conclude that $\DS{s}(n,m) = O(n+m)$ since
the constant hidden by the asymptotic notation, call it $\mu_{s,i}$, 
increases with $i$ and $s$.  

The discussion above is merely meant to foreshadow the analysis
of Recurrence~\ref{thm:recurrence-even} and subsequent 
Recurrences~\ref{lem:feather-rec}, \ref{thm:recurrence-odd}, \ref{rec:double-feather}, 
and \ref{rec:DS-double-feather}; see Appendices~\ref{appendix:agreeable} and \ref{sect:appendix:order5}.
We have made every attempt to segregate recurrences and structural arguments
from their quantitative analyses, which are important but nonetheless rote. 
As a consequence, Ackermann's function, its various inverses,
and quantities such as $\{\mu_{s,i}\}$ will be introduced as late as possible.

\subsection{The Evolution of Recurrence~\ref{thm:recurrence-even}}

The statement of Recurrence \ref{thm:recurrence-even} is simple, and
arguably cannot be made simpler.  We feel it is worthwhile to recount how it was 
assembled over the years in the 
works of~\cite{HS86,Sharir87,ASS89,Klazar99,Nivasch10}.

When $s$ is fixed the function $\DS{s}(n)$ depends only on one parameter, $n$,
a situation that would not ordinarily lead to expressions involving ``$\alpha$'', 
which is most naturally expressed as a function of two independent
parameters.\footnote{In graph algorithms these parameters typically correspond to nodes 
and edges~\cite{Tar79b,LengauerT79,Chaz00a},
in matrix problems~\cite{KK90,K92} to rows and columns,
and in data structures they may correspond to
elements and queries~\cite{Tar75,G85},
query time and preprocessing time~\cite{PetInvAck},
or input size and storage space~\cite{Yao82,AS87,CR91}.}
Hart and Sharir's~\cite{HS86} insight was to recognize an additional parameter
$m$ (the block count) and obtain bounds on $\DS{s}(n)$ via bounds on $\DS{s}(n,m)$.
See Lemma~\ref{lem:Sharir}.

Implicit in Hart and Sharir's analysis is a classification of symbols into local and global,
and of global occurrences into first, middle, and last.\footnote{This part of their analysis
is ostensibly about nodes and path compressions, not blocks and symbols.}
Agarwal, Sharir, and Shor~\cite{ASS89} made this local/global and 
first/middle/last classification
explicit, and arrived at a recurrence very close to Recurrence~\ref{thm:recurrence-even}.\footnote{Sharir~\cite{Sharir87}
split global occurrences into two categories---first and non-first---which leads to a near-linear upper bound 
of $\DS{s}(n) < n\cdot \alpha(n)^{O(\alpha(n))^{s-3}}$.}
However, they did not bound the contribution of global middle occurrences in the same way.
Whereas $\GmS_q$ is an $m_q$-block sequence,
it can be converted to 2-sparse one by removing up to $m_q-1$ repeated symbols
at block boundaries.
By Lemma~\ref{lem:Sharir}(\ref{item:Sharir1},\ref{item:Sharir2})
\[
|\GmS_q| < m_q + \DS{s-2}(\Gmn_q) \;\le\; m_q + \gamma_{s-2}(\Gmn_q)\cdot \Gmn_q \;\le\; m_q + \gamma_{s-2}(n)\cdot \Gmn_q
\]
In other words, when ``contracting'' $\GmS$ to form $\GmS'$, the shrinkage factor
is at most $\gamma_{s-2}(n)$.  A similar statement holds for first and last occurrences,
where the shrinkage factor is at most $\gamma_{s-1}(n)$.
This leads to a recurrence~\cite[p. 249]{ASS89} that {\em forgets} the role of $m$ when analyzing global  
occurrences.
\[
\DS{s}(n,m) \;\le\; \sum_{q=1}^{\Gm} \DS{s}(\Ln_q,m_q) \:+\: 2\cdot \gamma_{s-1}(n)\cdot n \:+\: \gamma_{s-2}(n)\cdot\DS{s}(\Gn,\Gm) \:+\: O(m)
\]
Nivasch's recurrence~\cite[Recurrence 3.1]{Nivasch10} improves that of Agarwal, Sharir, and Shor~\cite{ASS89} 
by not forgetting that $\GmS$ is an $m$-block sequence.
In particular, $|\GmS| \le \sum_q \DS{s-2}(\Gn_q,m_q)$ where $|\GmS'| < \sum_q \Gn_q \le \DS{s}(\Gn,\Gm)$.
Recurrence~\ref{thm:recurrence-even} is substantively no different than that of \cite{Nivasch10}
but it is more succinct, for two reasons.  First, the superadditivity of $\DS{s}$ lets us bound the number
of middle occurrences with the single term $\DS{s-2}(\DS{s}(\Gn,\Gm)-2\Gn,m)$.\footnote{One might think it would be dangerous
to bound middle occurrence with one aggregated term since we ``forget'' 
that $\GmS$ is partitioned into $\Gm-2$ order-$(s-2)$ DS sequences.  Doing this does not affect the solution
of $\DS{s}(n,m)$ asymptotically.}
Second, the function equivalent to $\DS{s}(n,m)$ from~\cite{ASS89,Nivasch10} is 
the extremal function of order-$s$ DS sequences that are {\em both} 2-sparse and have $m$ blocks.
This small change introduces $O(m)$ terms in~\cite[Recurrence 3.1]{Nivasch10} and \cite[p. 249]{ASS89}
since the derived sequences $\GS,\GS',$ and $\{\LS_q,\GfS_q,\GmS_q,\GlS_q\}_{1\le q\le \Gm}$ are not necessarily 2-sparse,
and must be made 2-sparse by removing $O(m)$ symbols at block boundaries.

Recurrence~\ref{thm:recurrence-even} 
could be made yet more succinct by removing the ``$-2\Gn$'' from the estimation of global middle occurrences.
This would not affect the solution asymptotically, but keeping it is essential for obtaining 
bounds on $\DS{3}(n)$ tight to the leading constant.

\section{Derivation Trees}\label{sect:derivation-tree}

A {\em derivation tree} $\Tree(S)$ for an $m$-block sequence $S$ is a rooted, ordered tree
whose nodes are identified with the blocks encountered in recursively decomposing $S$,
as in Section~\ref{sect:notation} and Recurrence~\ref{thm:recurrence-even}.
Let $\block(u)$
be the block associated with $u\in \Tree(S)$.  The leaf level of $\Tree(S)$ coincides
with $S$, that is, the $p$th leaf of $\Tree(S)$ holds the $p$th block of $S$.
As we are sometimes indifferent to the 
order of symbols within a block, $\block(v)$ is often treated as a set.
We assume without much loss in generality that no symbol appears just once in $S$.
As usual, we adopt the sequence decomposition notation from Section~\ref{sect:notation}.

\paragraph{Base Case.}
Suppose $S=\beta_1\beta_2$ is a two block sequence, where each block
contains the whole alphabet $\Sigma(S)$.
The tree $\Tree(S)$ consists of three nodes $u,u_1,$ and $u_2$, where $u$ is the parent of $u_1$ and $u_2$,
$\block(u_1)=\beta_1$, $\block(u_2)=\beta_2$, and $\block(u)$ does not exist.  For every $a\in\Sigma(S)$ call $u$
its {\em crown} and $u_1$ and $u_2$ its left and right {\em heads}, respectively.
These nodes are denoted $\crown_{|a}, \lefthead_{|a},$ and $\righthead_{|a}$.

\paragraph{Inductive Case.}
If $S$ contains $m>2$ blocks, choose an $\Gm < m$ and an arbitrary block partition $\{m_q\}_{1\le q\le \Gm}$.
Inductively construct derivation trees $\GTree = \Tree(\GS')$ and $\{\LTree_q\}_{1\le q\le \Gm}$, where 
$\LTree_q = \Tree(\LS_q)$, then identify the root of $\LTree_q$ (which has no block) with the $q$th leaf of $\GTree$.
Finally, place the blocks of $S$ at the leaves of $\Tree$.  This last step
is necessary since only local symbols appear in the blocks of $\{\LTree_q\}$ whereas
the leaves of $\Tree$ must be identified with the blocks of $S$.
Note that nodes at or above the leaf level of $\GTree$ carry only global symbols in their blocks
and that internal nodes in $\{\LTree_q\}$ carry only local symbols in their blocks. 
Local and global symbols only mingle at the leaf level of $\Tree$.

The crown and heads of each symbol $a\in \Sigma(S)$ are inherited from $\GTree$, if $a$ is global,
or some $\LTree_q$ if $a$ is local to $S_q$.  See Figure~\ref{fig:cr-lhe-rhe} for an illustration.

\begin{figure}
\centerline{\scalebox{.3}{\includegraphics{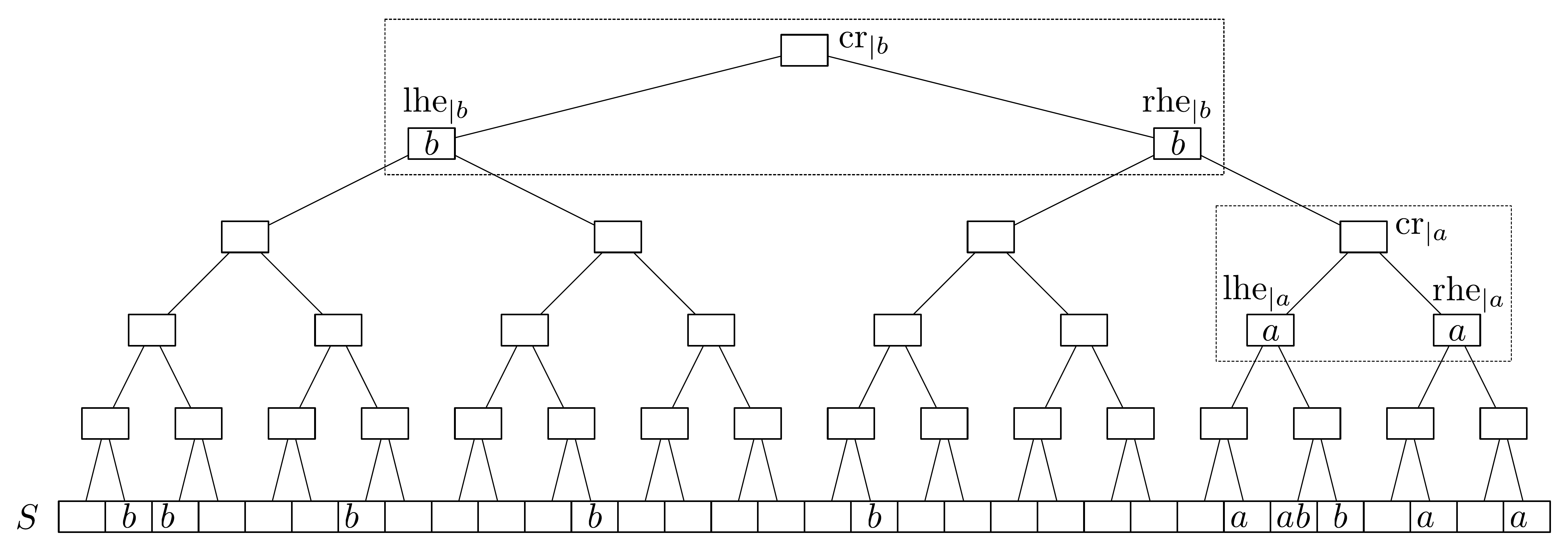}}}
\caption{\label{fig:cr-lhe-rhe}A derivation tree $\Tree(S)$ for a 32-block sequence
$S$.  The tree is generated by always choosing $\Gm=2$ and the uniform block partition 
$m_1=m_2=m/2$, where $m>2$ is the number of blocks in the given sequence.
The frames isolate the base case derivation trees
that assign the crown and heads for symbols $a,b\in\Sigma(S)$.}
\end{figure}

\subsection{Anatomy of the Tree}

The {\em projection of $\Tree$ onto $a\in\Sigma(S)$}, denoted $\Tree_{|a}$, is the tree
on the node set $\{\crown_{|a}\} \cup \{v \in \Tree \:|\: a\in\block(v)\}$ that inherits the ancestor/descendant
relation from $\Tree$, that is, the parent of $v$ in $\Tree_{|a}$, where $v\not\in \{\crown_{|a},\lefthead_{|a},\righthead_{|a}\}$, 
is $v$'s nearest strict ancestor $u$ for which $a\in \block(u)$.
For example, in Figure~\ref{fig:cr-lhe-rhe} $\Tree_{|a}$  consists
of $\crown_{|a}$, its children $\lefthead_{|a},\righthead_{|a}$, and four grandchildren at the leaf level of $\Tree$.

\begin{definition} (Anatomy)
\begin{itemize}
\item The leftmost and rightmost leaves of $\Tree_{|a}$ are {\em wingtips},
denoted $\lefttip_{|a}$ and $\righttip_{|a}$.
\item The left and right {\em wings}
are those paths in $\Tree_{|a}$ extending from $\lefthead_{|a}$ to $\lefttip_{|a}$ and from 
$\righthead_{|a}$ to $\righttip_{|a}$.
\item Descendants of $\lefthead_{|a}$ and $\righthead_{|a}$ in $\Tree_{|a}$ are called {\em doves} and {\em hawks}, respectively.
\item A child of a wing node that is not itself on the wing is called a {\em quill}.
\item A leaf is called a {\em feather} if it is the rightmost descendant of a dove quill or leftmost descendant of a hawk quill.
\item Suppose $v$ is a node in $\Tree_{|a}$.  
Let $\head_{|a}(v)$ be the head ancestral to $v$ and $\otherhead_{|a}(v)$ be the other head.
Let $\tip_{|a}(v)$ and $\othertip_{|a}(v)$ be the wingtips descending from $\head_{|a}(v)$ and $\otherhead_{|a}(v)$.
Let $\wing_{|a}(v)$ be the nearest wing node ancestor of $v$,
$\quill_{|a}(v)$ the quill ancestral to $v$,
and $\feather_{|a}(v)$ the feather descending from $\quill_{|a}(v)$.
See Figure~\ref{fig:feathers} for an illustration.
\end{itemize}
Once $a\in\Sigma(S)$ is known or specified, we will use these terms (feather, wingtip, etc.) to refer to nodes in $\Tree_{|a}$
or to the occurrences of $a$ within those blocks.  For example, an occurrence of $a$ in $S$ would 
be a feather if 
it appears in a block $\block(v)$ in $S$, where $v$ is a feather in $\Tree_{|a}$.
\end{definition}

Note that the nodes $\head_{|a}(v),\wing_{|a}(v),\quill_{|a}(v),\tip_{|a}(v),$ and $\feather_{|a}(v)$ are not
necessarily distinct.  It may be that $\head_{|a}(v)=\wing_{|a}(v)$,
and it may be that $v = \quill_{|a}(v) = \feather_{|a}(v)$ if $v$'s parent in $\Tree_{|a}$ is $\wing_{|a}(v)$.

\begin{figure}
\centering\scalebox{.3}{\includegraphics{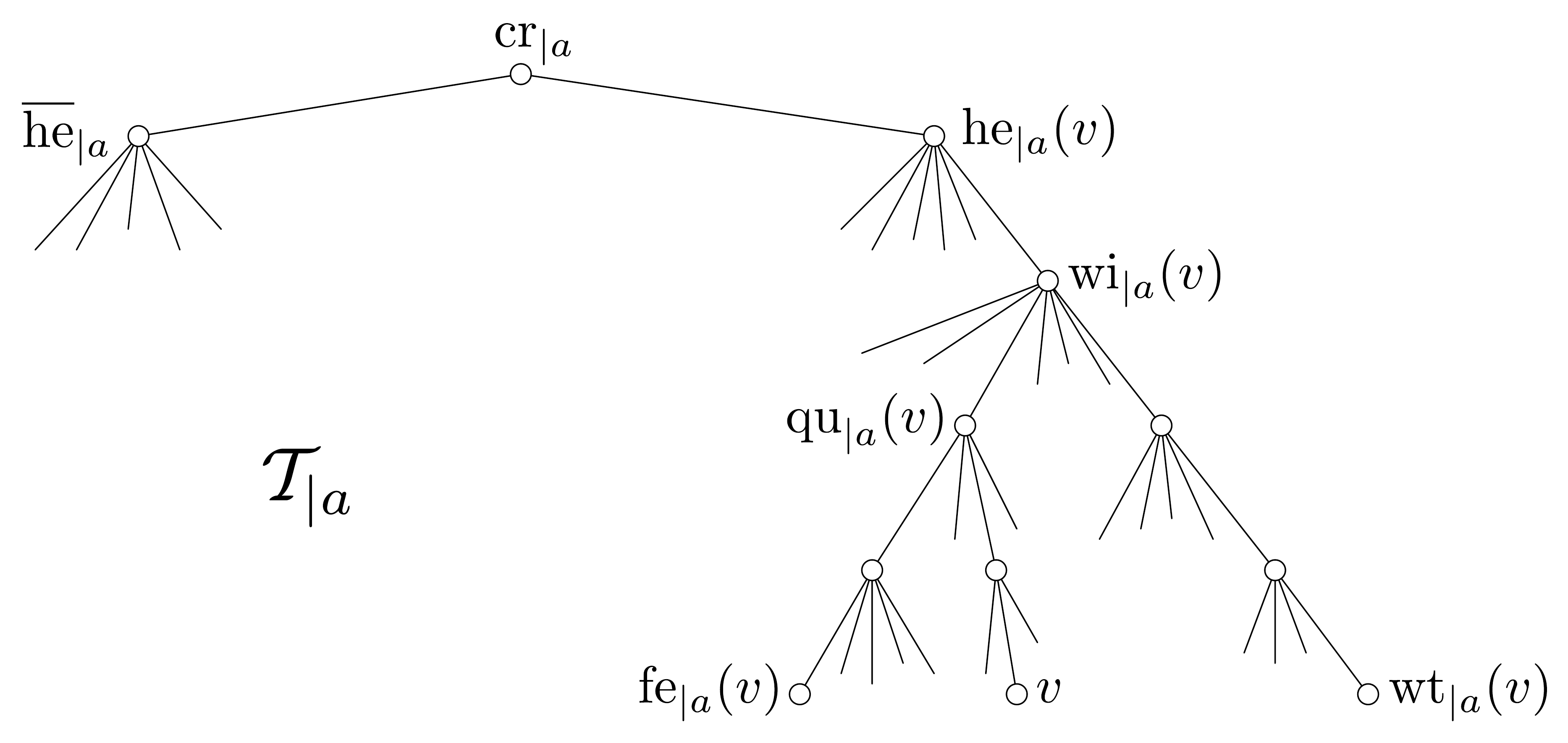}}
\caption{\label{fig:feathers} In this example $v$ is a hawk leaf in $\Tree_{|a}$ since its head $\head_{|a}(v) = \righthead_{|a}$ is the right child of $\crown_{|a}$.
Its wing node $\wing_{|a}(v)$, wingtip $\tip_{|a}(v)$, quill $\quill_{|a}(v)$, and feather $\feather_{|a}(v)$ are indicated.
}
\end{figure}

Lemma~\ref{lem:l-node} identifies one property of $\Tree$ used in the proof
of Lemma~\ref{lem:nested}.

\begin{lemma}\label{lem:l-node}
Suppose that on a leaf-to-root path in $\Tree$ we encounter
nodes $u,v,x,$ and $y$ (the last two possibly identical), 
where $u,x\in \Tree_{|a}$ and $v,y \in \Tree_{|b}$.
It must be that $a\in\block(v)$ and therefore $v\in\Tree_{|a}$.
\end{lemma}

\begin{proof}
Consider the decomposition of $\Tree$ into a global derivation tree $\GTree$ and local derivation trees $\{\LTree_q\}$.
If $v$ were an internal node in some $\LTree_q$ then $b$ would be classified
as local.  This implies $y\in \LTree_q$ as well and the claim follows by induction on the construction of $\LTree_q$.
If $v$ were an internal node in $\GTree$ then let $u'$ be the leaf of $\GTree$ ancestral to $u$.
The nodes $u',v,x,y\in\GTree$ also satisfy the criteria of the lemma; the claim follows by induction on the construction of $\GTree$.
Thus, we can assume $v$ is a leaf of $\GTree$ and $u$ is a leaf of $\Tree$.  See Figure~\ref{fig:uvxy}.
By construction all global symbols in $\block(u)$ also appear in $\block(v)$.
Since $x\in\GTree$, the symbol $a$ is classified as global and must appear in $\block(v)$.

\begin{figure}
\centering\scalebox{.32}{\includegraphics{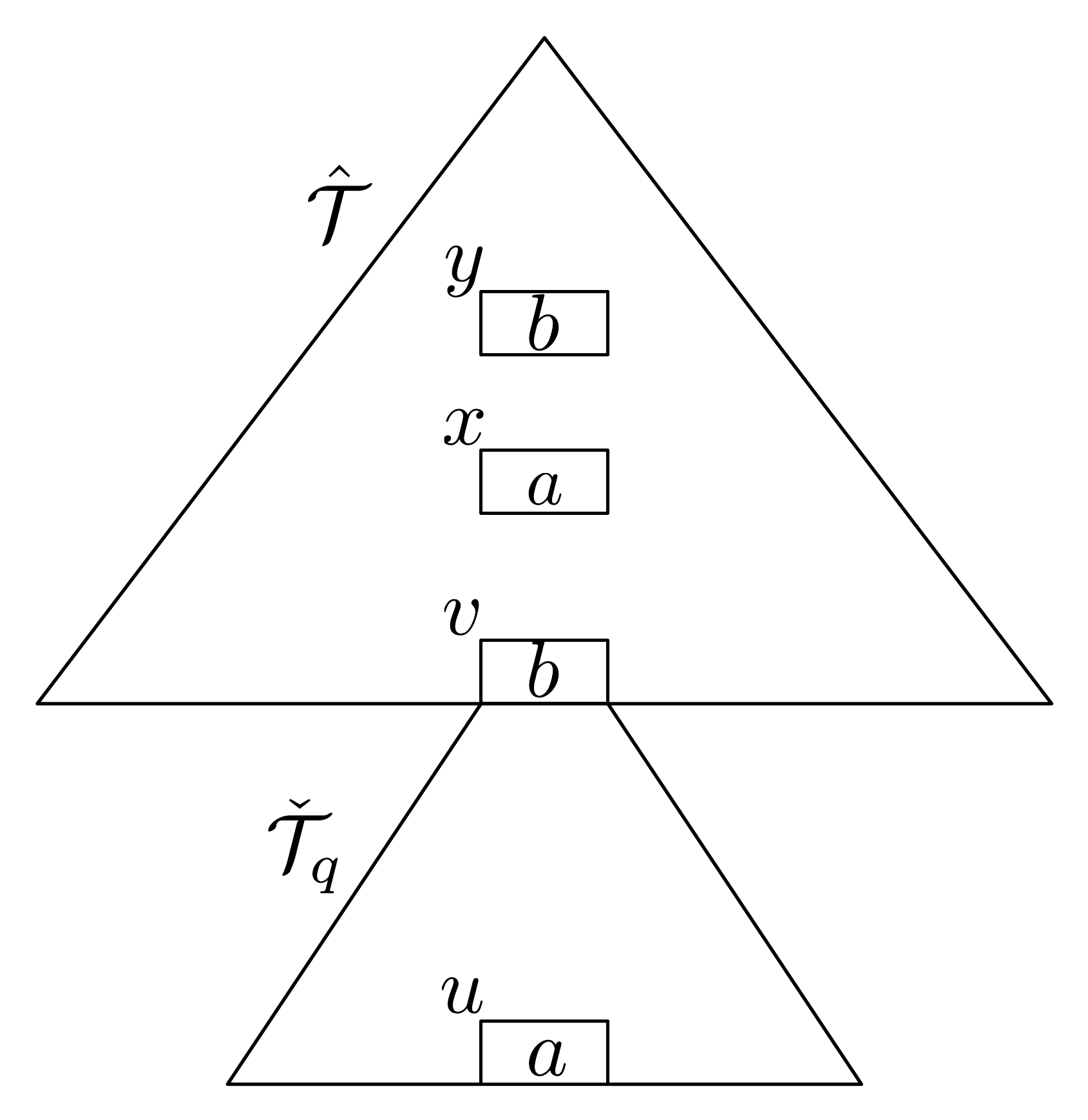}}
\caption{\label{fig:uvxy} The case where $v$ is a leaf of $\GTree$.  Both $x$ and $y$
are necessarily in $\GTree$, which implies that $a$ and $b$ are global, and further implies
that $u$ is a leaf of $\LTree_q$ since global symbols do not appear in the internal nodes
of $\LTree_q$.  All global symbols of $u$ also appear in $v$.
}
\end{figure}
\end{proof}

\subsection{Habitual Nesting}

Suppose a block $\beta$ in $S$ contains two symbols $a,b$ that are not wingtips, 
that is, they make neither their first nor last appearance in $\beta$.
We call $a$ and $b$ {\em nested in $\beta$} if $S$ contains either $ab\,\beta \, ba$
\ or \ $ba\, \beta\, ab$ and call them {\em interleaved in $\beta$} otherwise, that is,
if the occurrences of $a$ and $b$ in $S$ take the form $a^* b^* \,\beta\; a^* b^*$ or $b^*a^*\,\beta\; b^*a^*$.
Lemma~\ref{lem:nested} is the critical structural lemma used 
in our analysis.  It provides us with simple criteria for nestedness.

\begin{lemma}\label{lem:nested}
Suppose that $v\in\Tree(S)$ is a leaf and $a,b$ are symbols in a block $\block(v)$ of $S$.
If the following two criteria are satisfied then $a$ and $b$ are nested in $\block(v)$.
\renewcommand{\theenumi}{\roman{enumi}}
\begin{enumerate}
\item $v$ is not a wingtip in either $\Tree_{|a}$ or $\Tree_{|b}$.\label{criterion:middle}
\item $v$ is not a feather in either $\Tree_{|a}$ or $\Tree_{|b}$.\label{criterion:feather}
\end{enumerate}
\end{lemma}

\begin{proof}
Without loss of generality we can assert two additional criteria.
\renewcommand{\theenumi}{\roman{enumi}}
\begin{enumerate}
\setcounter{enumi}{2}
\item $\crown_{|b}$ is equal to or strictly ancestral to $\crown_{|a}$.\label{criterion:rootlevel}
\item $v$ is a dove in $\Tree_{|a}$.\label{criterion:dove}
\end{enumerate}
By Criterion (\ref{criterion:dove}) the leftmost leaf descendant of $\wing_{|a}(v)$ is $\tip_{|a}(v)$.  Let $u$
be its rightmost leaf descendant.
According to Criteria (\ref{criterion:middle},\ref{criterion:feather})
$v$ is distinct from both $\tip_{|a}(v)$ and $u$ since $u$ must be a feather.
We partition the sequence $S$ outside of $\block(v)$ into the following four intervals.
\begin{enumerate}
\item[$I_1$:] everything preceding the $a$ in $\block(\tip_{|a}(v))$,
\item[$I_2$:] everything from the end of $I_1$ to $\block(v)$,
\item[$I_3$:] everything from $\block(v)$ to the $a$ in $\block(u)$, and
\item[$I_4$:] everything following $I_3$.
\end{enumerate}
Since $v$ is not a wingtip of $\Tree_{|b}$, by Criterion (\ref{criterion:middle}), there must be occurrences of $b$ in $S$ both before and after
$\block(v)$.  If, contrary to the claim, $a$ and $b$ are not nested in $\block(v)$,
all other occurrences of $b$ must appear exclusively in $I_1$ and $I_3$ or exclusively
in $I_2$ and $I_4$.  We show that both possibilities lead to contradictions.
Figures~\ref{fig:I2I4} and \ref{fig:I1I3} illuminate the proof.

\paragraph{Case 1: $b$ does not appear in $I_1$ or $I_3$}
According to Criterion (\ref{criterion:middle}) the left wingtip $\lefttip_{|b}$ of $\Tree_{|b}$ is distinct from $v$,
and therefore appears in interval $I_2$.  Since $\lefttip_{|b}$ and $v$ are descendants
of $\wing_{|a}(v)$, which is a strict descendant of $\crown_{|a}$, which, by Criterion (\ref{criterion:rootlevel}),
is a descendant of $\crown_{|b}$, it must also be that $\lefttip_{|b}$ and $v$ descend
from the same child of $\crown_{|b}$, that is,
\renewcommand{\theenumi}{\roman{enumi}}
\begin{enumerate}
\setcounter{enumi}{4}
\item $v$ is a dove in $\Tree_{|b}$ and therefore $\tip_{|b}(v) = \lefttip_{|b}$.\label{criterion:dove-Tb}
\end{enumerate}
We shall argue below that
\begin{figure}
\centering\scalebox{.33}{\includegraphics{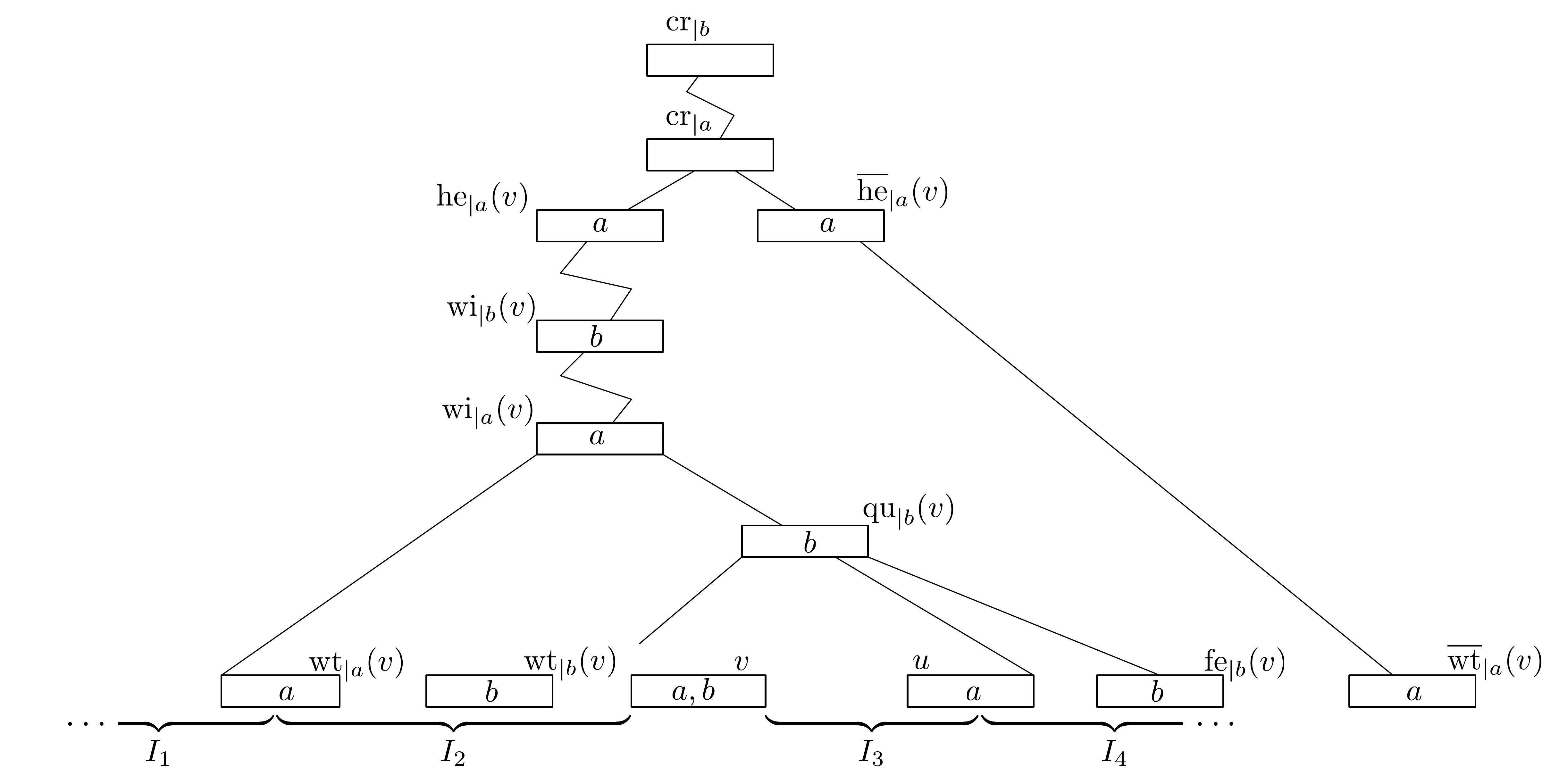}}
\caption{\label{fig:I2I4}Boxes represent nodes in $\Tree(S)$ and their associated blocks.
The blocks at the leaf-level correspond to those in $S$.
In Case 1 all occurrences of $b$ outside of $\block(v)$ appear in intervals $I_2$ and $I_4$.
Contrary to the depiction, it may be that $\crown_{|a}$ and $\crown_{|b}$ are identical, that 
$\tip_{|a}(v)$ and $\tip_{|b}(v)$ are identical,
that $u$ and $\feather_{|b}(v)$ are identical, and that $\wing_{|b}(v)$ is not a strict ancestor of $\wing_{|a}(v)$.
}
\end{figure}
\begin{figure}
\centering\scalebox{.35}{\includegraphics{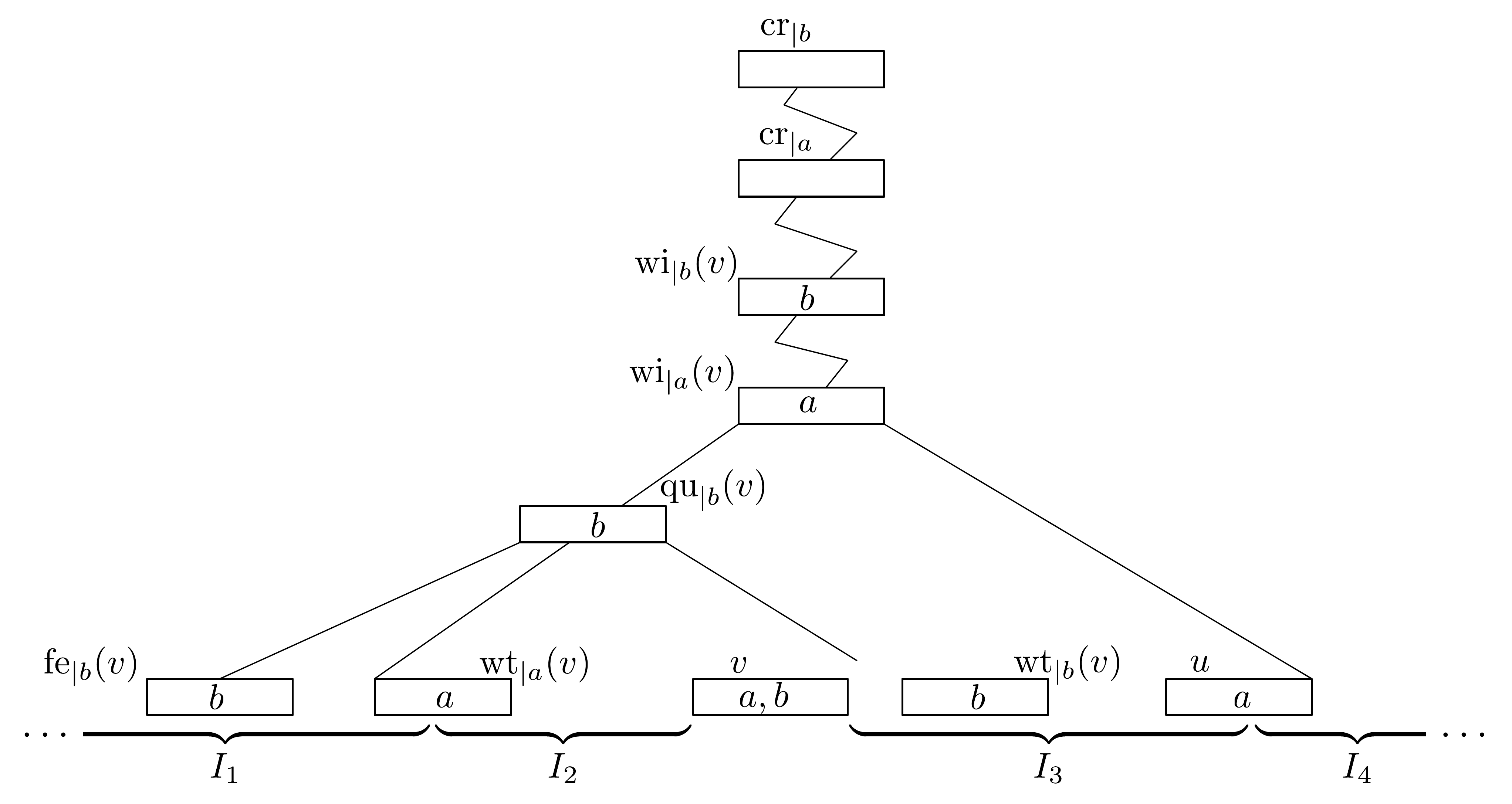}}
\caption{\label{fig:I1I3}In Case 2 all occurrences of $b$ outside of $\block(v)$ appear in intervals $I_1$ and $I_3$.
}
\end{figure}
\renewcommand{\theenumi}{\roman{enumi}}
\begin{enumerate}
\setcounter{enumi}{5}
\item In $\Tree$, $\quill_{|b}(v)$ is a strict descendant of $\wing_{|a}(v)$ and a strict ancestor of $u$,
and $\feather_{|b}(v)$ lies in interval $I_4$.\label{criterion:dovequill-Tb}
\end{enumerate}
The least common ancestor of $v$ and $\tip_{|b}(v)$ in $\Tree_{|b}$ is by definition $\wing_{|b}(v)$.
The quill $\quill_{|b}(v)$ is a child of $\wing_{|b}(v)$ not on a wing,
hence $\quill_{|b}(v)$ cannot be ancestral to $\tip_{|b}(v)$, and hence $\quill_{|b}(v)$ must be a strict descendant 
of $\wing_{|a}(v)$. 
By Criterion (\ref{criterion:feather}) and Inference (\ref{criterion:dove-Tb}), $\feather_{|b}(v)$ is the rightmost leaf descendant
of $\quill_{|b}(v)$ and distinct from $v$.  However, by supposition $I_3$ contains no occurrences of $b$,
so $\feather_{|b}(v)$ must lie in interval $I_4$.  
For $\quill_{|b}(v)$ to have descendants in both $I_2$ and $I_4$ it must be a strict ancestor
of $u$ in $\Tree$.
As we explain below, a consequence of Inference (\ref{criterion:dovequill-Tb}) is that
\renewcommand{\theenumi}{\roman{enumi}}
\begin{enumerate}
\setcounter{enumi}{6}
\item $\othertip_{|a}(v)$ lies to the right of $\feather_{|b}(v)$.\label{criterion:othertip}
\end{enumerate}
According to Inference (\ref{criterion:dovequill-Tb}) $\quill_{|b}(v)$ is a descendant of $\wing_{|a}(v)$,
which is a descendant of $\head_{|a}(v)$. According to Criterion (\ref{criterion:dove}) $\head_{|a}(v)$ is the
{\em left} head of $\Tree_{|a}$.  Since $\othertip_{|a}(v)$ is a descendant of $\otherhead_{|a}(v)$, the right sibling
of $\head_{|a}(v)$, $\othertip_{|a}(v)$ must lie to the right of $\feather_{|b}(v)$.

Let us review the situation.  Scanning the leaves from left to right we see the blocks $\tip_{|a}(v)$,
$\tip_{|b}(v)$,
$v$,
$u$,
$\feather_{|b}(v)$,
and
$\othertip_{|a}(v)$.
It may be that $\tip_{|a}(v)$ and $\tip_{|b}(v)$ are equal 
and it may be that $u$ and $\feather_{|b}(v)$ are equal.
If either of these cases hold then the $a$ 
precedes the $b$ in the given block.
The blocks $\tip_{|a}(v),\tip_{|b}(v),v,\feather_{|b}(v),\othertip_{|a}(v)$ certify that $a$ and $b$
are nested in $\block(v)$.

\paragraph{Case 2: $b$ does not appear in $I_2$ or $I_4$}
By Criterion (\ref{criterion:middle}) the right wingtip $\righttip_{|b}$ is distinct from $v$
and must therefore lie in $I_3$.  Following the same reasoning from Case 1 
we can deduce that 
\renewcommand{\theenumi}{\roman{enumi}}
\begin{enumerate}
\setcounter{enumi}{7}
\item $v$ is a hawk in $\Tree_{|b}$.\label{criterion:hawk-Tb}
\item In $\Tree$, $\quill_{|b}(v)$ is strict descendant of $\wing_{|a}(v)$ and a strict ancestor of $\tip_{|a}(v)$.\label{criterion:hawkquill-Tb}
\end{enumerate}
Inference (\ref{criterion:hawk-Tb}) follows since $v$ and $\righttip_{|b}$ must be descendants of
the same head in $\Tree_{|b}$.  This implies that $\feather_{|b}(v)$ is the {\em leftmost} leaf descendant
of $\quill_{|b}(v)$.  Since $\feather_{|b}(v)$ is distinct from $v$ and interval $I_2$ is free of $b$s,
it must be that $\feather_{|b}(v)$ lies in $I_1$ and that 
$\quill_{|b}(v)$ is a strict descendant of $\wing_{|a}(v)$ and a strict ancestor of $\tip_{|a}(v)$.
Inference (\ref{criterion:hawkquill-Tb}) follows.  See Figure~\ref{fig:I1I3}.

It follows from Criterion (\ref{criterion:rootlevel}) and Inference (\ref{criterion:hawkquill-Tb})
that on a leaf-to-root path one encounters the nodes $\tip_{|a}(v)$, $\quill_{|b}(v)$, $\wing_{|a}(v)$, and $\crown_{|b}$, in that order.
Lemma~\ref{lem:l-node} implies that $a\in\block(\quill_{|b}(v))$.
We have deduced that
$\quill_{|b}(v)$ is in $\Tree_{|a}$, is a strict descendant
of $\wing_{|a}(v)$, and is a common ancestor of $\tip_{|a}(v)$ and $v$.  
This contradicts the fact that $\wing_{|a}(v)$ is the {\em least} common ancestor of $v$ and $\tip_{|a}(v)$ in $\Tree_{|a}$.
\end{proof}

Note that Lemma~\ref{lem:nested} applies to any blocked sequence and an associated derivation tree.
It has nothing to do with Davenport-Schinzel sequences as such.

\section{A Recurrence for Odd Orders}\label{sect:odd}

Lemma~\ref{lem:nested} may be rephrased as follows.  Every blocked sequence $S$ 
is the union of four sequences:
two comprising wingtips (first occurrences and last occurrences, each of length $n$), 
one comprising all feathers,
and one comprising non-wingtip non-feathers.
The last sequence is distinguished by the property that each pair of symbols in any
block is nested with respect to $S$, 
which is a ``good'' thing if we are intent on giving strong upper bounds on odd-order sequences.
The sequence comprising feathers is ``bad'' in this sense, therefore we must obtain better-than-trivial 
upper bounds on its length if this strategy is to bear fruit.

Recall that {\em feather} is a term that can be applied to nodes in some $\Tree_{|b}$
or the corresponding {\em occurrences} of $b$ in the given sequence $S$.
This definition is with respect to one 
derivation tree $\Tree$ for $S$, which is not necessarily the best one.
In Recurrences~\ref{lem:feather-rec} and \ref{thm:recurrence-odd} it is useful
to reason about the {\em optimal} derivation tree.
Let $\Tree^*(S)$ be the derivation tree for $S$ that minimizes the number of occurrences in $S$ classified
as feathers.

\begin{recurrence}\label{lem:feather-rec}
Define $\Feather{s}(n,m)$ to be the maximum number of feathers in any
order-$s$, $m$-block DS sequence $S$ over an $n$-letter alphabet,
with respect to the optimal derivation tree $\Tree^*(S)$.
When $m\le 2$ we have $\Feather{s}(n,m)=0$.
For any $\Gm<m$, any block partition $\{m_q\}_{1\le q\le \Gm}$,
and any alphabet partition $\{\Gn\}\cup \{\Ln_q\}_{1\le q\le \Gm}$, we have
\begin{align*}
\Feather{s}(n,m)
&\le\; \sum_{q=1}^{\Gm} \Feather{s}(\Ln_q,m_q) 
					\,+\, \Feather{s}(\Gn,\Gm) \,+\, 2\cdot\DS{s-1}(\Gn,m) 
\end{align*}
\end{recurrence}

\begin{proof}
When $m\le 2$, $\Feather{s}(n,m)$ is trivially 0 since 
every occurrence in $S$ is a wingtip, and feathers are not wingtips.
When $m>2$ we shall decompose $S$ as in Section~\ref{sect:notation}.
The choice of $\Gm$ and the block partition $\{m_q\}_{1\le q\le \Gm}$ are not necessarily
those of the optimal derivation tree, but we do not need them to be.  
We are only interested in an upper bound on $\Feather{s}(n,m)$.
Let $\GTree^*$ and $\{\LTree_q^*\}_{1\le q\le \Gm}$ be the optimal
derivation trees for $\GS'$ and $\{\LS_q\}_{1\le q\le \Gm}$, and let $\Tree$ be
their composition, with the blocks of $S$ placed at $\Tree$'s leaves.

The number of occurrences of local feathers with respect to $\{\LTree^*_q\}$
is at most $\sum_q \Feather{s}(\Ln_q,m_q)$.
An occurrence of $a\in\block(v)$ in $\GS$ will be a dove feather  in $\Tree$
if either
(i) $v$ is the rightmost child of a dove feather in $\GTree^*_{|a}$
or 
(ii) $v$ is a non-wingtip child of the left wingtip in $\GTree^*_{|a}$,
which corresponds to an occurrence of $a$ in $\GfS$.
The same statement is true of hawk feathers, swapping the roles of left and right
and substituting $\GlS$ for $\GfS$.
There are at most $\Feather{s}(\Gn,\Gm)$ feathers of 
type (i) and, since $\GfS$ and $\GlS$ are order-$(s-1)$ DS sequences, less than 
$2\cdot \DS{s-1}(\Gn,m)$ of type (ii).
\end{proof}

We now have all the elements in place to provide a recurrence for odd-order Davenport-Schinzel sequences.

\begin{recurrence}\label{thm:recurrence-odd}
Let $m,n,$ and $s$ be the block count, alphabet size, and order parameters, where $s\ge 5$ is odd.
For any $\Gm<m$, any block partition $\{m_q\}_{1\le q\le \Gm}$, 
and any alphabet partition $\{\Gn\}\cup \{\Ln_q\}_{1\le q\le \Gm}$, we have
\begin{align*}
\DS{s}(n,m) &\leq \sum_{q=1}^{\Gm} \DS{s}(\Ln_q,m_q) \,+\, 2\cdot \DS{s-1}(\Gn,m) 
 \,+\, \DS{s-2}(\Feather{s}(\Gn,\Gm),m) \,+\, \DS{s-3}(\DS{s}(\Gn,\Gm), m)
\end{align*}
\end{recurrence}

\begin{proof}
As always, we adopt the notation from Section~\ref{sect:notation}.
Define $\GTree^*, \{\LTree^*_q\},$ and $\Tree$ as in the proof of Recurrence~\ref{lem:feather-rec}.
In Recurrence~\ref{thm:recurrence-even} we partitioned $S$ into local and global symbols and partitioned the occurrences 
of global symbols into first, middle, and last.  We now partition the middle occurrences one step further.
Define $\GfeatherS'$ and $\GnonfeatherS'$ to be the subsequences of $\GS'$ 
consisting of feathers (according to $\GTree^*$) and non-feather, non-wingtips, respectively.
That is, $|\GS'| = |\GfeatherS'| + |\GnonfeatherS'| + 2\Gn$.
In an analogous fashion define $\GfeatherS$ and $\GnonfeatherS$
to be the subsequences of $\GS$ consisting of children of occurrences in 
$\GfeatherS'$ and $\GnonfeatherS'$.
The sequences $\GfS$ and $\GlS$ represent the children of dove and hawk wingtips in $\GTree^*$.
Thus, $|S| = \sum_{q} |\LS_q| + |\GfS| + |\GlS| + |\GfeatherS| + |\GnonfeatherS|$.

The local sequences $\{\LS_q\}$ are order-$s$ DS sequences.  
According to the standard argument $\GfS$ and $\GlS$ are order-$(s-1)$ DS sequences
and $\GfeatherS = \GfeatherS_1\cdots\GfeatherS_{\Gm}$ is obtained from $\GfeatherS'$
by substituting for its $q$th block an order-$(s-2)$ DS sequence $\GfeatherS_q$.
From the superadditivity of $\DS{s-2}$ it follows that $|\GfeatherS| \le \DS{s-2}(|\GfeatherS'|,m) \le \DS{s-2}(\Feather{s}(\Gn,\Gm), m)$.

We claim that $\GnonfeatherS=\GnonfeatherS_1\cdots\GnonfeatherS_{\Gm}$
is obtained from $\GnonfeatherS'$ by substituting for its $q$th block an order-$(s-3)$ DS sequence $\GnonfeatherS_q$,
which, if true, would imply that $|\GnonfeatherS| \le \DS{s-3}(|\GnonfeatherS'|,m) < \DS{s-3}(\DS{s}(\Gn,\Gm),m)$.
Suppose for the purpose of obtaining a contradiction that 
the $q$th block $\beta$ in $\GnonfeatherS'$ contains $a,b\in\GSigma$,
and that $\GnonfeatherS_q$ 
is not an order-$(s-3)$ DS sequence, that is, it contains 
an alternating subsequence $ab\cdots ab$ of length $s-1$.
Note that $s-1$ is even.   By definition $\beta$ is a non-feather, non-wingtip 
in both $\GTree^*_{|a}$ and $\GTree^*_{|b}$.
According to Lemma~\ref{lem:nested}, $a$ and $b$ must be nested in $\beta$, which implies that $S$ contains a subsequence
of the form
\begin{align*}
&a \cdots b \cdots \left| \cdots \overbrace{a \cdots b \cdots a \cdots b}^{s-1} \cdots \right| \cdots b \cdots a\\
\mbox{ or } \; \; \; \; \; \; \; &
 b \cdots a \cdots \left| \cdots \overbrace{a \cdots b \cdots a \cdots b}^{s-1} \cdots \right| \cdots a \cdots b
\end{align*}
where the portion between bars is in $S_q$.
In either case $S$ contains an alternating subsequence 
with length $s+2$, contradicting the fact that $S$ is an order-$s$ DS sequence.
\end{proof}

\subsection{Analysis of the Recurrences}\label{sect:analysis}

The dependencies between $\DS{}$ and $\Phi$ established by 
Recurrences~\ref{thm:recurrence-even}, \ref{lem:feather-rec}, and \ref{thm:recurrence-odd}
are rather intricate.  For even $s$, $\DS{s}$ is a function of $\DS{s}, \DS{s-1}$ and $\DS{s-2}$,
and for odd $s$, $\DS{s}$ is a function of $\DS{s}, \DS{s-1},\DS{s-2},\DS{s-3}$, and $\Feather{s}$
while $\Feather{s}$ is a function of $\Feather{s}$ and $\DS{s-1}$.

The proof of Lemma~\ref{lem:agreeable} is by induction over parameters: $s,n,c,i,$ and $j$, where $s$ is the order,
$n$ the alphabet size, $c\ge s-2$ a constant that determines how $\Gm$ and the block partition is chosen,
$i\ge 1$ is an integer,
and $j$ is minimal such that the block count $m\le a_{i,j}^c$.
Some level of complexity is therefore unavoidable.  Furthermore,
when $s\ge 5$ is odd, $\DS{s}$ is so sensitive to approximations of $\DS{s-3}$ that we must 
treat $s\in\{1,2,3,4,5\}$ as distinct base cases, and treat even and odd $s\ge 6$ as separate inductive cases. 
Given these constraints we feel our analysis is reasonably simple.

\begin{lemma}\label{lem:agreeable}
Let $s\ge 1$ be the order parameter, $c\ge s-2$ be a constant,
and $i\ge 1$ be an arbitrary integer.
The following upper bounds on
$\DS{s}$ and $\Feather{s}$
hold for all $s\ge 1$ and all odd $s\ge 5$, respectively.
Define $j$ to be maximum such that $m\le a_{i,j}^c$.
\begin{align*}
\DS{1}(n,m) &= n + m-1						& \mbox{$s=1$}\\
\DS{2}(n,m) &= 2n + m-2						& \mbox{$s=2$}\\
\DS{3}(n,m) &\le (2i+2)n + (3i-2)cj(m-1)			& \mbox{$s=3$}\\
\DS{s}(n,m) &\le \M{s}{i}\paren{n + (cj)^{s-2}(m-1)}		& \mbox{all $s\ge 4$}	\\ 
\Feather{s}(n,m) &\le \N{s}{i}\paren{n + (cj)^{s-2}(m-1)}	& \mbox{odd $s\ge 5$}
\intertext{The values $\{\M{s}{i},\N{s}{i}\}$ are defined as follows, where $t=\floor{\frac{s-2}{2}}$.}
\istrut{5}\M{s}{i} 	&= \left\{
	\begin{array}{l}
		2^{{i+t+3}\choose t}	- 3(2(i+t+1))^t			\\
		\istrut{6}\fr{3}{2}(2(i+t+1))^{t+1}2^{{i+t+3}\choose t}
	\end{array}
	\right.
	&
	\begin{array}{r}
		\mbox{even $s\ge 4$\hcm[-.18]}\\
		\istrut{6}\mbox{odd $s\ge 5$\hcm[-.18]}
	\end{array}\\
\istrut{8}
\N{s}{i} &= 4\cdot 2^{{i+t+3}\choose t}		
	&
	\begin{array}{r}
		\mbox{odd $s\ge 5$\hcm[-.18]}
	\end{array}
\end{align*}
\end{lemma}

One may want to keep in mind that we will eventually substitute $\alpha(n,m)+ O(1)$ 
for the parameter $i$, and that ${i+t+O(1) \choose t} = i^t/t! + O(i^{t-1})$.
Lemma~\ref{lem:agreeable} will, therefore, imply bounds on $\DS{s}(n,m)$
analogous to those claimed for $\DS{s}(n)$ in Theorem~\ref{thm:main}.
The proof of Lemma~\ref{lem:agreeable} appears in Appendix~\ref{appendix:agreeable}.

\subsection{The Upper Bounds of Theorem~\ref{thm:main}}\label{sect:proofthmmain}

Fix $s\ge 3, n,m$ and let $c=s-2$.
For $i \ge 1$ let $j_i$ be minimum such that $m\le a_{i,j_i}^{c}$.
Lemma~\ref{lem:agreeable} implies that an order-$s$ DS sequence has length 
at most $\M{s}{i}(n + (cj_i)^{s-2} m)$.  
Choose $\iota$ to be minimum such that\footnote{We want $(cj_{\iota})^{s-2}m$
not to be the dominant term, so $(cj_{\iota})^{s-2}$ should be less than $\ceil{n/m}$.  On the other hand,
the first and second columns of Ackermann's function ($a_{i,1}$ and $a_{i,2}$) do not exhibit
sufficient growth, so $j_\iota$ must also be at least 3.}
$(cj_\iota)^{s-2} \le \max\{\f{n}{m}, (c\cdot 3)^{s-2}\}$.
One can show that $\iota = \alpha(n,m) + O(1)$.
By choice of $\iota$ it follows that $(cj_\iota)^{s-2} m = O(m+n)$, so 
$\DS{s}(n,m) = O((n+m)\M{s}{\iota})$.
According to Lemma~\ref{lem:agreeable}'s definition of $\M{s}{\iota}$, we  have
\begin{align*}
\DS{3}(n,m) &= O((n+m)\alpha(n,m))\\
\DS{4}(n,m) &= O\paren{(n+m)2^{\alpha(n,m)}}\\
\DS{5}(n,m) &= O\paren{(n+m)\alpha^2(n,m)2^{\alpha(n,m)}}\\
\DS{s}(n,m) &= (n+m)\cdot 2^{\alpha^t(n,m)/t! \,+\, O(\alpha^{t-1}(n,m))}	& \mbox{both even and odd $s\ge 6$, where $t=\ceil{\frac{s-2}{2}}$.}
\end{align*}
The bound on $\DS{5}(n,m)$ follows since
$\M{5}{\iota} = O(\iota^2 2^{\iota})$.
When $s\ge 6$ and $t\ge 2$,
$\M{s}{\iota} < \iota^{t+1}2^{\iota + t+ O(1) \choose t} = 2^{\iota^t/t! \,+\, O(\iota^{t-1})}$.

Theorem~\ref{thm:main} stated bounds on $\DS{s}(n)$ rather than $\DS{s}(n,m)$.
If it were known that extremal order-$s$ DS sequences consisted of $m=O(n)$ blocks we could
simply substitute $\alpha(n)$ for $\alpha(n,m)$ in the bounds above, but this is not known to be true.
According to Lemma~\ref{lem:Sharir}(\ref{item:Sharir2},\ref{item:Sharir4}), if $\gamma_s$ is such that
$\DS{s}(n)\le\gamma_{s}(n)\cdot n$ then $\DS{s}(n) \le \gamma_{s-2}(n)\cdot \DS{s}(n,2n-1)$
and $\DS{s}(n) \le \gamma_{s-2}(\gamma_s(n))\cdot \DS{s}(n,3n-1)$.
Applying Lemma~\ref{lem:Sharir} when $s\in\{3,4\}$ has no asymptotic affect since $\gamma_1=1$ and $\gamma_2=2$.
It has no {\em perceptible} effect when $s\ge 6$ since 
$\gamma_{s-2}(n)$ or $\gamma_{s-2}(\gamma_s(n))$ is dwarfed by the lower order terms in the exponent.
However, for $s\in\{3,5\}$ these reductions only show that
$\DS{3}(n) = O(n\alpha(n))$
and that
$\DS{5}(n) = O(n\alpha(\alpha(n))\alpha^2(n)2^{\alpha(n)})$,
which are weaker than the bounds
claimed in Theorem~\ref{thm:main}.

In Section~\ref{sect:order-3-ub} we prove $\DS{3}(n) = 2n\alpha(n) + O(n)$, which is a tiny improvement
over Klazar's bound~\cite{Klazar99,Nivasch10}, though it is within $O(n)$ of Nivasch's construction~\cite{Nivasch10}
and is therefore optimal in the Ackermann-invariant
sense.  See Remark~\ref{remark:Ackermann-invariance}.
To prove $\DS{5}(n) = \Theta(n\alpha(n)2^{\alpha(n)})$ we require a significant generalization of the derivation tree
method.
Sections~\ref{sect:lb} and \ref{sect:order5-ub} give the matching lower and upper bounds on order-$5$ DS sequences.

\subsubsection{Order $s=3$}\label{sect:order-3-ub}

Let $S$ be an order-$3$ DS sequence over an $n$-letter alphabet.  According 
to Lemma~\ref{lem:Sharir} (\cite{HS86}), $|S| \le \DS{3}(n) \le \DS{3}(n,m)$, where $m=2n-1$.  Letting $\iota$ be minimum such 
that $m \le a_{\iota,3}$, Lemma~\ref{lem:agreeable} implies that 
$\DS{3}(n,m) < (2 \iota+2)n + (3\iota-2)m < (8\iota -2)n$.
It is straightforward to show that $\iota \le \alpha(n)+O(1)$.
The problem is clearly that there are too many blocks.
Were there less than $(2n-1)/\iota$ blocks, Lemma~\ref{lem:agreeable} would
give a bound of $(2\iota+2)n + O(\iota m/\iota)=2n\alpha(n) + O(n)$.  
We can invoke Recurrence~\ref{thm:recurrence-even}
to divide $S$ into a global $\GS$ and local $\LS=\LS_1\cdots\LS_{\Gm}$, where $\Gm = m/\iota \le (2n-1)/\iota$,
that is, each $\LS_q$ is an $\iota$-block sequence.
Using Lemma~\ref{lem:agreeable} we will bound $\GS$ with $i=\iota$ 
and each of the $\{\LS_q\}_q$ with $i=1$.
\begin{align*}
|S| &\le \DS{3}(n) \le \DS{3}(n,m) & \mbox{\{where $m=2n-1$\}}\\
&\le \sum_{q=1}^{\Gm} \DS{3}(\Ln_q,\iota) \;+\; 2\cdot \DS{2}(\Gn,m) \;+\; \DS{1}(\DS{3}(\Gn,\Gm)-2\Gn,m) & \mbox{\{Recurrence~\ref{thm:recurrence-even}\}}\\
&< \sum_{q=1}^{\Gm} \SqBrack{4\Ln_q \;+\; \min\Bigcurly{\iota\ceil{\log \iota}, \;\, (\iota-1) + (2\Ln_q-1)\ceil{\log(2\Ln_q-1)}}}     & \mbox{(*)}\\
&\hcm +\; [4\Gn + 2m] \;+\; [2\iota\Gn \,+\, (3\iota-2)\Gm \,+\, m]                                                          & \mbox{\{Lemmas~\ref{lem:DS12}, \ref{lem:agreeable}\}}\\
&< \SqBrack{m + (n-\Gn)(4 + 2\ceil{\log (2\iota-1)})} + (2\iota+4)\Gn + (3\iota-2)m/\iota + 3m    & \mbox{\{$\Gm = m/\iota\}$}\\
&< (2\iota+4)n + 7m  	& \mbox{\{worst case if $\Gn=n$\}}\\
&\le 2n\alpha(n) + O(n)   & \mbox{\{$\iota = \alpha(n)+O(1)$\}}
\end{align*}

The bound on local symbols in line (*) follows from Lemma~\ref{lem:agreeable} and Hart and Sharir's~\cite{HS86}
observation that $\DS{3}(n)\le \DS{3}(n,2n-1)$.  When $i=1$ and $j=\ceil{\log\iota}$, Lemma~\ref{lem:agreeable}  gives us 
a bound of $\DS{3}(\Ln_q,\iota) \le 4\Ln_q + \iota\ceil{\log \iota}$.
Alternatively, we could make $\LS_q$ 2-sparse by removing up to $\iota-1$ duplicated
symbols at block boundaries, then
partitioning the remaining sequence into $2\Ln_q-1$ blocks, hence $\DS{3}(\Ln_q,\iota) \le \iota-1 + \DS{3}(\Ln_q,2\Ln_q-1) \le 
\iota-1 + 4\Ln_q + (2\Ln_q-1)\ceil{\log(2\Ln_q-1)}$.
This matches Nivasch's~lower bound~\cite{Nivasch10} on $\DS{3}(n)$ to within $O(n)$.

\section{Lower Bounds on Fifth-Order Sequences}\label{sect:lb}

We have established every bound claimed in Theorem~\ref{thm:main} except for those on order-$5$ DS sequences. 
In this section
we give a construction that yields bounds of
$\DS{5}(n,m) = \Omega(n\alpha(n,m)2^{\alpha(n,m)})$ and $\DS{5}(n) = \Omega(n\alpha(n)2^{\alpha(n)})$.
This is the first construction that is asymptotically longer than the order-$4$ DS sequences of~\cite{ASS89}
having length $\Theta(n2^{\alpha(n)})$.
Our construction is based on generalized forms of sequence {\em composition} and {\em shuffling} used by
Agarwal, Sharir, and Shor~\cite{ASS89}, Nivasch~\cite{Nivasch10}, and Pettie~\cite{Pettie-GenDS11}.

Recall from Section~\ref{sect:notation-and-terminology} that $\|S\|=|\Sigma(S)|$ is the alphabet
size of $S$ and, if $S$ is partitioned into blocks, $\bl{S}$ is its block count.

\subsection{Composition and Shuffling}\label{sect:comp-and-shuffle}

In its generic form, a sequence $S$ is assumed to be over the alphabet $\{1,\ldots,\|S\|\}$, that is, any totally ordered set with size $\|S\|$.  
To {\em substitute} $S$ for a block $\beta = [a_1\ldots a_{|\beta|}]$ means to replace $\beta$ with a copy $S(\beta)$ under the
alphabet mapping $k \mapsto a_k$, where $|\beta| \le \|S\|$.  If $|\beta|$ is strictly smaller than $\|S\|$, any occurrences of 
the $\|S\|-|\beta|$ unused symbols of $\Sigma(S)$ do not appear in $S(\beta)$.
We always assume that $S$ is in {\em canonical form}: the symbols are ordered according to the position of 
their first appearance in $S$.

\paragraph{Composition.}
If $\Smid$ is a sequence in canonical form with $\|\Smid\| = j$ and $\Stop$ a sequence partitioned into blocks with length at most $j$,
$\Ssub = \Stop \compose \Smid$ is obtained by substituting for each block $\beta$ in $\Stop$ a copy $\Smid(\beta)$.
Clearly $\bl{\Ssub} = \bl{\Stop}\cdot\bl{\Smid}$.  If $\Smid$ and $\Stop$ contain $\mu$ and $\mu'$ occurrences of each symbol, respectively, then $\Ssub$  contains $\mu\mu'$ occurrences of each symbol.  Composition preserves canonical form, that is, if $\Smid$ and $\Stop$
are in canonical form, so is $\Ssub$.

\paragraph{Shuffling.}
If $\Sbot$ is a $j'$-block sequence and $\Ssub$ is partitioned into blocks of length at most $j'$, we can form the 
{\em shuffle} $\Ssh = \Ssub \shuffle \Sbot$ as follows.
First create a sequence $\Sbot^*$
consisting of the concatenation of $\bl{\Ssub}$ copies of $\Sbot$, each copy being over an alphabet disjoint
from the other copies and disjoint from that of $\Ssub$.  
By design the length of $\Ssub$ is at most the number of blocks in $\Sbot^*$, and precisely the same 
if all blocks in $\Ssub$ have their maximum length $j'$.
The sequence $\Ssub \shuffle \Sbot$ is obtained by shuffling the $j'$ symbols of the $l$th block of $\Ssub$ into the $j'$ blocks of the $l$th copy
of $\Sbot$ in $\Sbot^*$.  Specifically, the $k$th symbol of the $l$th block is inserted at the {\em end} of the $k$th block of the $l$th copy of $\Sbot$.

\paragraph{Three-Fold Composition.}
Our construction of order-$5$ DS sequences uses a generalized form of composition that treats symbols in $\beta$ differently
based on context.  
Suppose $\Stop$ is partitioned into blocks with length at most $j$
and $\Smid^{\operatorname{f}},\Smid^{\operatorname{m}},$ and $\Smid^{\operatorname{l}}$ are sequences with 
alphabet size $\|\Smid^{\operatorname{f}}\|=\|\Smid^{\operatorname{m}}\|=\|\Smid^{\operatorname{l}}\|=j$.
The 3-fold composition $\Stop \compose \angbrack{\Smid^{\operatorname{f}},\Smid^{\operatorname{m}},\Smid^{\operatorname{l}}}$ is formed
as follows.  
For each block $\beta$ in $\Stop$, categorize its symbols 
as {\em first} if they occur in no earlier block, {\em last} if they occur in no later block,
and {\em middle} otherwise.
Let $\beta^{\operatorname{f}},\beta^{\operatorname{m}},$ and $\beta^{\operatorname{l}}$ be the subsequences
of $\beta$ consisting of first, middle, and last symbols.
Note that these three sequences do not necessarily occur contiguously in $\beta$, but each is nonetheless 
a subsequence of $\beta$.
Substitute for $\beta$ the concatenation of $\Smid^{\operatorname{f}}(\beta^{\operatorname{f}}), \Smid^{\operatorname{m}}(\beta^{\operatorname{m}}),$ and $\Smid^{\operatorname{l}}(\beta^{\operatorname{l}})$.
Note that if $\Stop,\Smid^{\operatorname{f}},\Smid^{\operatorname{m}},$ and $\Smid^{\operatorname{l}}$ contain 
$\mu \ge 2,\mu^{\operatorname{f}},\mu^{\operatorname{m}},$ and $\mu^{\operatorname{l}}$ occurrences of each symbol
then $\Stop \compose \angbrack{\Smid^{\operatorname{f}},\Smid^{\operatorname{m}},\Smid^{\operatorname{l}}}$ contains 
$\mu^{\operatorname{f}} + \mu^{\operatorname{l}} + (\mu-2)\mu^{\operatorname{m}}$ occurrences of each symbol.
Figure~\ref{fig:substitution-shuffle} gives a schematic of the generation 
of the sequence $\paren{\Stop\compose\angbrack{\Smid^{\operatorname{f}},\Smid^{\operatorname{m}},\Smid^{\operatorname{l}}}}
\shuffle \Sbot$.

\begin{figure}
\centering
\scalebox{.30}{\includegraphics{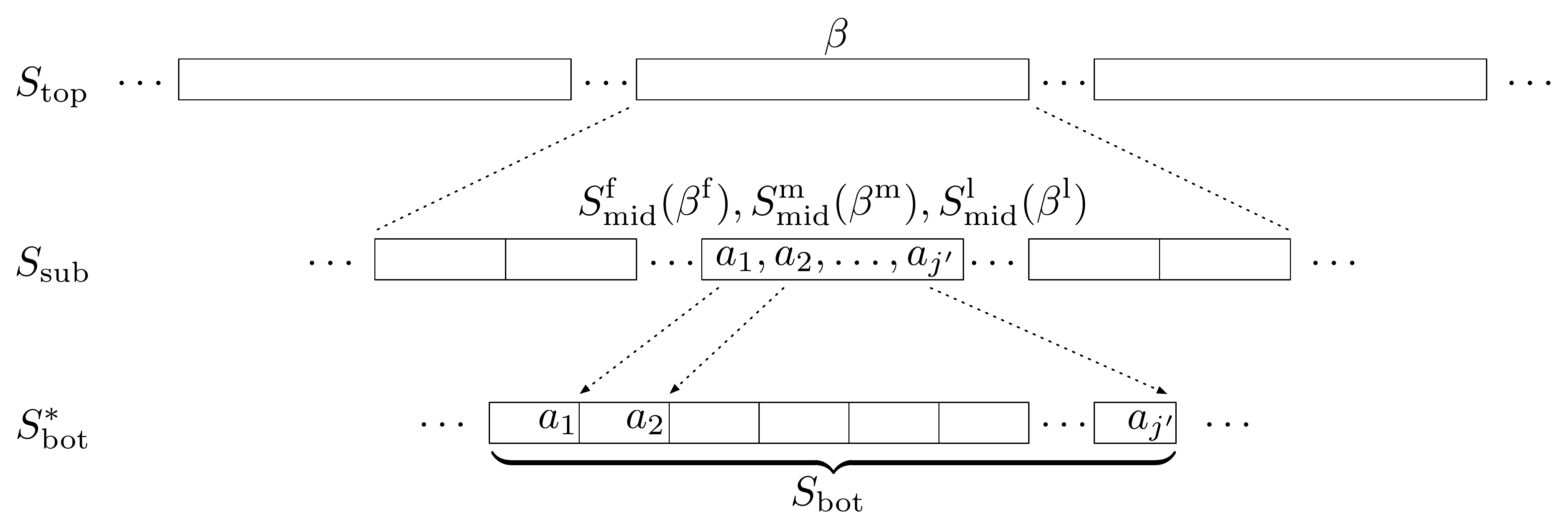}}
\caption{\label{fig:substitution-shuffle}Three-fold composition followed by shuffling.
Each block $\beta$ in $\Stop$ is replaced with the concatenation of
$\Smid^{\operatorname{f}}(\beta^{\operatorname{f}}), \Smid^{\operatorname{m}}(\beta^{\operatorname{m}})$,
and $\Smid^{\operatorname{l}}(\beta^{\operatorname{l}})$
and each block of {\em that} sequence is shuffled with a single copy of $\Sbot$ in $\Sbot^*$.
In general, blocks in $\Smid^{\operatorname{f}}(\beta^{\operatorname{f}}), \Smid^{\operatorname{m}}(\beta^{\operatorname{m}})$,
and $\Smid^{\operatorname{l}}(\beta^{\operatorname{l}})$ will not attain their maximum length $j'$.
}
\end{figure}

\subsection{Sequences of Orders 4 and 5}

The sequences $S_4(i,j)$ and $S_5(i,j)$ are defined inductively below.
As we will prove, $S_4(i,j)$ is an order-$4$ DS sequence partitioned into blocks
of length precisely $j$ in which each symbol appears $2^i$ times,
whereas $S_5(i,j)$ is an order-$5$ DS sequence partitioned into blocks of length {\em at most}
$j$ in which each symbol appears $(2i-3)2^i+4$ times.  Let $B_s(i,j) = \bl{S_s(i,j)}$ and $N_s(i,j) = \|S_s(i,j)\|$ be, respectively, 
the number of blocks in $S_s(i,j)$ and the alphabet size of $S_s(i,j)$.
By definition $|S_4(i,j)| = 2^i\cdot N_4(i,j) = j\cdot B_4(i,j)$ and $|S_5(i,j)| = ((2i-3)2^i+4)\cdot N_5(i,j) \le j\cdot B_5(i,j)$.
The construction of $S_4$ is the same as Nivasch's~\cite{Nivasch10} and similar to that of Agarwal et al.~\cite{ASS89}.

The base cases for our sequences are given below, where square brackets indicate blocks:
\begin{align*}
S_2(j) 	&= [12\cdots(j-1)j] \:\: \zero{[j(j-1)\cdots 21]}		& \mbox{ two blocks with length $j$}\\
S_4(1,j) &= S_5(1,j)	= S_2(j)\\
S_4(i,1)	&= [1]^{2^i}							& \mbox{ $2^i$ identical blocks}\\
S_5(i,1)	&= [1]^{(2i-3)2^i+4}						& \mbox{ $(2i-3)2^i + 4$ identical blocks}\\
\intertext{Observe that these base cases satisfy the property that symbols appear precisely $2^i$
times in $S_4(i,\cdot)$ and $(2i-3)2^i + 4$ times in $S_5(i,\cdot)$.  Define $S_4(i,j)$ as}
S_4(i,j)	&= \zero{\Paren{S_4(i-1, y)\compose S_2(y)} \shuffle S_4(i,j-1),} & \mbox{ where $y=B_4(i,j-1)$}\\
\intertext{and $S_5(i,j)$ as}
S_5(i,j)	&= \zero{\left(\Stop \compose \Angbrack{\Smid^{\operatorname{f}},\Smid^{\operatorname{m}},\Smid^{\operatorname{l}}}\right)
\shuffle \Sbot}\\
\mbox{where \ } & \Sbot = S_5(i,j-1), & \mbox{ $z=B_5(i,j-1)$}\\
	& \Smid^{\operatorname{f}} = \Smid^{\operatorname{l}} = S_4(i,z),    \\
	& \Smid^{\operatorname{m}} = S_2(N_4(i,z)),\\
\mbox{and \ }	& \Stop = S_5(i-1, N_4(i,z)),\\
\end{align*}

By definition $\Smid^{\operatorname{f}}$ and $\Smid^{\operatorname{l}}$ are partitioned into blocks with length $z$.
In the three-fold composition operation we also interpret $S_2(N_4(i,z))$ as a sequence of blocks of length precisely $z$.  It will
be shown shortly that $N_4(i,z)$ is, in fact, a multiple of $z$.
We argue by induction that symbols appear with the correct multiplicity in $S_4$ and $S_5$.
In the case of $S_4$ each symbol appears  $2^{i-1}$ times in $S_4(i-1,y)$ (by the inductive hypothesis), twice in $S_2(y)$,
and therefore $2^i$ times in $S_4(i-1,y)\compose S_2(y)$.  Symbols in copies of $S_4(i,j-1)$ already appear $2^i$ times,
by the inductive hypothesis.  In $S_5(i-1,N_4(i,z))$ each symbol appears
$(2i-5)2^{i-1} + 4$ times.  
The 3-fold composition operation increases the multiplicity of such symbols to 
$2\left((2i-5)2^{i-1} + 2\right) + 2\left(2^{i}\right) = (2i-3)2^i + 4$,
where the first term accounts for the blowup in middle occurrences and the second term for the blowup in first and last occurrences.
It follows that $B$ and $N$ are defined inductively as follows.
\begin{align*}
B_4(1,j) &= B_5(1,j) = B_2(j) = 2\\
B_4(i,1) &= 2^i\\
B_5(i,1) &= (2i-3)2^i+4\\
B_4(i,j)	&= B_4(i-1,y) \cdot 2 \cdot y	& \mbox{where $y=B_4(i,j-1)$}\\
B_5(i,j)	&= B_5(i-1,  N_4(i,z)) \cdot (2+2^{-i+1}) B_4(i,z) \cdot z & \mbox{where $z=B_5(i,j-1)$} \\
N_4(1,j) &= N_5(1,j) = N_2(j) = j\\
N_4(i,1) &= N_5(i,1) = 1\\
N_4(i,j)	&= N_4(i-1,y) \; + \; B_4(i-1,y) \cdot 2 \cdot N_4(i,j-1)\\
N_5(i,j)	&= \zero{N_5(i-1,N_4(i,z)) \; + \; B_5(i-1,N_4(i,z)) \cdot (2+2^{-i+1}) B_4(i,z) \cdot N_5(i,j-1)}
\end{align*}

The $2+2^{-i+1}$ factor in the definition of $B_5(i,j)$ and $N_5(i,j)$ comes from the fact
that in the shuffling step, $S_2(N_4(i,z))$ is interpreted as having $|S_2(N_4(i,z))|/z$ blocks of length $z$,
where
\[
\f{|S_2(N_4(i,z))|}{z} = \f{2\cdot N_4(i,z)}{z} = \f{2\cdot z\cdot B_4(i,z)}{z\cdot 2^{i}} = 2^{-i+1}B_4(i,z)
\]

\begin{lemma}
For $s\in\{4,5\}$, $S_s(i,j)$ is an order-$s$ Davenport-Schinzel sequence.
\end{lemma}

\begin{proof}
We use brackets to indicate block boundaries in (forbidden) patterns, e.g., $[ba]ba$ is a 
pattern where the first $ba$ appears in one block and the last $ba$ appears outside that block.
One can easily show by induction that $ba[ba] \nsubseq S_s(i,j)$ 
and $[ba]ab\nsubseq S_s(i,j)$ for all $s\in\{4,5\},i>1,j\ge 1$.
The base cases are trivial.  When $a$ is shuffled into the indicated block in a copy of $S_s(i,j-1)$, all $b$s appear in that copy
and all other $a$s are shuffled into different copies, hence $[ba]$ cannot be preceded by $ba$ or followed by $ab$.
This also implies that two symbols cannot both appear in two blocks of $S_s(i,j)$, for all $i>1$.
It follows that the patterns $ababab$ (and $abababa$) cannot be introduced into $S_4$ (and $S_5$) by the shuffling operation but 
must come from the composition (and 3-fold composition) operation.
Suppose $ba\subseq \beta$ for some block $\beta$ in $S_4(i-1, y)$.  
It follows that composing $\beta$ with $S_2(y)$ (a $baba$-free sequence)
does not introduce an $ababab$ pattern.  (Substituting $\underline{bab}$ for $ba\subseq \beta$ and projecting onto $\{a,b\}$ yields
sequences of the form $a^* b^* \underline{bab} b^* a^*$.)

Turning to $S_5$, suppose $ab\subseq \beta$ for some block $\beta$ in $\Stop$.
If $a$ and $b$ are both middle symbols in $\beta$ then, by the same argument, 
composing $\beta$ with $\Smid^{\operatorname{m}} = S_2(N_4(i,z))$ does not introduce an 
$ababab$ pattern much less an $abababa$ pattern.
If both $a,b$ are first then composing $\beta$ with an order-$4$ DS sequence $\Smid^{\operatorname{f}} = S_4(i,z)$ and 
projecting onto $\{a,b\}$
yields patterns of the form $\underline{a^*b^*a^*b^*a^*} a^*b^*$, where the underlined portion originated from $\beta$.
The case when $a$ and $b$ are last is symmetric.
The cases when $a$ and $b$ are of different types (first-middle, first-last, last-middle) are handled similarly.
\end{proof}

We have shown that $\DS{4}(N_4(i,j), B_4(i,j)) \ge 2^i N_4(i,j)$
and $\DS{5}(N_5(i,j), B_5(i,j)) \ge ((2i-3)2^i + 4) N_5(i,j)$.
Since any blocked sequence can be turned into a $2$-sparse
sequence by removing duplicates at block boundaries this also implies
that $\DS{4}(N_4(i,j)) \ge 2^i N_4(i,j) - B_4(i,j) > (1-1/j)2^iN_4(i,j)$.
Remember that all blocks in $S_4(i,j)$ have length exactly $j$.
There is no such guarantee for $S_5$, however.  It is conceivable that it consists 
largely of long runs of identical symbols (each in a block of length 1), 
nearly {\em all} of which would be removed when converting it to a 2-sparse sequence.
That is, statements of the form
$\DS{5}(N_5(i,j)) \ge ((2i-3)2^i+4)N_5(i,j) - B_5(i,j)$ 
become trivial if the $B_5(i,j)$ term dominates.
Lemma~\ref{lem:blockdensity} shows that for $j$ sufficiently large this does not occur
and therefore removing duplicates at block boundaries does not affect the length of $S_5(i,j)$ asymptotically.

\begin{lemma}\label{lem:blockdensity}
$N_5(i,j) \ge j\cdot B_5(i,j) / \xi(i)$, where $\xi(i) = 3^i 2^{i+1\choose 2}$.
\end{lemma}

\begin{proof}
When $i=1$ we have $N_5(1,j) = j \ge j\cdot B_5(1,j) / \xi(1) = 2j/6$.
When $j=1$ we have $N_5(i,1) = 1 \ge B_5(i,1)/\xi(i) = \left((2i-3)2^i+4\right)/ 3^i2^{i+1 \choose 2}$.
Assuming the claim holds for all $(i',j') < (i,j)$ lexicographically, 
\begin{align*}
\lefteqn{N_5(i,j)}\\
&= N_5(i-1,N_4(i,z)) + B_5(i-1,N_4(i,z)) \cdot (2+2^{-i+1}) B_4(i,z)\cdot N_5(i,j-1) \hcm[-.8] & \mbox{\{defn. of $N_5$\}}\\
		&\ge N_5(i-1,N_4(i,z)) +  \f{1}{\xi(i)} B_5(i-1,N_4(i,z)) \cdot (2+2^{-i+1}) B_4(i,z)\cdot (j-1) z \hcm[-.8]&\mbox{\{ind., defn. of $z$\}}\\
		&= N_5(i-1,N_4(i,z)) \,+\,  \f{j-1}{\xi(i)} B_5(i,j) & \mbox{\{defn. of $B_5$\}}\\
		&\ge \f{1}{\xi(i-1)} N_4(i,z)\cdot B_5(i-1, N_4(i,z)) \,+\, \f{j-1}{\xi(i)} B_5(i,j) & \mbox{\{ind. hyp.\}}\\
		&\ge \f{1}{\xi(i-1)\cdot 2^{i}}\cdot z\cdot B_4(i,z) \cdot B_5(i-1, N_4(i,z)) + \f{j-1}{\xi(i)} B_5(i,j) & \mbox{\{$N_4(i,z) = \fr{z}{2^i}B_4(i,z)$\}}\\
		&\ge \f{1}{\xi(i-1)\cdot 2^{i}\cdot 3} \cdot (2+2^{-i+1}) \cdot z  \cdot B_4(i,z) \cdot B_5(i-1,N_4(i,z)) + \f{j-1}{\xi(i)} B_5(i,j) \hcm[-.8]& \mbox{\{$2+2^{-i+1} \le 3$\}}\\
		&= \f{1}{\xi(i)}B_5(i,j) + \f{j-1}{\xi(i)} B_5(i,j) \; = \; \f{j}{\xi(i)}B_5(i,j) & \mbox{\{defn. of $B_5$, $\xi$\}}
\end{align*}

\end{proof}

\begin{theorem}\label{thm:DS5-lb}
For any $n$ and $m$, $\DS{5}(n,m) = \Omega(n \alpha(n,m)2^{\alpha(n,m)})$ and 
$\DS{5}(n) = \Omega(n\alpha(n)2^{\alpha(n)})$.
\end{theorem}

\begin{proof}
Consider the sequence $S_5 = S_5(i,j)$, where $j \ge \xi(i)$, and let $S_5'$ be obtained by removing duplicates at block boundaries.
It follows that $S_5'$ is $2$-sparse and, from Lemma~\ref{lem:blockdensity}, 
that $|S_5'| \ge ((2i-3)2^i + 3)N_5(i,j)$.  
It is straightforward to prove that $i = \alpha(N_5(i,j), B_5(i,j)) + O(1)$ and that $i = \alpha(N_5(i,j)) + O(1)$ when $j = \xi(i)$.
\end{proof}

\section{Upper Bounds on Fifth-Order Sequences}\label{sect:order5-ub}

Recall from Section~\ref{sect:derivation-tree} that a derivation tree $\Tree(S)$ for a blocked sequence $S$ 
is the composition of a global tree $\GTree = \GTree(\GS')$ and local trees $\{\LTree_q\}_{q\le\Gm}$,
where $\LTree_q = \Tree(\LS_q)$.
The composition is effected by identifying the $q$th leaf of $\GTree$, call it $x_q$, 
with the root of $\LTree_q$, then populating the leaves of $\Tree$ with the blocks of $S$.

\subsection{Superimposed Derivation Trees}

One can view $\Tree(S)$ as representing a hypothetical {\em process} for generating the sequence $S$,
but it only represents this process at one granularity.
For example, $\GS_q$ is the portion 
of $\GS$ at the leaf descendants of $x_q$ in $\Tree$.
The derivation tree $\Tree$ 
does not let us inspect the structure of $\GS_q$ and provides 
no explanation for how $\GS_q$ came to be.  

To reason about $\GS_q$ we could, of course, build a {\em new} derivation tree
$\GTree_q = \Tree(\GS_q)$ just for $\GS_q$.  One can see that $\LTree_q$ and $\GTree_q$
will be structurally identical if, in their inductive construction, we always choose block partitions in the same way.
One can imagine {\em superimposing} $\GTree_q$ onto $\LTree_q$, 
regarding both as being on the same node set but populated with different blocks.

In our actual analysis we do not consider the derivation tree for $\GS_q$, which includes all global occurrences 
in $S_q$, but just those derivation trees for $\GfS_q$ and $\GlS_q$, which are restricted to global symbols making their
first and last appearance in $S_q$, respectively.
Define $\FirstTree[x_q] = \Tree(\GfS_q)$ and $\LastTree[x_q] = \Tree(\GlS_q)$ to be any derivation trees of 
$\GfS_q$ and $\GlS_q$ that are defined on the same node set as $\LTree_q$.  Recall that $x_q$ is the $q$th leaf of $\GTree$. 

One can think of the $\FirstTree$ and $\LastTree$ derivation trees 
as filling in the gaps between wing nodes and quills.  
Suppose $v$ were a leaf in some derivation tree $\Tree$ whose block $\block(v)$ contains a symbol $a$.  
By definition $\quill_{|a}(v)$ is a child of $\wing_{|a}(v)$ in $\Tree_{|a}$.  If $v$ were a dove (or hawk) in $\Tree_{|a}$
then $\quill_{|a}(v)$ would be identified with a leaf of $\FirstTree[\wing_{|a}(v)]$ (or a leaf of $\LastTree[\wing_{|a}(v)]$)
whose block contains $a$.
However, {\em within} $\FirstTree[\wing_{|a}(v)]$ (or $\LastTree[\wing_{|a}(v)]$),
$\quill_{|a}(v)$ could be a dove or hawk, feather or non-feather, wingtip or non-wingtip.

The new concept needed to tightly bound order-$5$ DS sequences is that of a {\em double-feather}.
See Figure~\ref{fig:FirstTree}.

\begin{definition}\label{def:double-feather}
Let $\{\Tree\}\cup \{\FirstTree[u],\LastTree[u]\}_{u\in\Tree}$ be a derivation tree ensemble.
Let $v$ be a dove leaf in $\Tree$ for which $a\in \block(v)$,
and let $\FirstTree = \FirstTree[\wing_{|a}(v)]$.
We call $v$ a {\em double-feather} in $\Tree_{|a}$ if it is a feather (that is, it is the rightmost descendant 
of $\quill_{|a}(v)$ in $\Tree_{|a}$) \underline{and} $\quill_{|a}(v)$ is {\em either} a dove feather or hawk wingtip
in $\FirstTree_{|a}$.  The definition of double-feather is symmetric when $v$ is a hawk, that is,
we substitute $\LastTree$ for $\FirstTree$ and swap the roles of left and right, dove and hawk.
\end{definition}

\begin{figure}
\centerline{\scalebox{.35}{\includegraphics{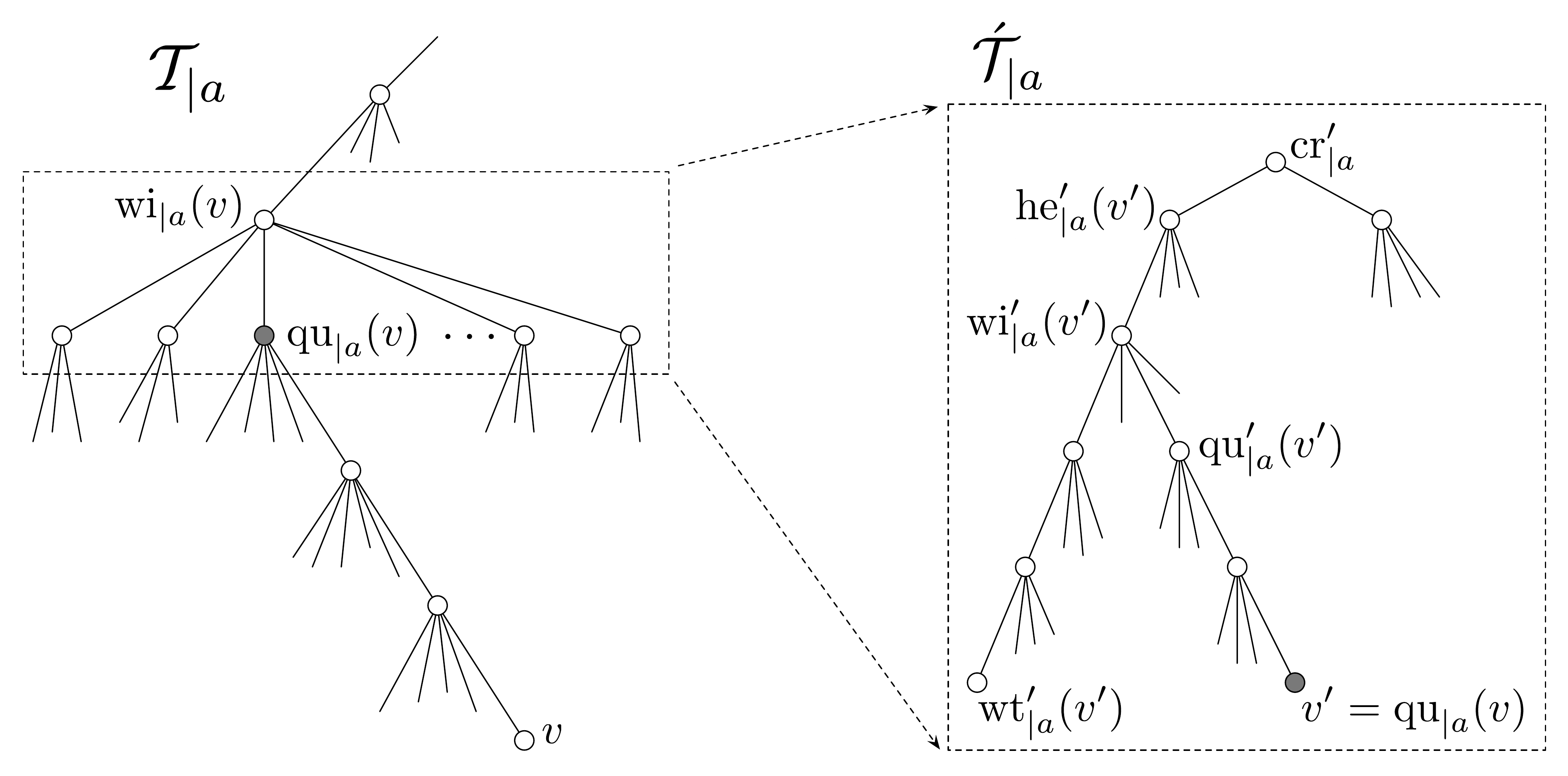}}}
\caption{\label{fig:FirstTree} Left: a dove feather $v$ in $\Tree_{|a}$.  The quill and wing
node of $v$ are indicated.  Right: the derivation tree $\FirstTree_{|a}$ where $\FirstTree=\FirstTree[\wing_{|a}(v)]$.
Within $\FirstTree_{|a}$ the leaf $v'=\quill_{|a}(v)$ has its own wing node $\wing_{|a}'(v')$, quill $\quill_{|a}'(v')$,
and so on.  By virtue of $v'$ being a dove feather in $\FirstTree_{|a}$, $v$ is a double-feather in $\Tree_{|a}$.}
\end{figure}

As with the term {\em feather}, {\em double-feather} is used to refer to leaf nodes in some derivation tree $\Tree_{|a}$
and the corresponding occurrences of $a$ in the underlying sequence $S$.
Lemma~\ref{lem:order5-nested} is a more refined
version of Lemma~\ref{lem:nested}.

\begin{lemma}\label{lem:order5-nested}
Let $\{\Tree\} \cup \{\FirstTree[u], \LastTree[u]\}_{u\in\Tree}$ be a derivation tree ensemble for some sequence $S$.
Suppose that $v\in\Tree$ is a leaf and $a,b$ are symbols in a block $\block(v)$ of $S$.
If the following three criteria are satisfied then $a$ and $b$ are nested in $\block(v)$.
\renewcommand{\theenumi}{\roman{enumi}}
\begin{enumerate}
\item $v$ is not a wingtip in either $\Tree_{|a}$ or $\Tree_{|b}$.\label{crit:middle}
\item $v$ is not a double-feather in either $\Tree_{|a}$ or $\Tree_{|b}$.\label{crit:feather}
\item $v$ is a dove in both $\Tree_{|a}$ and $\Tree_{|b}$ or a hawk in both $\Tree_{|a}$ and $\Tree_{|b}$.\label{crit:sametype}
\end{enumerate}
\end{lemma}

\begin{proof}
We assume the claim is false, that $a$ and $b$ are interleaved in $\block(v)$.
Without loss of generality, we can assume the following additional criteria.
\renewcommand{\theenumi}{\roman{enumi}}
\begin{enumerate}
\setcounter{enumi}{3}
\item $\crown_{|b}$ is equal to or strictly ancestral to $\crown_{|a}$.\label{crit:crown}
\item $v$ is a dove in both $\Tree_{|a}$ and $\Tree_{|b}$.\label{crit:dove}
\item $\Tree$ is the smallest derivation tree for which Criteria~(\ref{crit:middle}--\ref{crit:dove}) hold and where
$a$ and $b$ are interleaved in $\block(v)$.\label{crit:smallest}
\end{enumerate}

By Criterion (\ref{crit:dove}) the leftmost descendent of $\wing_{|a}(v)$ in $\Tree_{|a}$ is $\tip_{|a}(v)$.
Let $u$ be its rightmost descendant.
Criteria (\ref{crit:middle},\ref{crit:feather}) imply that
\renewcommand{\theenumi}{\roman{enumi}}
\begin{enumerate}
\setcounter{enumi}{6}
\item $v,\tip_{|a}(v)$, and $u$ are distinct nodes.\label{crit:uv-distinct}
\end{enumerate}
Criterion (\ref{crit:middle}) states that $v$ is distinct from $\tip_{|a}(v)$.  
Consider the derivation tree $\FirstTree = \FirstTree[\wing_{|a}(v)]$.
The node $u$ is the rightmost descendant (in $\Tree_{|a}$) of the hawk wingtip of $\FirstTree_{|a}$.
It is therefore a double-feather, and distinct from $v$, by Criterion~(\ref{crit:feather}).

Partition the sequence outside of $\block(v)$ into four intervals, namely
$I_1$: everything preceding the $a$ in $\block(\tip_{|a}(v))$, $I_2$: everything from the end of $I_1$ to $\block(v)$,
$I_3$: everything from $\block(v)$ to the $a$ in $\block(u)$, and $I_4$: everything following $I_3$.
If $a$ and $b$ are not nested in $\block(v)$ then all remaining $b$s lie exclusively in $I_1$ and $I_3$ or
exclusively in $I_2$ and $I_4$.  We claim
\renewcommand{\theenumi}{\roman{enumi}}
\begin{enumerate}
\setcounter{enumi}{7}
\item $I_1$ and $I_3$ contain no occurrences of $b$.\label{crit:I2I4}
\end{enumerate}
If the contrary were true, that all occurrences of $b$ outside $\block(v)$ were in $I_1$ and $I_3$,
then $I_3$ would contain $\righttip_{|b}$, which is distinct from $v$ according to Criterion (\ref{crit:middle}).
By Criterion (\ref{crit:crown}) $\crown_{|b}$ is ancestral to $\crown_{|a}$, which is strictly ancestral to $\wing_{|a}(v)$,
which is ancestral to both $v$ and $\righttip_{|b}$.  This implies that $v$ and $\righttip_{|b}$ descend from the same child of $\crown_{|b}$, namely, $\righthead_{|b}$, which violates Criterion (\ref{crit:dove}).
Thus, $\tip_{|b}(v)$ is the dove wingtip in $\Tree_{|b}$ and lies in interval $I_2$.

Define $v' = \quill_{|b}(v)$ to be $v$'s quill in $\Tree_{|b}$.  It is not necessarily the case that $v'$ is distinct from $v$,
though we can claim that
\renewcommand{\theenumi}{\roman{enumi}}
\begin{enumerate}
\setcounter{enumi}{8}
\item $v'$ is a strict descendant of $\wing_{|a}(v)$ and $v=\feather_{|b}(v)$ is a feather.\label{crit:v'}
\end{enumerate}
By definition quills are not wing nodes, so $v'$ cannot be ancestral to $\tip_{|b}(v)$.  
However, if $v'$ were ancestral to $\wing_{|a}(v)$ it would be ancestral
to $\tip_{|b}(v)$ as well, a contradiction.
If $v$ were distinct from $\feather_{|b}(v)$, that is, if $v$ were not the rightmost descendant of $v'$ in $\Tree_{|b}$,
then $\feather_{|b}(v)$ must, by Inference (\ref{crit:I2I4}), lie in interval $I_4$.  Since $\othertip_{|a}(v)$ is not a descendant
of $\wing_{|a}(v)$ it must lie to the right of $\feather_{|b}(v)$ in $I_4$.  
However, this arrangement of nodes (namely $\tip_{|a}(v),\tip_{|b}(v),v,\feather_{|b}(v),\othertip_{|a}(v)$, where $\tip_{|a}(v)$
and $\tip_{|b}(v)$ may be equal)
shows that $a$ and $b$ are nested in $\block(v)$, hence $v=\feather_{|b}(v)$ is a feather.

From Criterion~(\ref{crit:smallest}), Inference~(\ref{crit:v'}), and Lemma~\ref{lem:l-node} we shall infer that
\renewcommand{\theenumi}{\roman{enumi}}
\begin{enumerate}
\setcounter{enumi}{9}
\item$v=v'$.\label{crit:v=v'}
\end{enumerate}
Suppose $v\neq v'$.
Lemma~\ref{lem:l-node} implies that $a\in\block(v')$,
as witnessed by the nodes $v,v',\crown_{|a},\crown_{|b}$ on a leaf-to-root path.
This means that somewhere in the inductive construction of $\Tree$ we encountered
a derivation tree $\Tree^0$ containing both $\crown_{|a}$ and $\crown_{|b}$,
whose leaves are at the level of $v'$.  However, $\Tree^0$ and $v'$ satisfy the
conditions of the lemma, namely $a,b\in \block(v')$ and $v'$ is neither a wingtip nor double-feather nor hawk
in both $\Tree^0_{|a}$ and $\Tree^0_{|b}$.
Since $\Tree^0$ is smaller than $\Tree$, Criterion~(\ref{crit:smallest}) implies
that $a$ and $b$ are nested in $\block(v')$ with respect to $\Tree^0$, which
then implies that they are nested in $\block(v)$ with respect to $\Tree$ as well.
This contradicts the hypothesis that $a$ and $b$ are interleaved in $\block(v)$.

\begin{figure}
\centering\scalebox{.28}{\includegraphics{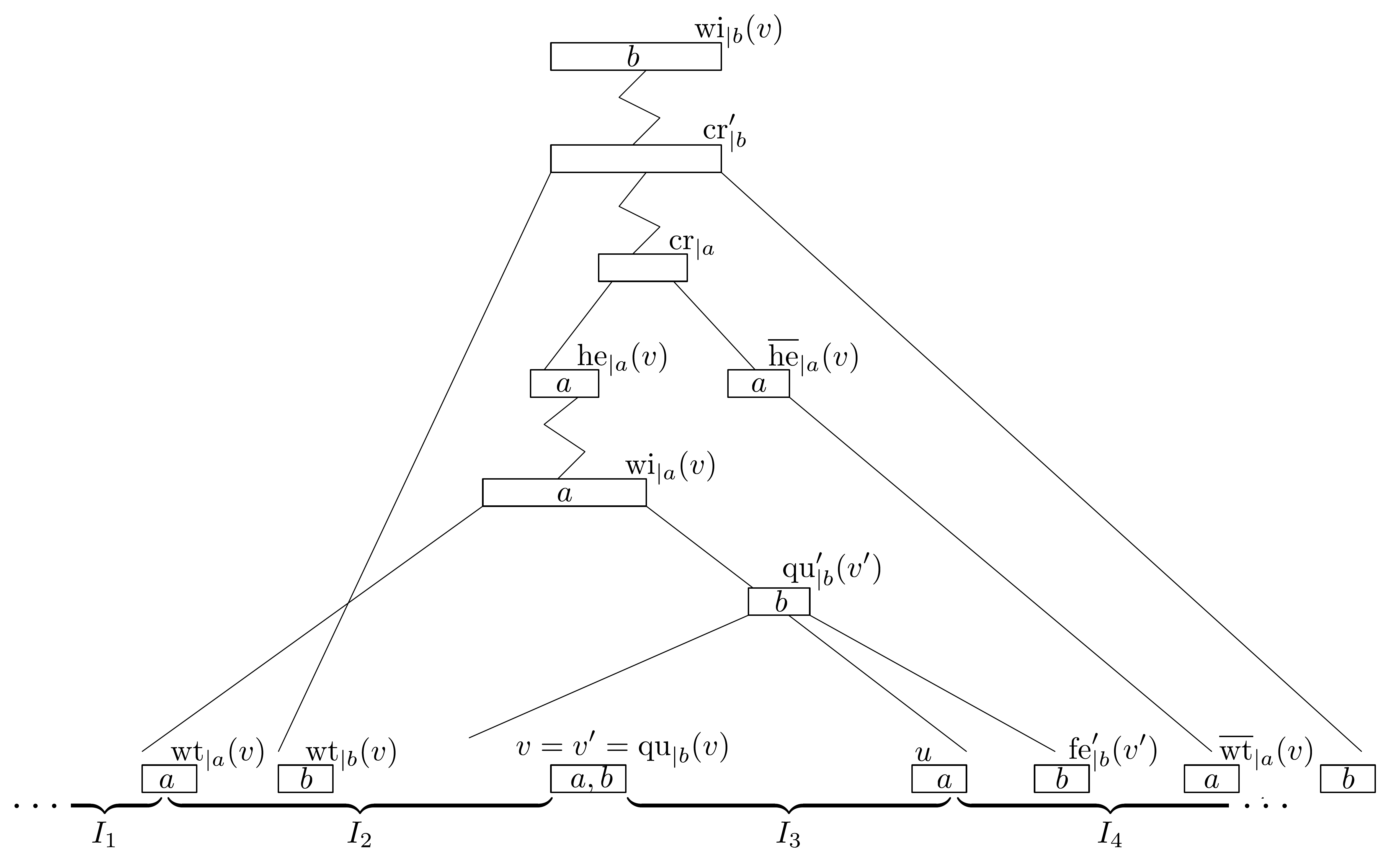}}
\caption{\label{fig:order5-nested}
Here $v'=\quill_{|b}(v)$ is $v$'s quill in $\Tree_{|b}$, and therefore a leaf of $\FirstTree = \FirstTree[\wing_{|b}(v)]$.
The tree $\FirstTree_{|b}$ is rooted at the crown $\crown_{|b}'$.
Looking within $\FirstTree_{|b}$, $v'$ has its own quill $\quill_{|b}'(v')$ and feather $\feather_{|b}'(v')$.
Since $\FirstTree$ is topologically identical to a corresponding subtree of $\Tree$, we can talk sensibly about
nodes in one tree being ancestral to nodes in the other.  For example, it is deduced that
$\crown_{|b}'$ (in $\FirstTree$) is ancestral to $\crown_{|a}$ (in $\Tree$) 
and that $\wing_{|a}(v)$ (in $\Tree$) is a strict ancestor of $\quill_{|b}'(v')$ (in $\FirstTree$).
}
\end{figure}

By definition $v'$ is the child of $\wing_{|b}(v)$ in $\Tree_{|b}$ but we have yet to deduce where $\wing_{|b}(v)$
is relative to other nodes.  
\renewcommand{\theenumi}{\roman{enumi}}
\begin{enumerate}
\setcounter{enumi}{10}
\item $\wing_{|b}(v)$ is ancestral to $\crown_{|a}$.\label{crit:wingb}
\end{enumerate}
Criterion (\ref{crit:middle}) and Inference (\ref{crit:I2I4}) imply that some $b$ appears in interval $I_4$.
All such $b$s must appear after $\othertip_{|a}(v)$ for otherwise $a$ and $b$ would be nested in $\block(v)$.
By Criterion (\ref{crit:feather}) $v$ cannot be the rightmost descendant of $\wing_{|b}(v)$, which is a double-feather.
Thus, the rightmost descendant of $\wing_{|b}(v)$ appears after the $a$ in $\othertip_{|a}(v)$, implying that $\wing_{|b}(v)$ is ancestral to $\crown_{|a}$.

Consider the derivation tree $\FirstTree = \FirstTree[\wing_{|b}(v)]$, the leaf-level of which 
coincides with the leaf-level of $\Tree$ since $\quill_{|b}(v) = v' = v$.
Let $\crown_{|b}'$ be $b$'s crown in $\FirstTree_{|b}$,
and let $\quill_{|b}'(v')$ and $\feather_{|b}'(v')$ be the quill and feather of $v'$ in $\FirstTree_{|b}$.\footnote{Observe that $\quill'_{|b}(v')$ lies strictly between $v'=\quill_{|b}(v)$ and $\wing_{|b}(v)$ in $\Tree$,
so its block contains $b$ only with respect to $\FirstTree$, not $\Tree$.  In contrast, $b$ appears in $\feather_{|b}'(v')$'s block
in both $\Tree$ and $\FirstTree$.}
We deduce the following variants of Criteria (\ref{crit:crown},\ref{crit:dove}) and Inference~(\ref{crit:v'}).
\renewcommand{\theenumi}{\roman{enumi}}
\begin{enumerate}
\setcounter{enumi}{11}
\item $\crown_{|b}'$ is ancestral to $\crown_{|a}$.\label{crit:crown'}
\item $v'$ is a dove in $\FirstTree_{|b}$.\label{crit:dove'}
\item $\quill_{|b}'(v')$ is a strict descendant of $\wing_{|a}(v)$ and strict ancestor of $u$.
Furthermore, $\feather_{|b}'(v')$ lies between $u$ and $\othertip_{|a}(v)$.\label{crit:v''}
\end{enumerate}
Inference (\ref{crit:crown'}) is just a restatement of Inference (\ref{crit:wingb}) since all 
$b$s descending from $\wing_{|b}(v)$ are also descendants of $\crown_{|b}'$.\footnote{The assertion
that $\crown_{|b}'$ is ancestral to $\crown_{|a}$ is only well defined if $\FirstTree[\wing_{|b}(v)]$ can be
regarded as a subtree of $\Tree$ but populated with different blocks.  This is why we do not permit
corresponding derivation trees to have different structure.}
Inference (\ref{crit:dove'}) follows from the fact that $\tip_{|b}(v)$ is also the dove wingtip in $\FirstTree_{|b}$
and that the least common ancestor of $\tip_{|b}(v)$ and $v$ is a descendant of $\wing_{|a}(v)$ and therefore
a strict descendant of $\crown_{|b}'$.
We now turn to Inference (\ref{crit:v''}).  By Inference (\ref{crit:dove'}) the quill $\quill_{|b}'(v')$ 
must be an ancestor of $v'$ but not $\tip_{|b}(v)$, so it must be a strict descendant of $\wing_{|a}(v)$.
If $v'$ were the rightmost descendant of $\quill_{|b}'(v')$ in $\FirstTree_{|b}$,
then $v'=\feather_{|b}'(v')$ and $v$ would be a double-feather in $\Tree$, contrary to Criterion (\ref{crit:feather}).
Since $I_3$ is $b$-free, it must be that $\quill_{|b}'(v')$ is a strict descendant of $\wing_{|a}(v)$,
a strict ancestor of $u$, and that $\feather_{|b}'(v')$ lies between $u$ and $\othertip_{|a}(v)$.
The blocks $\tip_{|a}(v),\tip_{|b}(v),v,\feather_{|b}'(v'),$ and $\othertip_{|a}(v)$ certify that $a$ and $b$ are
nested in $\block(v)$, a contradicting the hypothesis that they are not.
\end{proof}

\begin{remark}\label{remark:double-feather}
Whereas Lemma~\ref{lem:nested} implicitly partitioned occurrences of global symbols
into four categories: dove wingtips, hawk wingtips, feathers, and all remaining non-feathers,
Lemma~\ref{lem:order5-nested} further distinguishes dove non-feathers and hawk non-feathers.
The reason for this is rather technical.  If Criterion~(\ref{crit:sametype}) were dropped
and $v$ were a dove in $\Tree_{|a}$ and a hawk in $\Tree_{|b}$, we could deduce that $\quill_{|b}'(v')$
is a strict descendant of $\wing_{|a}(v)$ and strict ancestor of $\tip_{|a}(v)$.  However,
we could not deduce that $a\in\block(\quill_{|b}'(v'))$ (that is, $\quill_{|b}'(v')$'s block in $\Tree$)
since Lemma~\ref{lem:l-node} only applies to symbols and blocks that exist in the {\em same} derivation tree.
Note that $\FirstTree[\wing_{|b}(v)]$ can be regarded as being superimposed on the subtree of $\Tree$ rooted at $\wing_{|b}(v)$, 
but they are not identical derivation trees.
\end{remark}

As in Section~\ref{sect:derivation-tree} it is useful to define and reason about optimal derivation trees.
For technical reasons it is convenient to only permit uniform block partitions with widths that
are powers of two.  

\begin{definition}\label{def:permissible} {\bf (Permissible block partitions)}
Let $S$ be a blocked sequence and $m=\bl{S}$ be its block count.
A block partition $\{m_q\}_{1\le q\le \Gm}$ is {\em permissible} if $m_q = 2^r$ for all $q < \Gm$,
where $\Gm = \ceil{m/2^r}$ and $r\ge 1$ is an integer.
A derivation tree $\Tree(S)$ is {\em permissible} if it was defined using only permissible block partitions.
\end{definition}

Note that if $\Tree_0$ and $\Tree_1$ are two permissible derivation trees for some sequence with $m$ blocks,
they must have exactly the same structure, though their nodes may be populated with different blocks.  
In particular, both are binary trees with height $\ceil{\log m}$, where one-child nodes may only exist on the 
rightmost root-to-leaf path if $m$ is not a power of 2.

\subsection{Recurrences for Fifth-Order Sequences}

Lemma~\ref{lem:order5-nested} provides us with new criteria for nestedness.
In order to write a new recurrence for $\DS{5}$ we need to know how many 
double-feathers an order-5 DS sequence can have, which depends
on the number of dove and hawk feathers in an order-4 DS sequence.  

\begin{definition} {\bf (Optimal derivation trees)}
\begin{itemize}
\item When $S$ is an order-4 DS sequence let $\Tree^*(S)$ denote the permissible derivation tree that minimizes
the number of feathers of a given type (dove or hawk).  When $S$ is an order-5 DS sequence let
$\Ensemble^*(S) = \{\Tree^*(S)\} \cup \{\FirstTree^*[u],\LastTree^*[u]\}_{u\in\Tree^*(S)}$ denote the ensemble of permissible
derivation trees that minimize the number of double-feathers in $S$.
\item Define $\FFeather(n,m)$ to be the maximum number of feathers of one
type (dove or hawk) in an order-4 DS sequence $S$ with respect to $\Tree^*(S)$, where $\|S\|=n$ and $\bl{S}=m$.
\item Define $\DblFeather(n,m)$ to be the maximum number of double-feathers (of both types)
in an order-5 DS sequence with respect to the ensemble $\Ensemble^*(S)$.
\end{itemize}
\end{definition}

\begin{recurrence}\label{rec:double-feather}
Let $m$ and $n$ be the block count and alphabet size.
For any permissible block partition $\{m_q\}_{1\le q\le \Gm}$
and any alphabet partition $\{\Gn\}\cup \{\Ln_q\}_{1\le q\le \Gm}$, we have
\begin{align*}
\FFeather(n,m) &\;=\; \DblFeather(n,m) = 0	& \mbox{when $m\le 2$}\\
\FFeather(n,m) &\;\le\; \sum_{q=1}^{\Gm} \FFeather(\Ln_q,m_q) \, + \, \FFeather(\Gn,\Gm)
						\, + \, \DS{3}(\Gn,m) - \Gn\\
\DblFeather(n,m) &\;\le\; \sum_{q=1}^{\Gm} \DblFeather(\Ln_q,m_q) \, + \, \DblFeather(\Gn,\Gm)
						\, + \, 2(\FFeather(\Gn,m) + \Gn)		& \mbox{\rb{7}{}}
\end{align*}
\end{recurrence}

\begin{proof}
\ignore{
In the worst case every occurrence of a symbol, excluding the two wingtips, is a 
feather or double-feather,
so $\FFeather_{i}(n,m) \le \DS{4}(n,m) - 2n$ and $\DblFeather_{i}(n,m) \le \DS{5}(n,m) - 2n$ for any $i$.\footnote{The number of dove feathers and hawk feathers in an order-$4$ DS sequence should be about the same,
but we can afford to entertain the possibility that one type dominates the other.}
This bound is only invoked when $i=1$ or $j=1$.  By Lemma~\ref{lem:i=1}
$\DS{4}(n,m) \le 8n + (cj)^2(m-1)$ and $\DS{5}(n,m) \le 16n + (cj)^3(m-1)$, where $c$ is fixed at $3$.
}
Consider an order-4 DS sequence $S$. Let $\GTree^*$ and $\{\LTree_q^*\}_q$ be the optimal derivation trees
of the global sequence $\GS'$ and local sequences $\{\LS_q\}$, and let $\Tree$ be their composition.
The number of dove feathers of local symbols is at most 
$\sum_{q} \FFeather(\Ln_q,m_q)$.  Every global dove feather in $\GS$ 
is either (i) the rightmost child of a dove feather in $\GS'$,
or (ii) a child of a left wingtip in $\GfS'$, excluding the leftmost such child, which is also a left wingtip in $\GS$.
Category (i) is counted by $\FFeather(\Gn,\Gm)$
and  
Category (ii) is counted by $\DS{3}(\Gn,m)-\Gn$ since 
$\GfS$ (the children of dove wingtips) is an order-$3$ DS sequence
over an $\Gn$-letter alphabet.  A symmetric analysis applies to hawk feathers, reversing the roles of left and right.

The analysis of $\DblFeather(n,m)$ when $S$ is an order-$5$ sequence is similar.
After the block partition is selected, construct the optimal derivation tree ensemble 
$\Ensemble^*(\GS') = \{\GTree^*\} \cup \{\FirstTree^*[u],\LastTree^*[u]\}_{u\in\GTree^*}$
for $\GS'$ and, separately, the optimal derivation tree ensemble 
$\Ensemble_q^*(\LS_q) = \{\LTree_q^*\} \cup \{\FirstTree^*[u],\LastTree^*[u]\}_{u\in \LTree^*_q}$ for each $\LS_q$ separately.  
These ensembles are composed to form $\Ensemble(S) = \{\Tree\}\cup \{\FirstTree^*[u],\LastTree^*[u]\}_{u\in\Tree}$
in the obvious way.  As usual, $\Tree = \Tree(S)$ is the composition of $\GTree^*$ and the $\{\LTree^*_q\}$.
The only nodes $u\in\Tree$ whose derivation trees $\FirstTree^*[u]$
and $\LastTree^*[u]$ are not well defined (that is, they are not already included in $\Ensemble^*(\GS')$ or the $\{\Ensemble^*(\LS_q)\}$)
are those at the leaf level of $\GTree^*$ in $\Tree$.  Define $\FirstTree^*[x_q]$ and $\LastTree^*[x_q]$ to be the optimal
derivation trees for $\GfS_q$ and $\GlS_q$, respectively.  

Now let us argue that the recurrence correctly bounds the number of double-feathers in $S$ with respect to 
$\Ensemble(S)$.
The summation $\sum_q \DblFeather(\Ln_q,m_q)$ 
counts double-feathers of local symbols.  Every occurrence of a global double-feather in $S$
is either (i) the rightmost child of a dove double-feather in $\GS'$ or the leftmost child of a hawk double-feather in $\GS'$
(with respect to $\Ensemble^*(\GS')$),
(ii) a dove feather in $\GfS_q$ or hawk feather in $\GlS_q$, for some $q$, with respect to $\FirstTree^*[x_q]$ or $\LastTree^*[x_q]$,
or
(iii) a hawk wingtip in $\GfS$ or a dove wingtip in $\GlS$.
Category (i) is counted by $\DblFeather(\Gn,\Gm)$, Category (ii) by 
$\sum_q \SqBrack{\FFeather(\Gfn_q,m_q) + \FFeather(\Gln_q,m_q)} \le 2\cdot \FFeather(\Gn,m)$,
and Category (iii) by $2\Gn$.
\end{proof}

\begin{recurrence}\label{rec:DS-double-feather}
Let $m$ and $n$ be the block count and alphabet size parameters.
For any permissible block partition $\{m_q\}_{1\le q\le \Gm}$ and 
alphabet partition $\{\Gn\} \cup\{\Ln_q\}_{1\le q\le \Gm}$,
\begin{align*}
\DS{5}(n,m) &\;\le\; \sum_{q=1}^{\Gm} \DS{5}(\Ln_q,m_q) 
			\,+\, 2\cdot \DS{4}(\Gn,m) \,+\, \DS{3}(\DblFeather(\Gn,\Gm),m) \,+\, 
			\DS{2}(\DS{5}(\Gn,\Gm),2m-1)
\end{align*}
\end{recurrence}

\begin{proof}
We adopt the usual notation for an order-5 sequence $S$ with the following additions and modifications.
Let $\GfeatherS',\GlnonfeatherS',$ and $\GrnonfeatherS'$ be the subsequences of $\GS'$ consisting of,
respectively, double-feathers, non-double-feather doves, and non-double-feather hawks, all of which 
exclude wingtips, and let $\GfeatherS,\GlnonfeatherS,$ and $\GrnonfeatherS$ be the 
subsequences of $\GS$ made up of their children.  Recall that $\GfS$ and $\GlS$ are the children 
of dove and hawk wingtips in $\GS'$.

The contribution of local symbols to $|S|$ is bounded by 
$\sum_{q} \DS{5}(\Ln_q,m_q)$.  The global sequence $\GS$ is the union 
of five subsequences $\GfS, \GlS, \GfeatherS, \GlnonfeatherS,$ and $\GrnonfeatherS$.
Both $\GfS$ and $\GlS$ are order-4
sequences, so $|\GfS|+|\GlS| \le 2\cdot \DS{4}(\Gn,m)$.  Since double-feathers are all middle occurrences in $\GS'$,
$\GfeatherS$ is obtained by substituting for each block in $\GfeatherS'$ an order-$3$ DS sequence,
so $|\GfeatherS| \le \DS{3}(|\GfeatherS'|,m) \le \DS{3}(\DblFeather(\Gn,\Gm),m)$.

According to Lemma~\ref{lem:order5-nested}, the blocks of $\GlnonfeatherS'$
consist of {\em mutually} nested symbols.  The argument from Recurrence~\ref{thm:recurrence-odd}
shows that $\GlnonfeatherS$ is obtained by substituting an order-$2$ sequence for each block in $\GlnonfeatherS'$.  
The same is true for $\GrnonfeatherS$ as well, 
so 
\begin{align*}
|\GlnonfeatherS| + |\GrnonfeatherS| &\le \DS{2}(|\GlnonfeatherS'|,m) + \DS{2}(|\GrnonfeatherS'|,m)\\
&\le \DS{2}(|\GlnonfeatherS'| + |\GrnonfeatherS'|,2m-1) & \mbox{\{superadditivity of $\DS{2}$\}}\\
&< \DS{2}(|\GS'|,2m-1)\\
&\le \DS{2}(\DS{5}(\Gn,\Gm),2m-1)\\
\end{align*}
The last inequalities follow from fact that $\GlnonfeatherS'$ and $\GrnonfeatherS'$
are disjoint subsequences of $\GS'$, which is an order-5 DS sequence.
\end{proof}

Recurrences~\ref{rec:double-feather} and \ref{rec:DS-double-feather}
allow us to find closed-form bounds on the number of feathers and double-feathers,
and on the length of order-5 DS sequences.  Refer to Appendix~\ref{sect:appendix:order5}
for proof of Lemma~\ref{lem:doublefeather-bounds}.

\begin{lemma}\label{lem:doublefeather-bounds}
Let $n$ and $m$ be the alphabet size and block count.
After a parameter $i\ge 1$ is chosen let $j\ge 1$ be minimum such that $m \le a_{i,j}^c$,
where $c=3$ is fixed.
We have the following upper
bounds on $\FFeather(n,m), \DblFeather(n,m),$ and $\DS{5}(n,m)$.

\begin{align*}
\FFeather(n,m) &\le \nu_i'\paren{n + (cj)^2(m-1)}	& \mbox{where $\nu_i' = 3{i+1\choose 2} + 3$}\\
\DblFeather(n,m) &\le \nu_i''\paren{n + (cj)^3(m-1)}  & \mbox{where $\nu_i'' = 6{i+2\choose 3} + 8i$}\\
\DS{5}(n,m) &\le \mu_{5,i}\paren{n + (cj)^3(m-1)}  & \mbox{where $\mu_{5,i} = i2^{i+7}$}\\
\end{align*}
\end{lemma}

\subsection{Blocked versus 2-Sparse Order-$5$ Sequences}\label{sect:blockedvs2sparseorder5}

Lemma~\ref{lem:doublefeather-bounds} states that for any $i$, $\DS{5}(n,m) < \mu_{5,i}(n + (3j_i)^3m)$
where $j_i$ is minimum such that $m\le a_{i,j_i}^3$.  
Choose $\iota\ge 1$ be minimum such that $(3j_\iota)^3 \le \max\{\f{n}{m}, (3\cdot 3)^3\}$.
One can show that $\iota = \alpha(n,m) + O(1)$, implying
that $\DS{5}(n,m) = O((n+m)\mu_{5,\iota}) = O((n+m)\alpha(n,m)2^{\alpha(n,m)})$,
matching the construction from Section~\ref{sect:lb}.
According to Lemma~\ref{lem:Sharir}(\ref{item:Sharir2},\ref{item:Sharir4}) 
$\DS{5}(n) = O(\alpha(\alpha(n)))\cdot \DS{5}(n,3n-1)$.
In this Section we 
present a more efficient reduction from 2-sparse, order-5 DS sequences
to blocked order-5 sequences, thereby removing the extra $\alpha(\alpha(n))$ factor.

\begin{theorem}\label{thm:order5-ub}
$\DS{5}(n) = O(n\alpha(n)2^{\alpha(n)})$
and 
$\DS{5}(n,m) = m+O(n\alpha(n,m)2^{\alpha(n,m)})$.
\end{theorem}

\begin{proof}
The second bound is asymptotically the same as $O((n+m)\alpha(n,m)2^{\alpha(n,m)})$ if $m=O(n)$.
If not, we remove up to $m-1$ repeated symbols at block boundaries, yielding a 
2-sparse, order-5 DS sequence.
Our remaining task is therefore to prove that $\DS{5}(n) = O(n\alpha(n)2^{\alpha(n)})$.

Let $S$ be a $2$-sparse, order-$5$ DS sequence with $\|S\|=n$.
Greedily partition $S$ into maximal order-$3$ DS sequences $S_1S_2\cdots S_{m}$.  
According to Sharir's argument~\cite{Sharir87}, $m\le 2n-1$.  See the proof of Lemma~\ref{lem:Sharir}(\ref{item:Sharir2})
in Appendix~\ref{appendix:Sharir}.
As usual, let $\LS,\GS\subseq S$ be the subsequences of local and global symbols, 
and let $\GS'$ be derived by contracting each interval to a single block.
The number of global symbols is $\Gn = \|\GS\|$.
In contrast to the situations we considered earlier, 
$\LS$ and $\GS$ are {\em neither} 2-sparse nor partitioned into blocks.

Let $\Ensemble^*(\GS') = \{\Tree^*(\GS')\} \cup \{\FirstTree^*[u],\LastTree^*[u]\}_{u\in\Tree^*(\GS')}$ be the optimal derivation tree ensemble for $\GS'$.
This ensemble categorizes all occurrences in $\GS'$ as (i) double-feathers, (ii) non-double-feather doves, (iii) non-double-feather hawks, 
(iv) dove wingtips, or (v) hawk wingtips.  Furthermore, occurrences in $\GS$ inherit the category of their corresponding
occurrence in $\GS'$.
Let $\GfeatherS',\GlnonfeatherS',$ and $\GrnonfeatherS'$ be the subsequences of $\GS'$ in categories (i--iii)
and
let $\GfeatherS,\GlnonfeatherS,\GrnonfeatherS,\GfS,$ and $\GlS$ be the subsequences of $\GS$
in categories (i--v), none of which are necessarily 2-sparse. 
Define $\GfeatherS^*,\GlnonfeatherS^*,\GrnonfeatherS^*,\GfS^*,$ and $\GlS^*$ to be their maximal length 
2-sparse subsequences, and define $\LS^*$ to be the maximal length 2-sparse subsequence of $\LS$.

Lemma~\ref{lem:Sharir}(\ref{item:Sharir2}) and the arguments from Recurrence~\ref{rec:DS-double-feather}
imply that
\begin{align}
|\GfeatherS^*| + |\GlnonfeatherS^*| + |\GrnonfeatherS^*| + |\GfS^*| + |\GlS^*|
 &\le  \DS{3}(\varphi,2\varphi-1) + 2\cdot \DS{5}(\Gn,2n-1) + 4\cdot \DS{4}(\Gn,2\Gn-1)\label{eqn:sixconstituents}\\
\istrut{7}\mbox{where }\; \varphi &= \DblFeather(\Gn,2n-1)
\nonumber
\end{align}

The sequence $\GfeatherS^*$ is obtained by substituting for each block in $\GfeatherS'$ a 2-sparse,
order-3 DS sequence.  Since $|\GfeatherS^*| \le \varphi$, by the superadditivity of $\DS{3}$ 
we have $|\GfeatherS^*| \le \DS{3}(\varphi)$, which is at most $\DS{3}(\varphi,2\varphi-1)$ by Lemma~\ref{lem:Sharir}(\ref{item:Sharir2}).
By the same reasoning, $|\GfS^*|$ and $|\GlS^*|$ 
are each at most $\DS{4}(\Gn) \le 2\cdot \DS{4}(\Gn,2\Gn-1)$ 
and $|\GlnonfeatherS^*| + |\GrnonfeatherS^*| \le 
\DS{2}(|\GlnonfeatherS'|+ |\GrnonfeatherS'|) < \DS{2}(|\GS'|) < 2\cdot \DS{5}(\Gn,2n-1)$.
We can also conclude that $|\LS^*| \le \DS{3}(n-\Gn) \le \DS{3}(n-\Gn,2(n-\Gn)-1)$.

In bounding various sequences above, the second argument of $\DS{s}$ and $\DblFeather$
is never more than $2\cdot \varphi$.
Choose $\iota$ to be minimal such that $2\cdot\varphi \le a_{\iota,3}^3$, so 
$j=3$ will be constant whenever we invoke Lemmas~\ref{lem:agreeable} and \ref{lem:doublefeather-bounds}
with $s\le 5$ and $c=3$.
It is straightforward to show that $\iota = \alpha(n) + O(1)$.

Observe that $S$ can be constructed by shuffling its six non-2-sparse constituent subsequences
$\GfeatherS$, $\GlnonfeatherS$, $\GrnonfeatherS$, $\GfS$, $\GlS$, $\LS$ in some fashion 
that restores 2-sparseness.  In other words, there is a 1-1 map between positions in $S$
and positions in its six constituents, and a surjective map $\psi$ from positions in $S$ to 
positions in its 2-sparse constituents 
$\GfeatherS^*$, $\GlnonfeatherS^*$, $\GrnonfeatherS^*$, $\GfS^*$, $\GlS^*$, $\LS^*$.
Partition $S$ into intervals $T_1T_2\cdots T_{\ceil{|S|/h}}$,
each with length $h = \ceil{\f{\DS{5}(6)+1}{2}} = O(1)$.
The image of $\psi$ on two consecutive intervals 
$T_{p-1}$ and $T_{p}$ (where $p<\ceil{|S|/h}$) cannot be
identical, for otherwise $T_{p-1} T_p$ would be a $2$-sparse, order-$5$ DS sequence with length 
$2h > \DS{5}(6)$ over a 6-letter alphabet, a contradiction.
Therefore, 
\begin{align*}
|S| &\le \zero{h\cdot (|\LS^*|+|\GfeatherS^*| + |\GlnonfeatherS^*| + |\GrnonfeatherS^*| + |\GfS^*| + |\GlS^*|)}\\
&= \zero{h \cdot n\cdot O(\mu_{3,\iota} + \mu_{3,\iota}\nu_{\iota}'' + \mu_{5,\iota} + \mu_{4,\iota})}\\
&= O(n\iota 2^\iota)  & \mbox{\{Since $\mu_{3,\iota} = O(\iota),\; \nu_{\iota}'' = O(\iota^3),\; \mu_{4,\iota}=O(2^{\iota})$, and 
$\mu_{5,\iota} = O(\iota2^{\iota})$\}}\\
&= O(n\alpha(n)2^{\alpha(n)}).
\end{align*}

\end{proof}

\section{Discussion and Open Problems}\label{sect:discussion}

Davenport-Schinzel sequences have been applied almost exclusively to problems in
combinatorial and computational geometry, with only a smattering of applications in 
other areas.  For example, see~\cite{SalvoP07,AlstrupLST97,Pettie-Deque-08,Pettie10a}.
One explanation for this, which is undoubtedly true, is that there is a natural fit
between geometric objects and their characterizations in terms of forbidden substructures.\footnote{E.g.,
in general position two lines do not share two points, three spheres do not share three points, degree-$d$ polynomials
do not have $d+1$ zeros, and so on.}
An equally compelling explanation, in our opinion, is that DS sequences are simply underpublicized,
and that the broader algorithms 
community is not used to analyzing algorithms and data structures with
forbidden substructure arguments.  We are optimistic that with increased awareness
of DS sequences
and their generalizations 
(e.g., forbidden 0-1 matrices)
the {\em forbidden substructure method}~\cite{Pettie10a} will become a standard tool in every algorithms researcher's toolbox.

Our bounds on Davenport-Schinzel sequences are sharp
for every order $s$, leaving little room for improvement.\footnote{That is, they cannot be expressed
more tightly using a generic inverse-Ackermann function $\alpha(n)$.  See Remark~\ref{remark:Ackermann-invariance}.}   However, there are many open problems 
on the geometric realizability of DS sequences and on various generalizations of DS sequences.
The most significant realizability result is due to Wiernik and Sharir~\cite{WS88}, who proved that
the lower envelope of $n$ line segments (that is, $n$ linear functions, each defined over a different interval)
has complexity $\Theta(\DS{3}(n)) = \Theta(n\alpha(n))$.  It is an open question whether this result can be generalized 
to degree-$s$ polynomials or polynomial segments.
In particular, it may be that the lower envelope of any set of $n$ degree-$s$ polynomials has complexity $O(n)$,
where $s$ only influences the leading constant.
Although our results do not address problems
of geometric realizability,
we suspect that modeling lower envelopes by derivation trees (rather than just sequences)
will open up a new line of attack on these fundamental realizability problems.

There are several challenging open problems in the realm of {\em generalized} Davenport-Schinzel sequences,
the foremost one being to characterize the set 
of all linear forbidden subsequences: those $\sigma$ for which $\Ex(\sigma,n)=O(n)$~\cite{Klazar02,Pettie-GenDS11}.
Linear forbidden subsequences and minimally nonlinear ones were exhibited by
Adamec, Klazar, and Valtr~\cite{AKV92}, Klazar and Valtr~\cite{KV94}, and 
Pettie~\cite{Pettie-DS-nonlin11,Pettie-GenDS11,Pettie-FH11,Pettie-SoCG11}.
It is also an open problem to characterize minimally non-linear 
forbidden 0-1 matrices~\cite{FurediH92}.  Though far from being solved,
there has been significant 
progress on this problem in the last decade~\cite{MarcusT04,Tardos05,Fulek09,Geneson09,Pettie-FH11,Pettie-GenDS11}.


%
%
%
%
%
%
%
%
%
%
\appendix

\section{Proof of Lemma~\ref{lem:Sharir}}\label{appendix:Sharir}

Recall the four parts of Lemma~\ref{lem:Sharir}.\\

\noindent{\bf Restatement of Lemma~\ref{lem:Sharir}\ } 
{\em
Let $\gamma_s(n) : \mathbb{N}\rightarrow\mathbb{N}$ be a non-decreasing function such that 
$\DS{s}(n) \le \gamma_s(n)\cdot n$.
\begin{enumerate}
\item (Trivial) For $s\ge 1$,  $\DS{s}(n,m) \le m-1 + \DS{s}(n)$.
\item (Sharir~\cite{Sharir87}) For $s\ge 3$, $\DS{s}(n) \le \gamma_{s-2}(n)\cdot \DS{s}(n,2n-1)$.
(This generalizes Hart and Sharir's proof~\cite{HS86} for $s=3$.)
\item (Sharir~\cite{Sharir87}) For $s\ge 2$, $\DS{s}(n) \le \gamma_{s-1}(n)\cdot \DS{s}(n,n)$.
\item (New) For $s\ge 3$, $\DS{s}(n) = \gamma_{s-2}(\gamma_s(n))\cdot \DS{s}(n,3n-1)$.
\end{enumerate}
}

\begin{proof}
Removing at most $m-1$ repeated symbols at block boundaries makes any sequence 2-sparse,
which implies Part (\ref{item:Sharir1}).

For Parts (\ref{item:Sharir2}) and (\ref{item:Sharir3}), consider the following method for greedily partitioning a 2-sparse, order-$s$ DS sequence
$S$ with $\|S\|=n$.
Write $S$ as $S_1S_2\cdots S_m$, where $S_1$ is the longest order-$(s-2)$ prefix of $S$,
$S_2$ is the longest order-$(s-2)$ prefix of the remainder of the sequence, and so on.
Each $S_q$ contains the first or last occurrence of some symbol, which implies $m\le 2n-1$ since $S_1$ must contain the first
occurrence of at least two symbols.  To see this, consider the symbol $b$ which caused the termination of $S_q$, 
that is, $S_q$ has order $s-2$ but $S_qb$ contains an alternating subsequence $\sigma_s = aba\cdots ab$ or $ba\cdots ab$ with length $s$;
whether it starts with $a$ depends on the parity of $s$.
If $S_q$ contained neither the first nor last occurrence of both $a$ and $b$, $S$ would contain an alternating subsequence 
$\sigma_{s+2}$ of length $s+2$, a contradiction.  
Obtain $S'$ from $S$ replacing each $S_q$ with a block containing exactly one occurrence of each symbol in $\Sigma(S_q)$.
Thus, 
\begin{align*}
|S| = \sum_{q=1}^m |S_q|
	&\le \sum_{q=1}^m \gamma_{s-2}(\|S_q\|)\cdot \|S_q\|	& \mbox{\{$S_q$ has order $s-2$, defn.~of $\gamma_{s-2}$\}}\\
	&\le \gamma_{s-2}(n)\cdot \sum_{q=1}^m \|S_q\|		& \mbox{\{$\gamma_{s-2}$ is non-decreasing\}}\\
	&= \gamma_{s-2}(n) \cdot |S'| \le \gamma_{s-2}(n)\cdot \DS{s}(n,m)	& \mbox{\{$S'\subseq S$ has order $s$\}}
\end{align*}
which proves Part~(\ref{item:Sharir2}).  Part~(\ref{item:Sharir3}) is proved in the same way except that we partition
$S$ into order-$(s-1)$ DS sequences.  In this case each $S_q$ must contain the last occurrence of some symbol, so $m\le n$.
We turn now to Part~(\ref{item:Sharir4}).

Partition $S$ into order-$(s-2)$ sequences $S_1S_2\cdots S_m$ as follows.  After $S_1,\cdots,S_{q-1}$
have been selected, let $S_{q}$ be the longest prefix of the remaining sequence that (i) has order $s-2$
and (ii) has length at most $\gamma_s(n)$.  The number of such sequences that were terminated due to (i) is at most
$2n-1$, by the same argument from Part~(\ref{item:Sharir2}).  The number terminated due to (ii) is at most $n$
since $|S| \le \gamma_s(n)\cdot n$, so $m \le 3n-1$.
Obtain an $m$-block sequence $S'$ in the usual way, by replacing each $S_q$ with a block containing its alphabet.
Thus,
\begin{align*}
|S| = \sum_{q=1}^m |S_q|
	&\le \sum_{q=1}^m \gamma_{s-2}(\|S_q\|)\cdot \|S_q\|	& \mbox{\{$S_q$ has order $s-2$, defn.~of $\gamma_{s-2}$\}}\\
	&\le \gamma_{s-2}(\gamma_s(n)) \cdot \sum_{q=1}^m \|S_q\|		& \mbox{\{$\gamma_{s}$ is non-decreasing, 
	$\|S_q\|\le |S_q|\le \gamma_s(n)$\}}\\
	&= \gamma_{s-2}(\gamma_s(n)) \cdot \zero{|S'| \le \gamma_{s-2}(\gamma_s(n))\cdot \DS{s}(n,m)}	& \mbox{\{$S'\subseq S$ has order $s$\}}
\end{align*}
\end{proof}

Note that while Part~(\ref{item:Sharir4}) is stronger than Part~(\ref{item:Sharir2}), it requires an upper bound on
$\gamma_s(n)$ to be applied, which is obtained by invoking Part~(\ref{item:Sharir2}).
In the end it does not matter precisely what $\gamma_s(n)$ is.
Once $\gamma_s(n)$ is known to be some primitive recursive
function of $\alpha(n)$, it follows that $\gamma_{s-2}(\gamma_s(n)) = \gamma_{s-2}(\alpha(n)) + O(1)$.

\section{Proof of Lemma~\ref{lem:agreeable}}\label{appendix:agreeable}

Recall our definition of Ackermann's function: $a_{1,j} = 2^j$, $a_{i,1} = 2$, 
and $a_{i,j} = w\cdot a_{i-1,w}$ where $w=a_{i,j-1}$.
Our task in this section is to prove the omnibus 
Lemma~\ref{lem:agreeable} in several stages.\\

\noindent {\bf Restatement of Lemma~\ref{lem:agreeable}}\ 
{\em 
Let $s\ge 1$ be the order parameter, $c\ge s-2$ be a constant,
and $i\ge 1$ be an arbitrary integer.
The following upper bounds on
$\DS{s}$ and $\Feather{s}$
hold for all $s\ge 1$ and all odd $s\ge 5$, respectively.
Define $j$ to be maximum such that $m\le a_{i,j}^c$.
\begin{align*}
\DS{1}(n,m) &= n + m-1						& \mbox{$s=1$}\\
\DS{2}(n,m) &= 2n + m-2						& \mbox{$s=2$}\\
\DS{3}(n,m) &\le (2i+2)n + (3i-2)cj(m-1)			& \mbox{$s=3$}\\
\DS{s}(n,m) &\le \M{s}{i}\paren{n + (cj)^{s-2}(m-1)}		& \mbox{all $s\ge 4$}	\\ 
\Feather{s}(n,m) &\le \N{s}{i}\paren{n + (cj)^{s-2}(m-1)}	& \mbox{odd $s\ge 5$}
\intertext{The values $\{\M{s}{i},\N{s}{i}\}$ are defined as follows, where $t=\floor{\frac{s-2}{2}}$.}
\istrut{5}\M{s}{i} 	&= \left\{
	\begin{array}{l}
		2^{{i+t+3}\choose t}	- 3(2(i+t+1))^t			\\
		\istrut{6}\fr{3}{2}(2(i+t+1))^{t+1}2^{{i+t+3}\choose t}
	\end{array}
	\right.
	&
	\begin{array}{r}
		\mbox{even $s\ge 4$\hcm[-.18]}\\
		\istrut{6}\mbox{odd $s\ge 5$\hcm[-.18]}
	\end{array}\\
\istrut{8}
\N{s}{i} &= 4\cdot 2^{{i+t+3}\choose t}		
	&
	\begin{array}{r}
		\mbox{odd $s\ge 5$\hcm[-.18]}
	\end{array}
\end{align*}
}

\paragraph{Overview.} 
The proof is by induction on $(s,i,j)$ with respect to any fixed $c\ge s-2$.
In Section~\ref{sect:agreeable:basecases} we confirm that Lemma~\ref{lem:agreeable}
holds when $i=1$.
In Section~\ref{sect:blockpartition} we discuss the role that Ackermann's function plays
in selecting block partitions for Recurrences~\ref{thm:recurrence-even},
\ref{lem:feather-rec}, and \ref{thm:recurrence-odd}.
In Section~\ref{sect:agreeable:order3} we confirm Lemma~\ref{lem:agreeable} at $s=3$.
In Section~\ref{sect:agreeable:lb} we identify sufficient lower bounds on
the elements of $\{\M{s}{i},\N{s}{i}\}_{s\ge 2,i\ge 1}$, then, in Section~\ref{sect:agreeable:happy},
prove that the particular ensemble $\{\M{s}{i},\N{s}{i}\}_{s\ge 2,i\ge 1}$ proposed
in Lemma~\ref{lem:agreeable} does, in fact, satisfy these lower bounds.

\subsection{Base Cases}\label{sect:agreeable:basecases}

\begin{lemma}\label{lem:i=1}
Let $n,m,$ and $s\ge 2$ be the alphabet size, block count, and order parameters.
Given $i\ge 1$, let $j'$ be minimum such that $m\le a_{i,j'}$.
Whether $i=1$ and $j'\ge 1$ or $j'=1$ and $i>1$, we have
\begin{align*}
\Feather{s}(n,m) &\;\le\; \DS{s}(n,m) \;\le\; 2^{s-1}n + j'^{s-2}(m-1).
\end{align*}
\end{lemma}

\begin{proof} 
First note that $\Feather{s}(n,m) \le \DS{s}(n,m) - 2n$ holds trivially since, in the worst case,
every occurrence in the sequence is a feather, except for the first and last occurrence of each letter.

At $s=2$ the claim follows directly from Lemma~\ref{lem:DS12}.
At $s\ge 3,j'=1$, the claim is trivial since there are only $a_{i,1}=2$ blocks and $\DS{s}(n,2)=2n$.

In the general case we have $s\ge 3$ and $j'>1$.
Let $S$ be an order-$s$, $m$-block sequence over an $n$-letter alphabet, where $m\le a_{1,j'}=2^{j'}$.  
Let $S=S_1S_2$ be the partition of $S$ using a uniform block partition with width $a_{1,j'-1} = 2^{j'-1}$,
so $\bl{S_1} = a_{1,j'-1}$ and $\bl{S_2} = m-a_{1,j'-1} \le a_{1,j'-1}$.
Note that $\GS'=\beta_1\beta_2$ consists of two blocks, where each $\beta_q$ is some permutation of the global alphabet $\GSigma$.
Since there are no middle occurrences in $\GS'$ or $S$
we can apply a simplified version of Recurrence~\ref{thm:recurrence-even}.
\begin{align*}
\lefteqn{\DS{s}(n,m)}\\
 &\leq \sum_{q=1,2} \DS{s}(\Ln_q, \bl{S_q}) \;+\; \DS{s-1}(\Gn, \bl{S_1}) \;+\; \DS{s-1}(\Gn, \bl{S_2}) & \mbox{\{local, first, and last\}}\\
		&\leq 2^{s-1}(n-\Gn) + (j'-1)^{s-2}(m-2) \;+\; 2(2^{s-2}\Gn)
		+ (j'-1)^{s-3}(m-2)     & \mbox{\{inductive~hypothesis\}}\\
		&< 2^{s-1}n \,+\, j'^{s-2}(m-1)\hcm[.5]			
\end{align*}
The last inequality follows from the fact that when $s\ge 3$,
$(j'-1)^{s-2} + (j'-1)^{s-3} \le j'^{s-2}$.
This concludes the induction. 
\end{proof}

If we introduce the `$c$' parameter and define $j$ to be minimum such that $m\le a_{i,j}^c$,
Lemma~\ref{lem:i=1} implies that $\DS{s}(n,m) \le 2^{s-1}n + (cj)^{s-2}(m-1)$ since $j' \le cj$.
Note that by definition of Ackermann's function, 
$a_{1,j}^c = (2^j)^c = a_{1,cj}$ and $a_{i,1}^c = 2^c = a_{1,c}$.

Lemma~\ref{lem:i=1} implies the claims of Lemma~\ref{lem:agreeable} at $i=1$.
When $s=3$, $2i+2 = 4 = 2^{s-1}$ and $3i-2 = 1$.
When $s\ge 4$ is even, $\M{s}{1} = 2^{t+4\choose t} - 3(2(t+2))^t \ge 2^{2t+1} = 2^{s-1}$.
When $s\ge 5$ is odd, $\M{s}{1} = \fr{3}{2}(2(t+2))^{t+1}2^{{t+4}\choose t} \ge 2^{2t+2} = 2^{s-1}$
and $\N{s}{1} = 4\cdot 2^{t + 4\choose t} \ge 2^{2t+2} = 2^{s-1}$.
The bounds above also imply that Lemma~\ref{lem:agreeable} holds at $j=1$ and $i>1$
since $a_{i,1}^c = a_{1,1}^c$ and both $\M{s}{i}$ and $\N{s}{i}$ are increasing in $i$.

\subsection{Block Partitions and Inductive Hypotheses}\label{sect:blockpartition}

When analyzing order-$s$ DS sequences we express the block count $m$
and partition size $\Gm$ in terms of constant powers of Ackermann's function $\{a_{i,j}^c\}$,
where the constant $c\ge s-2$ is fixed.  Recall that once $i$ is selected, $j$ is minimal such that $m \le a_{i,j}^c$. 
The base cases $i=1$ and $j=1$ have been handled so 
we can assume both are at least 2.  Let $w=a_{i,j-1}$.  

We always choose a uniform block partition $\{m_q\}_{1\le q\le \Gm}$ 
with width $w^c$, that is, $m_q = w^c$ for all $q<\Gm = \ceil{m/w^c}$ and the leftover $m_{\Gm}$ may be smaller.
When invoking the inductive hypothesis (Lemma~\ref{lem:agreeable}) 
on the $\Gm$-block sequence $\GS'$
we use parameter $i-1$.  In all other invocations of the inductive hypothesis we use parameter $i$.
When applied to any $m_q$-block sequences the `$j$' parameter is decremented since 
$m_q \le w^c = a_{i,j-1}^c$.
When applied to a $\Gm$-block sequence the `$j$' parameter is $w$ since
\[
\Gm = \ceil{\frac{m}{w^c}} \le \paren{\frac{a_{i,j}}{w}}^c = a_{i-1,w}^c.
\]
Furthermore, in such an invocation the dependence on $\Gm$ will always be at most linear in $m$
since $(cw)^{s-2}(\Gm-1) \le (cw)^{s-2}(\ceil{\f{m}{w^c}}-1) \le c^{s-2}(m-1)$.  This is the reason we require
the lower bound $c\ge s-2$.\\

If one is more familiar with the slowly growing row-inverses of Ackermann's function, it may be
helpful to remember that $cj = \log m - O(1)$ when $i=1$
and that $j = \log^{[i-1]}(m) - O(1)$ when $i>1$, the effect of the $c$ parameter being 
negligible since $a_{i,j}$ and $a_{i,j}^c$ are essentially identical relative to any sufficiently slowly growing function.\footnote{Recall that $\log^{[i-1]}(m)$ is short for $\log^{\star\cdots\star}(m)$ with $i-1$ $\star$s.}
Thus, the bounds of Lemma~\ref{lem:agreeable} could be rephrased
as $\DS{s}(n,m) \le \M{s}{i}\paren{n + O\paren{m(\log^{[i-1]}(m))^{s-2}}}$.
Since $\M{s}{i}$ is increasing in $i$, the best bounds are obtained by choosing $i$ to be minimal
such that $\log^{[i-1]}(m) = n/m + O(1)$.

\subsection{Order $s=3$}\label{sect:agreeable:order3}

\begin{lemma} {\bf (Order $s=3$)}
Let $n$ and $m$ be the alphabet size and block count of an order-3 DS sequence $S$.
For any $i, c\ge 1$, define $j$ to be minimum such that $m \le a_{i,j}^c$.
Then $\DS{3}$ is bounded by
\[
\DS{3}(n,m) \le (2i+2)n + (3i-2)cj(m-1)
\]
\end{lemma}

\begin{proof}
The base cases $i=1$ and $j=1$ have been handled already.
Let $i,j>1$ and $w=a_{i,j-1}$.
We invoke Recurrence~\ref{thm:recurrence-even}
with the uniform block partition $\{m_q\}_{1\le q\le \Gm}$, where $\Gm=\ceil{m/w^c}$.  (See Section~\ref{sect:blockpartition}.)
\begin{align*}
\DS{3}(n,m) &\leq \zero{\sum_{q=1}^{\Gm} \DS{3}(\Ln_q,m_q) \;+\; 2\cdot \DS{2}(\Gn,m) \;+\; \DS{1}(\DS{3}(\Gn,\Gm) - 2\Gn, m)}\\
		&\leq \zero{(2i+2)(n-\Gn)	\,+\, (3i-2)c(j-1)(m-\Gm)}	& \mbox{\{ind.~hyp.: local symbols\}}\\
			&\hcm[.5] \;+\; 4\Gn 			\,+\, 2(m-1) \hcm[5] & \mbox{\{global first and last occurrences\}}\\
			&\hcm[.5] \;+\; \zero{(2i-2)\Gn \;+\; (3(i-1)-2)cw(\Gm-1) \;+\; (m-1)} & \mbox{\{global middle occurrences\}}\\
		&\leq \zero{(2i+2)n \;+\; (3i-2)cj(m-1)}\\
		&\hcm[.5] \;+\; \zero{\SqBrack{- (2i+2) + 4 + (2i - 2)}\Gn \;+\; \SqBrack{- c(3i-2) + (3i-5) + 3}(m-1)}\\
		&\leq \zero{(2i+2)n \;+\; (3i-2)cj(m-1)}
\end{align*}
The last inequality holds since $c\ge s-2=1$.  
\end{proof}

At $s=2$ and $s=3$ the terms involving $n$ and $m$ have different leading constants, namely 
$2$ and $1$ when $s=2$
and $2i+2$ and $3i-2$ when $s=3$.
To provide some uniformity in the analyses below 
we will use the inequalities 
$\DS{2}(n,m) \le \M{2}{i}(n+m-1)$
and
$\DS{3}(n,m) \le \M{3}{i}(n+(cj)(m-1))$ 
when
invoking the inductive hypothesis at $i\ge 2$ and $s\in\{2,3\}$, where 
$\M{2}{i}=2$ and $\M{3}{i}=3i$ by definition.
Note that when $i\ge 2$, $\M{3}{i} = 3i \ge \max\{2i+2,3i-2\}$.

\subsection{Lower Bounds on $\M{s}{i}$ and $\N{s}{i}$}\label{sect:agreeable:lb}

Call an ensemble of values $\{\M{s'}{i'}, \N{s'}{i'}\}_{(s',i') \le (s,i)}$
{\em happy} if $\DS{s'}(n,m) \le \M{s'}{i'}(n+(cj)^{s'-2}(m-1))$ 
and $\Feather{s'}(n,m) \le \N{s'}{i'}(n+ (cj)^{s'-2}(m-1))$,
where $c$ and $j$ are defined as usual.
(In the subscript `$\le$' represents lexicographic ordering on tuples.)
In Lemma~\ref{lem:agreeable:lb} 
we determine lower bounds on $\M{s}{i}$ and $\N{s}{i}$ in a happy ensemble.
In Section~\ref{sect:agreeable:happy}
we prove that the {\em specific} ensemble proposed in Lemma~\ref{lem:agreeable} is, in fact, happy.

\begin{lemma}\label{lem:agreeable:lb}
Let $s\ge 4$ and $i\ge 2$.  Define $n,m,c,$ and $j$ as usual.
If $\{\M{s'}{i'},\N{s'}{i'}\}_{(s',i') \le (s,i-1)}$ is happy then $\{\M{s'}{i'},\N{s'}{i'}\}_{(s',i') \le (s,i)}$ is as well,
so long as
\begin{align*}
\M{s}{i} &\;\ge\; 2\M{s-1}{i} \;+\; \M{s-2}{i}\M{s}{i-1}						& \mbox{even $s$}\\
\M{s}{i} &\;\ge\; 2\M{s-1}{i} \;+\; \M{s-2}{i}\N{s}{i-1} \;+\; \M{s-3}{i}\M{s}{i-1}		& \mbox{odd $s$}\\
\N{s}{i} &\;\ge\; \N{s}{i-1} \;+\; 2\M{s-1}{i}								& \mbox{odd $s$}
\end{align*}
\end{lemma}

\begin{proof}
When $s\ge 4$ is even, Recurrence~\ref{thm:recurrence-even} implies that
\begin{align}
\DS{s}(n,m) &\leq \sum_{q=1}^{\Gm} \DS{s}(\Ln_q,m_q) \;+\; 2\cdot \DS{s-1}(\Gn,m) \;+\; \DS{s-2}(\DS{s}(\Gn,\Gm), m)\nonumber\\
		&\leq \M{s}{i}\paren{(n-\Gn) 			\;+\; (c(j-1))^{s-2}(m-\Gm)} \hcm[2] \mbox{\{happiness of the ensemble\}}\nonumber\\
		&\hcm[.5] +\; 2\M{s-1}{i}\paren{\Gn 			\;+\; (cj)^{s-3}(m-1)}\nonumber
			\\&\hcm[.5] +\; \M{s-2}{i}\bigg(\M{s}{i-1}\Big(\Gn	 \;+\; (cw)^{s-2}(\Gm-1)\Big) \;+\; (cj)^{s-4}(m-1)\bigg)\nonumber\\
		&\leq \M{s}{i}\paren{n \;+\; (cj)^{s-2}(m-1)}\nonumber\\
			&\hcm[.5] +\; \SqBrack{-\M{s}{i} + 2\M{s-1}{i} + \M{s-2}{i}\M{s}{i-1}}\cdot \Gn\label{eqn:M-even}\\
			 &\hcm[.5] +\; \SqBrack{- \M{s}{i}c^{s-2}j^{s-3}  + 2\M{s-1}{i}(cj)^{s-3} + \M{s-2}{i}\M{s}{i-1}c^{s-2} + \M{s-2}{i}(cj)^{s-4}}\cdot (m-1)\label{eqn:K-even}\\
		&\leq \M{s}{i}\paren{n \;+\; (cj)^{s-2}(m-1)}\label{eqn:s-even}
\end{align}
Inequality~(\ref{eqn:s-even}) will be satisfied whenever (\ref{eqn:M-even}) and (\ref{eqn:K-even}) are non-positive, that is, when
\begin{align}
\M{s}{i} &\;\ge\; 2\M{s-1}{i} \;+\; \M{s-2}{i}\M{s}{i-1}								\label{eqn:M-rec-even}\\
\M{s}{i} &\;\ge\; \f{2\M{s-1}{i}}{s-2} \;+\; \f{\M{s-2}{i}\M{s}{i-1}}{2^{s-3}} \;+\; \f{\M{s-2}{i}}{2(s-2)^2}	\label{eqn:K-rec-even}
\end{align}
Inequality~(\ref{eqn:K-rec-even}) was obtained by dividing (\ref{eqn:K-even}) through by $c^{s-2}j^{s-3}$ and noting that $c\ge s-2\ge 2$ and $j\ge 2$.
Note that Inequality~(\ref{eqn:K-rec-even}) is weaker than Inequality~(\ref{eqn:M-rec-even})
since $\M{s}{i} > \M{s-1}{i} > \M{s-2}{i}$, so it suffices to consider only the former.

When $s\ge 5$ is odd, Recurrence~\ref{thm:recurrence-odd} implies that
\begin{align}
\DS{s}(n,m) &\leq \sum_{q=1}^{\Gm} \DS{s}(\Ln_q,m_q) + 2\cdot \DS{s-1}(\Gn,m) + \DS{s-2}(\Feather{s}(\Gn,\Gm),m) + \DS{s-3}(\DS{s}(\Gn,\Gm),m)\nonumber\\
		&\leq \M{s}{i}\paren{(n-\Gn) 		\;+\; (c(j-1))^{s-2}(m-\Gm)}	 \hcm[2] \mbox{\{happiness of the ensemble\}}\nonumber\\					
			&\hcm[.5] +\; 2\M{s-1}{i}\paren{\Gn 		\;+\; (cj)^{s-3}(m-1)}\nonumber\\							
			&\hcm[.5] +\; \M{s-2}{i}\bigg(\N{s}{i-1}\paren{\Gn \;+\; (cw)^{s-2}(\Gm-1)}       \;+\; (cj)^{s-4}(m-1)\bigg)\nonumber\\	
			&\hcm[.5] +\; \M{s-3}{i}\bigg(\M{s}{i-1}\paren{\Gn \;+\; (cw)^{s-2}(\Gm-1)}   \;+\; (cj)^{s-5}(m-1)\bigg)\nonumber\\
		&\leq \M{s}{i}\paren{n \;+\; (cj)^{s-2}(m-1)}\nonumber\\
			&\hcm[.5] +\; \SqBrack{-\M{s}{i} \;+\; 2\M{s-1}{i} \;+\; \M{s-2}{i}\N{s}{i-1} \;+\; \M{s-3}{i}\M{s}{i-1}}\cdot \Gn\label{eqn:M-odd}\\
			&\hcm[.5] +\; \SqBrack{-\M{s}{i}c^{s-2}j^{s-3} \;+\; 2\M{s-1}{i}(cj)^{s-3} \;+\; \M{s-2}{i}\N{s}{i-1}c^{s-2}\nonumber\\
			&\hcm[1.3] +\; \M{s-2}{i}(cj)^{s-4} \;+\; \M{s-3}{i}\M{s}{i-1}c^{s-2} \;+\; \M{s-3}{i}(cj)^{s-5}}\cdot (m-1)	\label{eqn:K-odd}\\
		&\leq \M{s}{i}\paren{n \;+\; (cj)^{s-2}(m-1)}\label{eqn:s-odd}
\end{align}
Inequality~(\ref{eqn:s-odd}) will be satisfied whenever (\ref{eqn:M-odd}) and (\ref{eqn:K-odd}) are non-positive, that is, when
\begin{align}
\M{s}{i} &\;\ge\; 2\M{s-1}{i} \;+\; \M{s-2}{i}\N{s}{i-1} \;+\; \M{s-3}{i}\M{s}{i-1}		\label{eqn:M-rec-odd}\\
\M{s}{i} &\;\ge\; \f{2\M{s-1}{i}}{s-2} \;+\; \f{\M{s-2}{i}\N{s}{i-1}}{2^{s-3}} \;+\; \f{\M{s-2}{i}}{2(s-2)^2} \;+\; \f{\M{s-3}{i}\M{s}{i-1}}{2^{s-3}} 
\;+\; \f{\M{s-3}{i}}{4(s-2)^3}		\label{eqn:K-rec-odd}
\end{align}
The denominators of Inequality~(\ref{eqn:K-rec-odd}) follow by dividing (\ref{eqn:K-odd}) 
through by $c^{s-2}j^{s-3}$ and noting that $c\ge s-2\ge 3$ and $j\ge 2$.  
Inequality~(\ref{eqn:K-rec-odd}) is weaker than Inequality~(\ref{eqn:M-rec-odd}) 
since $\M{s-1}{i} > \M{s-2}{i} > \M{s-3}{i}$ so it suffices to consider only Inequality~(\ref{eqn:M-rec-odd}).

Using similar calculations, one derives from Recurrence~\ref{lem:feather-rec} 
the claimed lower bound on $\N{s}{i}$.
\begin{align}
\N{s}{i}	&\;\ge\; \N{s}{i-1} \;+\; 2\M{s-1}{i}\label{eqn:N-rec}
\end{align}
\end{proof}

\subsection{The Happiness of the Ensemble}\label{sect:agreeable:happy}

From this point on we argue the happiness of the {\em specific} ensemble $\{\M{s}{i},\N{s}{i}\}$ stated in Lemma~\ref{lem:agreeable}.
We can say an individual value $\M{s}{i}$ or $\N{s}{i}$ is happy if it satisfies the appropriate lower bound inequality, either~(\ref{eqn:K-rec-even}), (\ref{eqn:K-rec-odd}),
or (\ref{eqn:N-rec}).  

\begin{lemma}\label{lem:agreeable:happy}
The ensemble~$\{\M{s}{i}, \N{s}{i}\}$ defined in Lemma~\ref{lem:agreeable} is happy.
\end{lemma}

\begin{proof}
All $\N{s}{i}$ are happy since
\begin{align*}
\N{s}{i-1} + 2\M{s-1}{i} 
&< 4\cdot 2^{{i+t+2}\choose t} + 2\cdot 2^{{i+t+3}\choose t} & \mbox{\{by definition\}}\\
&\le 4 \cdot 2^{{i+t+3}\choose t} \;=\; \N{s}{i}	& \mbox{\{$t\ge 1,\; 4\cdot 2^{i+t+2\choose t} \le 2\cdot 2^{i+t+3\choose t}$\}}\\
\end{align*}

When $s=4$ and $t=\floor{\frac{s-2}{2}} = 1$ the expression for $\M{4}{i}$ simplifies to $2^{i+4} - 6(i+2)$.
The happiness of $\M{4}{i}$ follows easily, as seen below.
\begin{align*}
2\M{3}{i} + \M{2}{i}\M{4}{i-1} &= 2(3i) \;+\; 2\cdot \left(2^{i+3} - 6(i+1)\right)	& \mbox{\{by definition\}}\\				
		&= 2^{i+4} \;+\; 6i \,-\, 12(i+1)\\
		&= 2^{i+4} \,-\, 6(i+2) = \M{4}{i} \\
\end{align*}

When $s=5$ and $t=\floor{\frac{s-2}{2}}=1$, the expression for $\M{5}{i}$ simplifies to $\fr{3}{2} (2(i+2))^2 2^{i+4}$,
which lets us quickly certify the happiness of $\M{5}{i}$.
\begin{align*}
   2\M{4}{i} + \M{3}{i}\N{5}{i-1} + \M{2}{i}\M{5}{i-1}
&= 2\paren{2^{i+4} - 6(i+2)}
	\;+\; 3i \cdot 4\cdot 2^{i+3}
	\;+\; 2\cdot \fr{3}{2}(2(i+1))^2 2^{i+3}\\
&\le \Paren{2 \;+\; 6i \;+\; \fr{3}{2}(2(i+1))^2} 2^{i+4}\\
&\le \fr{3}{2} (2(i+2))^2 2^{i+4} \;=\; \M{5}{i}
\end{align*}

We now turn to the happiness of $\M{s}{i}$ for even $s\ge 6$.  Note that when we invoke the definition of $\M{s-1}{i}$ and $\M{s-2}{i}$
their ``$t$'' parameter is $t-1 = \floor{\fr{(s-1)-2}{2}} = \floor{\fr{(s-2)-2}{2}}$.
\begin{align*}
\lefteqn{2\M{s-1}{i} \,+\, \M{s-2}{i}\M{s}{i-1}}\\
		&\le 2\cdot \left[\fr{3}{2}(2(i+t))^t2^{{i+t+2}\choose t-1}\right] \,+\, \left[2^{{i+t+2}\choose t-1} - 3(2(i+t))^{t-1}\right]\cdot \left[2^{{i+t+2}\choose t} - 3(2(i+t))^t\right]\\
		&= \left[3(2(i+t))^t2^{{i+t+2}\choose t-1}\right]
		\,+\, 2^{{i+t+2}\choose t-1}  \left[2^{{i+t+2}\choose t} - 3(2(i+t))^t\right]
		\,-\, 3(2(i+t))^{t-1} \left[2^{{i+t+2}\choose t} - 3(2(i+t))^t\right]\\
		&= 2^{{i+t+2}\choose t-1}  2^{{i+t+2}\choose t} \,-\, 3(2(i+t))^{t-1}  \left[2^{{i+t+2}\choose t} - 3(2(i+t))^t\right]\\
		&\le 2^{{i+t+3}\choose t} - 3(2(i+t+1))^t		 \, =\, \M{s}{i}		
\end{align*}
In other words, $\M{s}{i}$ satisfies Inequality~(\ref{eqn:M-rec-even}) when $s\ge 6$ is even.
It also satisfies Inequality~(\ref{eqn:M-rec-odd}) at odd $s\ge 7$, which can be seen as follows.
Note that the ``$t$'' parameter for $s-1$ it $t$, whereas it is $t-1$ for $s-2$ and $s-3$.
\begin{align*}
\lefteqn{2\M{s-1}{i} \:+\: \M{s-2}{i}\N{s}{i-1} \:+\: \M{s-3}{i}\M{s}{i-1} }\\
		&\le 2\cdot 2^{{i+t+3}\choose t} \:+\: \fr{3}{2}(2(i+t))^t2^{{i+t+2}\choose t-1} \cdot 4\cdot 2^{{i+t+2}\choose t}
			\:+\: 2^{{i+t+2}\choose t-1} \cdot \fr{3}{2}(2(i+t))^{t+1}2^{{i+t+2}\choose t}\\
		&\le \SqBrack{2 \:+\: \fr{3}{2}4\cdot (2(i+t))^t \:+\: \fr{3}{2}(2(i+t))^{t+1}} \cdot 2^{{i+t+3}\choose t}\\
		&\le \SqBrack{2\:+\: \fr{3}{2}(2(i+t))^t\cdot 2(i+t+2)} \cdot 2^{{i+t+3}\choose t}\\
		&\le \fr{3}{2}(2(i+t+1))^{t+1}\cdot 2^{{i+t+3}\choose t} = \M{s}{i}
\end{align*}
We have shown that $\{\M{s}{i}\}$ and $\{\N{s}{i}\}$ are happy over the full range of parameters.
This concludes the proof of Lemma~\ref{lem:agreeable}.
\end{proof}

\section{Proof of Lemma~\ref{lem:doublefeather-bounds}}\label{sect:appendix:order5}

The proof of Lemma~\ref{lem:doublefeather-bounds} closely mimics that of
Lemma~\ref{lem:agreeable}.\\

\noindent{\bf Restatement of Lemma~\ref{lem:doublefeather-bounds}\ }
{\em 
Let $n$ and $m$ be the alphabet size and block count.
After a parameter $i\ge 1$ is chosen let $j\ge 1$ be minimum such that $m \le a_{i,j}^c$,
where $c=3$ is fixed.
We have the following upper
bounds on $\FFeather(n,m), \DblFeather(n,m),$ and $\DS{5}(n,m)$.

\begin{align*}
\FFeather(n,m) &\le \nu_i'\paren{n + (cj)^2(m-1)}	& \mbox{where $\nu_i' = 3{i+1\choose 2} + 3$}\\
\DblFeather(n,m) &\le \nu_i''\paren{n + (cj)^3(m-1)}  & \mbox{where $\nu_i'' = 6{i+2\choose 3} + 8i$}\\
\DS{5}(n,m) &\le \mu_{5,i}\paren{n + (cj)^3(m-1)}  & \mbox{where $\mu_{5,i} = i2^{i+7}$}\\
\end{align*}
}

\begin{proof}
We use the following upper bounds on order-$3$ and order-$4$ DS sequences
from Lemma~\ref{lem:agreeable}.
\begin{align*}
\DS{3}(n,m) &\le \zero{\mu_{3,i}[n + (cj)(m-1)]}\hcm[2.5]		& \mbox{where $\mu_{3,i} = 3i+1 \ge \max\{2i+2,3i-2\}$}\\
\DS{4}(n,m) &\le \zero{\mu_{4,i}[n + (cj)^2(m-1)]}	& \mbox{where $\mu_{4,i} = 2^{i+4} - 6(i+2)$}
\intertext{and, when $i=1$,}
\DS{4}(n,m) &\le \zero{2^3 n + (cj)^2(m-1)}		\\
\DS{5}(n,m) &\le \zero{2^4 n + (cj)^3(m-1)}			& \mbox{See Lemma~\ref{lem:i=1}.}
\end{align*}

\paragraph{Base Cases.} In the worst case every occurrence in an order-4 sequence
is a dove (or hawk) feather, except for the first and last occurrence of each symbol,
which are wingtips.  This implies that 
\begin{align*}
\FFeather(n,m) &\le \DS{4}(n,m) - 2n\\
			&\le (2^3 - 2)n + (3j)^2(m-1)	& \{\mbox{by Lemma~\ref{lem:i=1}}\}\\
			&\le \nu_1'\paren{n + (3j)^2(m-1)}	& \{\mbox{since $\nu_1' = 6$}\}
\end{align*}
The same argument implies that
$\DblFeather(n,m) \le \DS{5}(n,m) - 2n \le (2^4-2)n+(3j)^3(m-1)$,
which is at most $\nu_1''(n + (3j)^3(m-1))$ since $\nu_1'' = 14$.
This confirms the claimed bounds when $i=1$.  
It also holds when $i>1$ and $j=1$ since $a_{i,1}^3 = a_{1,1}^3$
and all of $\nu_i',\nu_i'',$ and $\mu_{5,i}$ are increasing in $i$.

\paragraph{Inductive Cases.}
We can assume that $i,j>1$.
As in the proof of Lemma~\ref{lem:agreeable} we always
apply Recurrence~\ref{rec:double-feather} with 
a uniform block partition $\{m_q\}_{1\le q\le \Gm}$ 
with width $w^c$, where $w=a_{i,j-1}$, $\Gm=\ceil{m/w^c}$,
and $c=3$ is fixed.  

According to Recurrence~\ref{rec:double-feather} and the inductive hypothesis
we have:
\begin{align}
\FFeather(n,m) &\le \sum_{q=1}^{\Gm} \FFeather(\Ln_q,m_q) \, + \, \FFeather(\Gn,\Gm)
						\, + \, \DS{3}(\Gn,m) - \Gn\nonumber\\
	&\le \nu_i'\SqBrack{(n-\Gn) + (c(j-1))^2(m-\Gm)}\nonumber\\
	&\hcm +\; \nu_{i-1}'\SqBrack{\Gn + (cw)^2(\Gm-1)}\nonumber\\
	&\hcm +\; \mu_{3,i}\SqBrack{\Gn + (cj)(m-1)} \;-\; \Gn & \mbox{\{inductive hypothesis\}}\nonumber\\
	&\le \nu_i'\sqbrack{n + (cj)^2(m-1)}\nonumber\\
	&\hcm +\; \SqBrack{-\nu_i' + \nu_{i-1}' + \mu_{3,i} - 1}\cdot \Gn \nonumber\\
	& \hcm +\; \SqBrack{- (c^2j)\nu_i'  + (c^2/w)\nu_{i-1}' + (cj)\mu_{3,i}}\cdot (m-1)\label{ineq:nu_i'-1}\\
	&\le \nu_i'\sqbrack{n + (cj)^2(m-1)}\label{ineq:nu_i'-2}\\
\intertext{Inequality (\ref{ineq:nu_i'-1}) follows from the fact
that $(c(j-1))^2 \le (cj)^2 - cj$ and that $(cw)^2(\Gm-1) = (cw)^2(\ceil{\f{m}{w^c}}-1) \le (c^2/w) (m-1)$.
Inequality (\ref{ineq:nu_i'-2}) will follow so long as $\nu_i'$ satisfies the following.
}
\nu_i'  &\;\ge\; \nu_{i-1}' + \mu_{3,i} - 1 \;=\; \nu_{i-1}' + 3i\label{ineq:nu'}
\end{align}
One may confirm that $\nu_i' = 3{i+1 \choose 2} + 3$ satisfies (\ref{ineq:nu'}).
In a similar fashion we can obtain a lower bound on $\nu_i''$ as follows.

\begin{align}
\DblFeather(n,m) &\le \sum_{q=1}^{\Gm} \DblFeather(\Ln_q,m_q) \, + \, \DblFeather(\Gn,\Gm)
						\, + \, 2(\FFeather(\Gn,m) + \Gn)\nonumber\\
			&\le \nu_i''\SqBrack{(n-\Gn) + (c(j-1))^3(m-\Gm)}\nonumber\\
			&\hcm +\; \nu_{i-1}''\SqBrack{\Gn + (cw)^3(\Gm-1)}\nonumber\\
			&\hcm +\; 2\nu_i'\SqBrack{\Gn + (cj)^2(m-1)}
				+ 2\Gn	& \mbox{\{inductive hypothesis\}}\nonumber\\
			&\le \nu_i''\sqbrack{n + (cj)^3(m-1)}\nonumber\\
			&\hcm +\; \SqBrack{-\nu_i'' + \nu_{i-1}'' + 2\nu_i' + 2}\cdot \Gn\nonumber\\
			&\hcm +\; \SqBrack{-(c^3j^2) \nu_i'' + c^3\nu_{i-1}'' + (cj)^2\cdot 2\nu_i'}\cdot (m-1)\label{ineq:nu_i''-1}\\
			&\le \nu_i''\sqbrack{n + (cj)^3(m-1)} \label{ineq:nu_i''-2}
\intertext{Inequality (\ref{ineq:nu_i''-1}) follows since $(c(j-1))^3 \le (cj)^3 - (cj)^2$
and $(cw)^3(\Gm-1) \le c^3(m-1)$.
Inequality (\ref{ineq:nu_i''-2}) will follow from (\ref{ineq:nu_i''-1}) if $\nu_i''$ satisfies}
\nu_i'' &\;\ge\; \nu_{i-1}'' + 2\nu_i' + 2 \;=\; \nu_{i-1}'' + 6{i+1\choose 2} + 8\label{ineq:nu_i''}
\end{align}
Again, one may confirm that $\nu_i'' = 6{i+2\choose 3} + 8i$ satisfies (\ref{ineq:nu_i''}).
We are now ready to calculate a lower bound constraint on $\mu_{5,i}$.
By Recurrence~\ref{rec:DS-double-feather} and the inductive hypothesis
we have

\begin{align}
\DS{5}(n,m) &\le \zero{\sum_{q=1}^{\Gm} \DS{5}(\Ln_q,m_q) 
			\,+\, 2\cdot \DS{4}(\Gn,m) \,+\, \DS{3}(\DblFeather(\Gn,\Gm),m) \,+\, 
			\DS{2}(\DS{5}(\Gn,\Gm),2m-1)}\nonumber\\
		&\le \zero{\mu_{5,i}\SqBrack{(n-\Gn) + (c(j-1))^3(m-\Gm)}
		+ 2\mu_{4,i}\SqBrack{\Gn + (cj)^2(m-1)}} \nonumber\\
		&\hcm \zero{+\; \mu_{3,i}\SqBrack{\nu_{i-1}''\SqBrack{\Gn + (cw)^3(\Gm-1)} + (cj)(m-1)}}\nonumber\\
		&\hcm \zero{+\; 2\mu_{5,i-1}\SqBrack{\Gn + (cw)^3(\Gm-1)} + 2(m-1)} & \mbox{\{inductive hypothesis\}}\nonumber\\
		&\le \zero{\mu_{5,i}[n + (cj)^3(m-1)]}\nonumber\\
		&\hcm \zero{+\; \SqBrack{-\mu_{5,i} + 2\mu_{4,i} + \mu_{3,i}\nu_{i-1}'' + 2\mu_{5,i-1}}\cdot\Gn}\nonumber\\
		&\hcm +\; \SqBrack{- (c^3j^2)\mu_{5,i}  + 2(cj)^2 \mu_{4,i} 
					+ c^3 \mu_{3,i}\nu_{i-1}'' \nonumber\\
					&\hcm[1.5] + (cj)\mu_{3,i} + 2(c^3)\mu_{5,i-1} + 2}\cdot(m-1)\nonumber\\
		&\le \mu_{5,i}[n + (cj)^3(m-1)]\label{ineq:mu5}
\intertext{Inequality (\ref{ineq:mu5}) will hold so long as $\mu_{5,i}$ satisfies}
\mu_{5,i} &\ge 2\mu_{5,i-1} + 2\mu_{4,i} + \mu_{3,i}\nu_{i-1}''\nonumber\\
		&= 2\mu_{5,i-1} + \sqbrack{2\paren{2^{i+4} - 6(i+2)} + \paren{3i+1}\paren{6{i+1\choose 3} + 8(i-1)}}\nonumber\\
		&= 2\mu_{5,i-1} + \SqBrack{2^{i+5} + \paren{3i+1}\paren{i+1}\paren{i}\paren{i-1} - 4(i+8)}\label{ineq:mu5-req}
\end{align}

The bracketed term is less than $2^{i+7}$, so it suffices to show that
$\mu_{5,i} \ge 2\mu_{5,i-1} + 2^{i+7}$, when $i>1$.  
One may confirm that $\mu_{5,i} = i\cdot 2^{i+7}$ satisfies (\ref{ineq:mu5-req}).

\end{proof}

\ignore{
\FFeather_{i}(n,m)  &\le 6n \, + \, (3j)^2(m-1)  \\
\DblFeather_{i}(n,m) &\le 14n \, + \, (3j)^3(m-1)   & \mbox{\rb{4}{when $i=1$ or $j=1$}}\\
&\\
\FFeather_{i}(n,m) &\le \sum_{q=1}^{\Gm} \FFeather_{i}(\Ln_q,m_q) \, + \, \FFeather_{i-1}(\Gn,\Gm)
						\, + \, \DS{3}(\Gn,m) - \Gn\\
\DblFeather_{i}(n,m) &\le \sum_{q=1}^{\Gm} \DblFeather_{i}(\Ln_q,m_q) \, + \, \DblFeather_{i-1}(\Gn,\Gm)
						\, + \, 2(\FFeather_{i}(\Gn,m) + \Gn)		& \mbox{\rb{7}{when $i,j>1$}}

\DS{5}(n,m) &\le \sum_{q=1}^{\Gm} \DS{5}(\Ln_q,m_q) 
			\,+\, 2\cdot \DS{4}(\Gn,m) \,+\, \DS{3}(\DblFeather_{i-1}(\Gn,\Gm),m) \,+\, 
			\DS{2}(\DS{5}(\Gn,\Gm),2m-1)\\

}

\end{document}